\RequirePackage{fix-cm}
\documentclass[smallextended]{svjour3}       
\smartqed  
\usepackage{graphicx}
\usepackage{amsmath}
\usepackage{amssymb}
\usepackage{amsmath}
\usepackage{color}
\newcommand{\N}{{\mathbb{N}}}
\newcommand{\R}{{\mathbb{R}}}

\begin{document}

\title{Emergent Properties in a V1-Inspired Network of Hodgkin-Huxley Neurons
}


\author{M. Maama         \and
        B. Ambrosio \and M.A. Aziz-Alaoui \and S.~M. Mintchev
}


\institute{M. Maama, B. Ambrosio, M.A. Aziz-Alaoui \at
                  Normandie Univ, UNIHAVRE, LMAH, FR-CNRS-3335, ISCN, 76600 Le Havre, France\\
                  S.~M. Mintchev \at
                  Department of Mathematics, The Cooper Union,
 New York, NY 10003, USA\\
\email{maama.mohamed@gmail.com, benjamin.ambrosio@univ-lehavre.fr,\\ aziz.alaoui@univ-lehavre.fr, mintchev@cooper.edu}\\                         
}

\date{Received: date / Accepted: date}

\maketitle

\begin{abstract}
This article is devoted to the theoretical and numerical analysis of a network of excitatory and inhibitory neurons of Hodgkin-Huxley (HH) type, for which the topology is inspired by that of a single local layer of visual cortex V1. Our model is related to the recent CSY model and therefore differs from other classical models in the field. It combines a driven stochastic drive -- which may be interpreted as an ambient drive for each neuron -- with recurrent inputs resulting from the network activity. After a review of the dynamics of a single HH equation for both the deterministic and the stochastically driven case, we proceed to an analysis of the network. This analysis reveals emergent properties of the system such as partial synchronization and synchronization (defined here as a state of the network for which all the neurons spike within a short interval of time), correlation between excitatory and inhibitory conductances, and oscillations in the gamma-band frequency. The collective behavior enumerated herein is observed when the input-amplitude parameter $S^{EE}$ measuring excitatory-to-excitatory coupling (recurrent excitation) increases to within a certain range. Of note, our work indicates a distinct mechanism for obtaining the emergent properties, some of which have been classically observed. As a consequence our article contributes to the understanding of how assemblies of inhibitory and excitatory cells interact together to produce rhythms in the network. It also aims to bring problems from neuroscience to the realm of mathematics, where they can be analyzed rigorously.  
\end{abstract}

\section{Introduction}

\subsection{Background and motivation}
This paper deals with an analysis of the dynamics of a network of Hodgkin-Huxley (HH) ordinary differential equations (ODEs). 
We briefly recall that the HH equations were introduced in 1952, see \cite{HH}, in order to describe the formation and propagation of action potentials along the giant squid nerve axon.
They have served as a basis for numerous other mathematical models in neuroscience, e.g., the FitzgHugh-Nagumo (FHN) \cite{Fit-1961,Nagumo}, the Moris-Lecar \cite{ML-1981}, and the Hindmarsh-Rose model \cite{HR}, to cite only a few. For a general textbook on mathematical models related to the HH model, we refer to \cite{BookBor2017,BookCronin,BookErm2010,BookIzh2006}. All of the above-mentioned models (used to describe the electrical activity of a single neuron) have subsequently been implemented in various studies of neural ensembles, by way of networks of ODE's; related meanfield models have also been introduced~\cite{BookBor2017,Bru-2000,Cha-2010,BookDay2001,BookErm2010,BookGer2014,BookKan2013,Wil-1972}.

From a physiological point of view, the modeling of an assembly of neurons in the brain with networks of ODEs may address issues such as how neurons in primary sensory areas produce their responses, how cortical maps process information, etc... 
Network models of ODEs can be used to address one of the typical questions prevalent in the physiological literature: To what extent are the responses of neuronal ensembles (and the subsequent perception of information) shaped by feed-forward vs. recurrent circuitry? This question is a current challenge and a subject of debate in neuroscience; a typical example is provided by the neurons in layer IV of the primary visual cortex, which receive sparse afferent inputs from the lateral geniculate nucleus (LGN) concurrent with abundant inputs from other cortical neurons, see \cite{BookKan2013}.

Networks of ODEs have attracted a lot of attention from the applied mathematics and physics communities. We refer the reader to \cite{BookBarabasi,BookBarrat2008,BookNewman2010networks} for various standard concepts developed generally for the theory of networks with real-world topologies. Topics of interest regarding dynamical networks classically include synchrnonization, the existence of an attractor, pattern formation, bifurcation phenomena, wave propagations, etc... Our own contributions to this subject reside both in the context of networks of ODEs as well as in Reaction-Diffusion PDEs inspired by Neuroscience \cite{Ambrosio2,Ambrosio3,Ambrosio8,Ambrosio7}. The present paper represents a special thread in this line of work, as it is devoted specifically to studying emergent phenomena such as synchronization and rhythms in a random network with V1-inspired topology.

Admittedly, synchronization has been studied extensively in various works dedicated to the theoretical analysis of network dynamics. Alternative definitions of this emergent phenomenon indicate various approaches to studying it: The most mathematically tractable version of this dynamical mode is {\it complete synchronization}, wherein all nodes of the network evolve identically. A particularly striking physical example of complete synchronization arises in networks of metronomes: random initial conditions notwithstanding, when a few of these devices are placed on a light wooden board that is set atop a freely moving chassis, their oscillations will synchronize asymptotically. Biology abounds with further examples of synchronization. For instance, we mention the complete synchronization of the flashes of fireflies. One can also think about synchronization of heart pacemaker cells, or even the synchronization of applause in an assembly. Further  examples are given in \cite{Ala-2006,OKe2017,Pan2002,Str1993} as well as the references cited therein.

Other types of synchronization have been introduced to better describe different natural phenomena. One such of specific interest in Neuroscience is the so-called {\it cluster synchronization}, observed when different subgroups of the total neural population synchronize their activity differently, see for example \cite{Belykh3,be-mw_09,Pec2014}. 

In the present paper we focus further on identifying numerically a region of parameters for which the network of ODEs exhibits dynamics relevant to mathematical neuroscience. Accordingly, our setting features an ensemble of HH systems for which the coupling relies on the following assumptions:
\begin{itemize}
\item the feed-forward inputs are represented by a stochastic drive;
\item the recurrent inputs are incorporated by way of excitatory and inhibitory synaptic coupling terms. 
\end{itemize}
As discussed above, the basic idea of leveraging the interplay between feed-forward and recurrent connection structures is critical to the relevance of such models to neuroscience applications; among other information coming from real physical data, this has been used successfully in a series of recent papers aiming to describe the electrical activity of the visual cortex V1, see \cite{Char-2014,Char-2016,Char-2018}. The latter three articles cited here have influenced our approach significantly and deserve a brief overview, since various modeling aspects and assumptions therein play an important role in our analysis of the network.

In \cite{Char-2014}, see also \cite{Ran2013} (and further \cite{Sun-2009,Zho-2013} and references therein cited for related complementary modeling of V1), the authors considered a network featuring a mixture of a few hundred Integrate-and-Fire neurons of excitatory (E) and inhibitory (I) type, connected via a sparse network topology inspired by real data from V1. This numerical study documents and analyzes emergent spiking behavior in local neuronal populations. Emphasis is given to the cluster synchronization phenomenon, also referred as {\it partial synchronization}, which gives a name to the tendency (of both large and small) random groups of neurons to coordinate spontaneously their spiking activity. The occurrence of a notable amount of spikes is referred as an {\it event}. One of the conclusions of this work is that driving the system with a relatively high feed-forward input imposes a certain regularity on its inter-event times, producing a rhythm consistent with broad-band gamma oscillations.

In \cite{Char-2016}, the authors developed a sophisticated sequel (termed CSY) to the mathematical model of \cite{Char-2014}, in order to build a more biologically realistic representation of a layer $4C\alpha$ of the macaque monkey's visual cortex. In particular, several new features should be noted: data-driven inputs from LGN, feedback inputs from layer 6, and specific equations for these neurons. Their model is presented as the first realistic model that successfully derives orientation selectivity within the macaque primary visual cortex despite incorporating the sparseness of magnocellular LGN inputs. To highlight further a major contribution of this work, we note that with the parameters of the model first tuned to the background regime, orientation selectivity arises by way of only a light variation of the inputs: the total energy contained in the inputs is very low in comparison to the energy coming from recurrent connections, yet these slight changes in the inputs have the outsize effect of inducing orientation selectivity.

Finally in \cite{Char-2018}, the authors further analyzed the CSY model    in order to study mechanisms for cortical gamma-band activity in the cerebral cortex and thereby identify neurobiological factors that influence such activity. At the end of this article the authors emphasized the role played by the so called REI (recurrent excitation inhibition) mechanism in the clustering phenomenon. The authors argue that this mechanism is at play in any local population of E and I neurons that have many recurrent connections. The explanation is the following:  occurrence of grouped spikes in a population of spike trains is typically precipitated by the crossing of threshold by one or more E neurons. This may or may not lead to more substantial excitatory firing through recurrent excitation, but the excitation that is thereby produced raises the membrane potentials of E and I neurons alike. In turn, this may lead possibly to the firing of more spikes and the generation of a partially synchronous event. Since I-cells are densely connected to both E and I-cells, firing in I-cells eventually brings the excitation down and furthermore leads to hyper-polarization of membrane potentials for a large fraction of the local population; this induces a period of quiescence. From there, the decay of inhibition in conjunction with the continued excitatory drive return the network to a state where, through stimulation or by chance, another clustering of spikes is initiated.

Of course, the numerical analysis of networks of E and I neurons is far from new. The above-mentioned three recent contributions notwithstanding, we note also the major  classical references from the field so as to place in context the contribution of the present paper. A well known and often cited framework that plays a role in synchronization and the emergence of $\gamma$-rhythms in such networks is the so-called PING mechanism (Pyramidal-Interneuronal Network Gamma). \cite{Bor-2005,Erm-1998,Whi-2000}; see also~\cite{Tra-1999-2}, as well as \cite{BookBor2017} for a broad overview. Of note, models such as the one described in \cite{Tra-1994} are interesting but not directly relevant to our work, since they also include modeling of axons (spatial dimension) while our simplified model does not focus on this aspect. 

At first glance the PING and REI mechanisms may seem to describe the same phenomenon (significant time-correlated rise and decrease of neuronal spiking activity -- by way of initial excitation followed by concomitant inhibition, in cyclical fashion).
However, in~\cite{Char-2014,Char-2016,Char-2018} the authors already present a compelling contrast between the two models. Namely, commensurate with the goal of reproducing realistic outputs in V1, the REI model targets milder phenomenology (a small part of neurons included in the events, and the subpopulations are random...) in favor of strong synchronization. In particular, the events discerned in REI do not result from alternation between excitatory and inhibitory phases, the latter dynamic being a significant theme in the PING model.

In certain parameter regimes, the REI mechanism is supported by the CSY mathematical model. The modeling choices associated with the latter that are relevant for our own model include the strategic choice of network topology -- featuring neither sparse, nor complete graph connection structure -- as well as specifics in the equations governing the synaptic connections themselves. At the same time, the model considered herein is separate and distinct from CSY, most notably because of the different neuronal dynamics governing individual cells -- the CSY model employs LIF for each cell, while ours features HH neurons. The HH model has a richer structure than IF models, but it is more difficult to track theoretically and numerically. For IF models the reset and resting delay are set manually whereas in HH they are intrinsic properties of the dynamics. There are also further distinctions to our work: we employ Dirac impulses to model both the drive to the network and the signals between interacting units -- a feature incorporated in~\cite{Char-2014} but subsequently modified with the introduction of CSY; moreover, our analysis leverages the network outputs in order to extract a single parameter ($S^{EE}$) that plays a crucial role in the onset of emergent behaviors. The ensuing parameter study indicates a path from the homogeneous stochastic state to a partial synchronization regime, global synchronization, the emergence of a gamma rhythm, and $g_E$ -- $g_I$ correlation. N.B.: in the global synchronized regime, we observe that I-population events completely encapsulate those by the E-population, i.e., the I-population spiking starts before and ends at the same time or shortly after the end of the $E$-spiking. This is a novel phenomenology that has to date not been observed in PING or REI. The work herein carries the potential to serve as a bridge to further rigorous mathematical analysis. Most results in this setting, including most of what we discuss in the present paper are of a computational nature. Nevertheless, to move further in the direction of rigor, our final discussion proposes three 2-dimensional models deduced from our work that can be studied mathematically.




With this background in mind, our goal is to proceed to the analysis of a model of a few hundred HH neurons driven individually by a stochastic input and connected with a topology inspired by V1 according to~\cite{Char-2014,Char-2016,Char-2018} but with the aforementioned distinctions.
Proceeding hierarchically, we review and revisit some known facts on the dynamical characteristics of a single HH ODE model in Sec.~\ref{Sec:single_HH}. We focus specifically on the analysis for a range of the parameter $I$ in which HH transitions from steady state to oscillatory behavior. This detail is critical to us going forward: since the current balance in individual cells resulting from inhibitory and excitatory inputs is achieved in this parameter regime, this balance replaces the parameter $I$ in the ensuing treatment.
In Sec.~\ref{Sec:single_HH-sto} we replace the parameter $I$ by a stochastic Poisson input corresponding to a feed-forward drive, and we analyze the switch between quiescent and excitatory activity as the drive intensity is increased.
Finally, we investigate the dynamics of the entire network in Sec.~\ref{Sec:network_analysis}; we illustrate therein some numerical simulations of the network and discuss the emergence of organized behavior subject to the variation of appropriate parameters.
We offer some concluding remarks and perspective regarding future work in Sec.~\ref{Sec:concl-persp}.
To set up an appropriate emphasis for our treatment, we first complete our introduction in Sec.~\ref{Subsec:general_network} by way of describing in greater detail the network of HH equations to be considered in Sec.~\ref{Sec:network_analysis}.

\subsection{The HH ODE network model}\label{Subsec:general_network}   
  The single unit of our network is the classical following HH  ODE system:
\begin{equation}\label{eq:HH-single}
\left\{
\begin{array}{rcl}
CV_t  &=&I+\overline{g}_{Na} m^3 h(E_{Na}-V)+\overline{g}_{K} n^4(E_{K}-V)+\overline{g}_{L}(E_{L}-V)\\

n_t  &=& \alpha_{n}(V)(1-n)-\beta_{n}(V)n\\
m_t  &=& \alpha_{m}(V)(1-m)-\beta_{m}(V)m\\
h_t  &=& \alpha_{h}(V)(1-h)-\beta_{h}(V)h\\
\end{array}
\right.
\end{equation}
with
\begin{equation}
\label{eq:param}
\begin{array}{rcl}
 \alpha_{n}(V)=0.01 \frac{-V-55}{\exp{(-5.5-0.1V)}-1}& & \beta_{n}(V) =0.125\exp(-(V+65)/80)\\
 \alpha_{m}(V)=0.1 \frac{-V-40}{\exp{(-4-0.1V)}-1}& & \beta_{m}(V) =4 \exp(-(V+65)/18)\\
\alpha_{h}(V)= 0.07\exp(-(V+65)/20) & & \beta_{h}(V) = \frac{1}{ 1+\exp( -0.1V-3.5) )}
\end{array}
\end{equation}
%
%
and
\[C=1 \; E_{K}=-77 \;~~E_{Na}=50, \ ~~E_{L}=-54.387 \]
\[\overline{g}_{K}=36 \;  \overline{g}_{Na}=120 \;\overline{g}_{L}=0.3 \]
In equation \eqref{eq:HH-single}, $V$ represents the membrane potential and $n,m,h$ are gating variables modeling the opening and the closing of ionic channels; $n$ is related to potassium fluxes, $m$ and $h$ with sodium fluxes. Hodgkin and Huxley (\cite{HH}) established these equations by a  series of voltage measurements thanks to the  voltage clamp technique, which allowed them to maintain constant the membrane potential. 
 Thanks to this approach, they were able to fit the functional parameters of their model of four ODEs with their experiments. Basically, the model is obtained by considering the cell as an 
 electrical circuit, and writing the Kirchoff law between internal and external currents. Then, the model assumes that the membrane acts as a capacitor. Ionic currents result from ionic channels acting
 as variable voltage dependent resistances. The model takes into account three ionic currents: potassium ($K^+$), sodium ($Na^+$) and leakage (mainly chlorure, $Cl^-$). 
In equation \ref{eq:HH-single}, subscript $t$ stands for the derivative $\frac{d}{dt}$, $I$ is the external membrane current, $C$ is the membrane capacitance, $ \overline{g}_{i},\  E_i, i\in\{K,Na,L\}$
are respectively the maximal conductances and the
 Nernst equilibrium potentials.  We refer to  \cite{BookBor2017,BookCronin,BookErm2010,BookIzh2006,HH} for textbooks presenting this very classical model.  The value of parameters is taken from \cite{BookDay2001,Sun-2009}. The final goal of this article is to study networks of HH equations, 
in which each neuron receives excitatory and inhibitory inputs from its presynaptic neurons. Hence, adapting \cite{Char-2014}, we consider a network of $N=500$ HH-neurons with additional entries for excitatory and inhibitory presynaptic inputs which leads to the following equation for the network:

\begin{equation}\label{eq:HH-Network}
\left\{
\begin{array}{rcl}
CV_{it}  &=&\overline{g}_{Na} m^3 h(E_{Na}-V_i)+\overline{g}_{K} n^4(E_{K}-V_i)+\overline{g}_{L}(E_{L}-V_i)\\
& &+g_{E}(E_{E}-V_i)+g_{I}(E_{I}-V_i), i \in\{1,...,N\}\\
n_{it}  &=& \alpha_{n}(V_i)(1-n_i)-\beta_{n}(V_i)n_i\\
m_{it} &=& \alpha_{m}(V_i)(1-m_i)-\beta_{m}(V_i)m_i\\
h_{it} &=& \alpha_{h}(V_i)(1-h_i)-\beta_{h}(V_i)h_i\\
\tau_E g_{Eit}  &=& -g_{Ei}+S^{dr}\sum_{s\in \mathcal{D}(i)}\delta(t-s)+S^{QE}\sum_{j\in \Gamma_E(i),s \in \mathcal{N}(j)}\delta(t-s)\\
\tau_I g_{Iit}  &=& -g_{Ii}+S^{QI}\sum_{j\in \Gamma_I(i),s \in \mathcal{N}(j)}\delta(t-s)\\
\end{array}
\right.
\end{equation}
where $Q$ is equal to $E$ for $E$-neurons and $I$ for $I$-neurons.\\
This means that for each node in the network, we add two equations to the original HH ODE system. Those are the equations which contain coupling terms inputs coming from:
\begin{enumerate}
\item the network (presynaptic E and I-neurons)
\item a stochastic input drive,  only for $g_E$.
\end{enumerate}
These two variables stand here for gating variables and are added as such to the first equation.  
Let us explicit some notations and characteristics of the network and also set up the values of fixed parameters.\\

\textbf{Network inputs}\\
 When the variable $V_j$ of a given neuron $j$ crosses a threshold $Tr$ upward,  a kick is generated in all its postsynaptic neurons. 
In equation \eqref{eq:HH-Network}, the kicks coming from $E$ neurons are represented mathematically by the Dirac term  $\delta(t-s)$ in the $g_{Ei}$  equation. Accordingly, in equation \eqref{eq:HH-Network}, $\Gamma_E(i)$ denotes the set of the presynaptic  E-neurons of the neuron $i$. The notation $\mathcal{N}(j)$ refers to the set of times at which the neuron  $j$  crosses the threshold $Tr$ upwards. Similar notations hold for the  $g_{Ii}$ equation.  Each $E$ neuron receives kicks with coupling strength  $S^{EE}$ from presynaptic $E$-neurons, and coupling strength $S^{EI}$ from presynaptic $I$-neurons. 
Each $I$ neuron receives kicks with coupling strength $S^{IE}$ from $E$-neurons,   and coupling strength  $S^{II}$ from $I$-neurons, see figure 1. We set also
\[E_E=0 \mbox{ and }E_I=-80\]
which makes kicks coming from presynaptic E-neurons to have a depolarizing effect on the membrane potential $V_i$ and kicks coming from presynaptic I-neurons to have hyperpolarizing effect on $V_i$.\\
The value of the threshold $Tr$ is set to $-10$. \\

\textbf{Poissonian Input} \\
Driving conductances with Poisonian inputs is found to be common in the literature. We refer for example to \cite{Char-2014,Char-2016,BookGer2014,Lin-2012}. Accordingly, a drive is added to the $g_{Ei}$ evolution equation as Poisson inputs of parameter $\rho_E=0.9$ for E-neurons and parameter $\rho_I=2.7$ for $I$ neurons.  Assuming that the time unit is ms, this implies that the mean value between two stochastic driven inputs is about $\frac{1}{0.9}$ ms for E-neurons and $\frac{1}{2.7}$ ms for I-neurons. Analog notations, as those used for the kicks coming from the network  are used for the stochastic input drive: $\mathcal{D}(i)$ refers to the set of times at which the neuron $i$ receives kicks corresponding to the input drive. Finally, note that the kicks have an amplitude given by the parameter $\frac{S^{dr}}{\tau_E}=0.04/2=0.02$. \\

\textbf{Network Topology}\\
We consider a network of $N=500$ neurons with $Ne=375$  E-neurons and $Ni=125$ I-neurons. The network is constructed as follows: for each neuron,
we pick-up randomly a fixed number of $E$ and $I$ presynaptic neurons. 
These fixed number only depend on the nature of the neuron. They are important fixed parameters of our model. More explicitly: for each $E$ neuron, we pickup randomly $N_{ee}=50$ presynaptic $E$ neurons, for each $E$ neuron, we pickup randomly $N_{ei}=25$ presynaptic $I$ neurons, for each $I$ neuron, we pickup randomly $N_{ie}=190$ presynaptic $E$ neurons, for each $I$ neuron, we pickup randomly $N_{ii}=25$ presynaptic $I$ neurons. So, the network is an inhomegeneous random oriented graph, where the probability of connexion depends on the type of neurons. In our network, $I$ neurons are highly connected with presynaptic $E$ neurons.  The network topology is inspired by the ratio found in \cite{Char-2016}.\\

\textbf{Value of parameters}\\
The dynamical effects of parameters $S^{II},S^{EE},S^{IE},S^{EI}$ and $S^{dr}$ will be discussed in the following sections. Table 1 summarizes the values of fixed parameters.
\begin{table}

\begin{tabular}{|c|c|c|c|c|c|c|c|}
\hline
$N=500$&$N_e=375$&$N_i=125$&$V_E=0$&$V_I=-80$&$\tau_E=2$ & $\tau_I=3$\\
\hline
$N_{ee}=50$&$N_{ei}=25$&$N_{ie}=190$&$N_{ii}=25$ &$S^{dr}=0.04$&$\rho_E=0.9$&$\rho_I=2.7$ \\
\hline
\end{tabular}
\label{table:1}
\caption{This table summarizes the value of fixed parameters used in numerical simulations of the network equation \eqref{eq:HH-Network}.}
\end{table}
\begin{figure}
\label{fig:1}
\begin{center}
\includegraphics[scale=0.2]{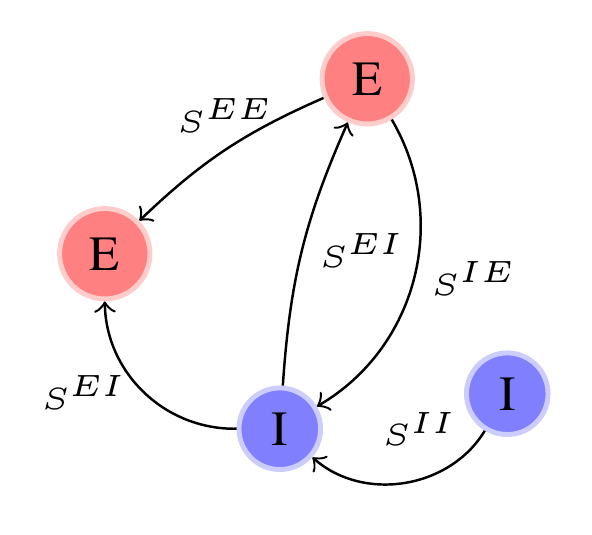}
\begin{quotation}
   \caption{Schematic representation of the coupling in the network.
   Each $E$ neuron receives kicks with coupling strength  $S^{EE}$ from $E$ neurons,  coupling strength   $S^{EI}$ from $I$ neurons.  
   Each $I$ neuron receives kicks with coupling strength  $S^{IE}$ from $E$ neurons,   coupling strength  $S^{II}$ from $I$ neurons.}
\end{quotation}
\end{center}
\end{figure}
\begin{figure}
\begin{center}
\includegraphics[height=3.5cm,width=7cm]{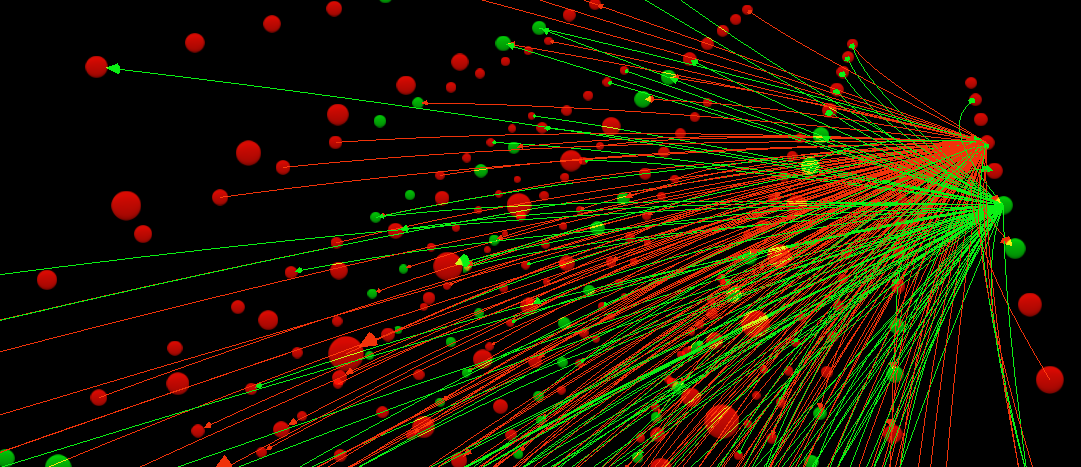}
   \caption{A three-dimensional representation shot of our network.    We  have represented the pre and postsynaptic connections of the E-neuron \#100 as well as the postsynaptic connections   of the I neuron \#400. Postsynaptic connections of E-neurons are represented in red while postsynaptic connections of I neurons are represented in green.}

\end{center}
\end{figure}

\section{Analysis of the HH equation}\label{Sec:single_HH}

\label{sec:1}
In this section, we recall and revisit some properties of solutions of equation \eqref{eq:HH-single} :
\begin{equation*}
\left\{
\begin{array}{rcl}
C\dfrac{dV}{dt}  &=&I+\overline{g}_{Na} m^3 h(E_{Na}-V)+\overline{g}_{K} n^4(E_{K}-V)+\overline{g}_{L}(E_{L}-V),\\

\dfrac{dn}{dt}  &=& \alpha_{n}(V)(1-n)-\beta_{n}(V)n,\\
\dfrac{dm}{dt}  &=& \alpha_{m}(V)(1-m)-\beta_{m}(V)m\\
\dfrac{dh}{dt}  &=& \alpha_{h}(V)(1-h)-\beta_{h}(V)h,\\
\end{array}
\right.
\end{equation*}
with the above values of parameters. A simple analysis shows that the following theorem holds.
\begin{theorem}
\label{th:posinv}
There exists $V_m,V_M \in \R$ such that the compact set:
\[\mathcal{K}=[V_m,V_M]\times [0,1]^3\]

is positively invariant for system \eqref{eq:HH-single}.
\end{theorem}
\begin{proof}
First note that $\alpha$'s and $\beta$'s functions are positive. This implies that $[0,1]^3$ is positively invariant for $n,m$ and $h$.  
Note then, that from the first equation, $V_t$ becomes negative for $V$ large enough. It becomes positive for $-V$ large enough. This implies the result. 
\end{proof}
The next proposition follows from simple computations.
\begin{proposition}
The stationary solutions of \eqref{eq:HH-single} satisfy:
\begin{equation}\label{eq:HH-single-sta}
\left\{
\begin{array}{rcl}
f(V) &=&0\\
n&=& \frac{\alpha_{n}(V)}{\alpha_n(V)+\beta_{n}(V)},\\
m&=& \frac{\alpha_{m}(V)}{\alpha_m(V)+\beta_{m}(V)},\\
n&=& \frac{\alpha_{m}(V)}{\alpha_m(V)+\beta_{m}(V)},\\
\end{array}
\right.
\end{equation}
with 
\[f(V)=I+\overline{g}_{Na} \big(g_m(V)\big)^3g_h(V)(E_{Na}-V)+\overline{g}_{K} \big((g_n(V)\big)^4(E_{K}-V)+\overline{g}_{L}(E_{L}-V)\]
and 
\[g_\kappa(V)=\frac{\alpha_{\kappa}(V)}{\alpha_\kappa(V)+\beta_{\kappa}(V)},\,\kappa \in\{n,m,h\}\]
Furthermore,
\[\lim_{V\rightarrow -\infty}f(V)=+\infty\,\,\lim_{V\rightarrow +\infty}f(V)=-\infty\]
\end{proposition}
Numerical simulations provide evidence that $f$ is decreasing. Figure 3 illustrates the case $I=0$, which gives a unique stationary solution with $V\simeq -65$. 

It is known, from numerical simulations \cite{Bal-2018} and references therein, that as $I$ increases trough an interval $[I_0-I_1]$ within in the region of interest here,  system \eqref{eq:HH-single} exhibits a cascade of bifurcations giving rise to unstable and stable limit cycles, with persistence of the stable stationary point.  At some  value $I_2$ the stationary point becomes finally unstable trough a subcritical Hopf bifurcation. Figure 4, taken from \cite{Bal-2018}  summarizes these facts.  

For the readers convenience, we have illustrated here some part of the phenomenon, with our set of parameters. In figure 5, we have drawn the bifurcation diagram for two initial conditions. Figure 5 shows the coexistence of stable limit-cycle and stable stationary solution for an interval containing $[6.4,7.5]$. 

Figure 6 illustrates this phenomenon with trajectories for $I=7$. 

Figures 7 and 8 illustrate the temporal evolution of the variables, and highlight a clear slow fast dynamic. 

Figure 9 illustrates phase space dynamics along  with nullcines.  These figures allow to provide both dynamical and biological interpretation. Over a period of time, we can observe distinct dynamic phases.  First, one can remark the proximity between $V$ and $m$.  We now go trough a more detailed description. We start at a time where $V$ is at its maximum, for example $t\simeq 18$. It corresponds roughly at a time at which $m$ (in green in figure 8) reaches also its maximum. A more attentive look at figure 9, bottom-left (BL), however, shows that $V$ reaches its maximum a little before $m$. Indeed, when $V$ is at its maximum, $m$ increases, $n$ increases, $h$ decreases (see figure 9). Those dynamics are fast. The variable $V$ decreases very fastly. During this fast decreasing phase of $V$, the dynamics of $n$ and $h$ change. At $t\simeq 20$, $V$ is at its minimum. So is $m$ (see figures 8 and 9 BL). Note the remarkable change in the dynamics for $m$ and $V$ at their minimum in figure $9$ BL. At this time, one can distinguish the beginning of a slow phase. The variables $V,m,h$ ($h $ appears in blue in figure 8) increase (so that the sodium conductance increase), while $n$ (in red in figure 8)  decreases (the  potassium conductance decrease). Dynamics are still slow. Around $t=27$, dynamics of $n$ and $h$ change again:  $n$ starts to decrease while $h$ starts to increase. Dynamics are still slow until $t\simeq 34$. At this time, we observe that $V$ enters a fast dynamical phase, and so $n,m,h$. We observe a drastic increase in $V$ ( which means that the sodium flux dominates, recall that the $E_{Na}=50$ ).Variables $m$ and $n$ increase while $h$ decrease. During this fast phase, $V$ reaches its maximum again, and a full loop is completed. It is worth noting that $V$ have a fast increase followed by a fast decrease, which means that $V_t$ crosses abruptly $0$ downwards.  Denoting $F(V,n,m,h)$ the first component of the right-hand side of \eqref{eq:HH-single}, this corresponds to  
\[DF(V_M,n_M,m_M,h_M). \begin{pmatrix}
0\\
\delta n\\
0\\
\delta h\\
\end{pmatrix}<<0\]
where $ (V_M,n_M,m_M,h_M)$ is the point for which $V$ is at its maximum and with $\delta n >0, \delta h<0$. The  potassium channel plays an important in repolarization. Also, 
\[DF(V_m,n_m,m_m,h_m   ). \begin{pmatrix}
0\\
\delta n\\
0\\
\delta h\\
\end{pmatrix}>>0\]
where $ (V_m,n_m,m_m,h_m)$ is the point  for which $V$ is at its mimimum, and with $\delta n <0, \delta h>0$. 
It is worth noting how the trajectory in the $(V,m) $ follows the manifold $n=\frac{\alpha(V)}{\alpha(V)+\beta(V)}$ (figure 9 BL).

\textbf{Biological interpretation} According to the model, variable $n$  regulate the flux of potassium, while variables $m$ and $h$ regulate the sodium flux. The spike in $V$ occurs when it is pushed trough the sodium gradient $(E_{Na}=50)$, and corresponds to the depolarization of the membrane. Repolarization and hyperpolarization, result from the potassium gradient ($E_K=-77$). According to the HH paper, high values for $m$ and $n$ correspond to the activation of respectively potassium and sodium gates. Low value of $h$ correspond to the inactivation of the sodium gate.

\begin{figure}
 \label{fig:f-sta}
\includegraphics[height=5cm,width=5cm]{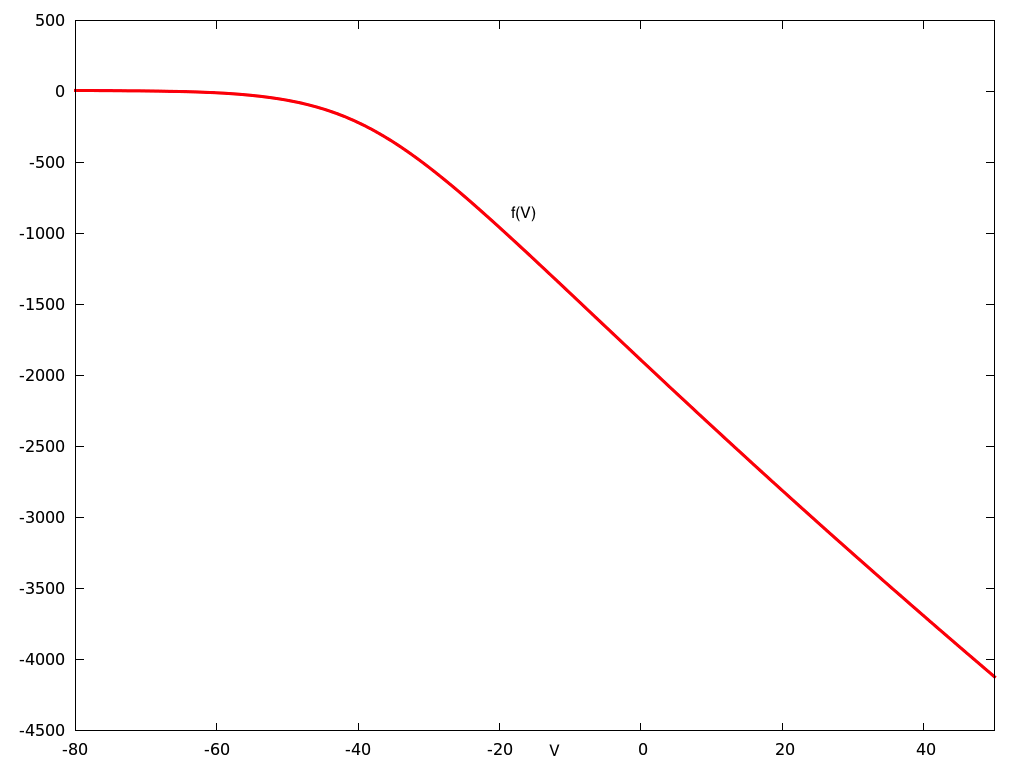}
\includegraphics[height=5cm,width=5cm]{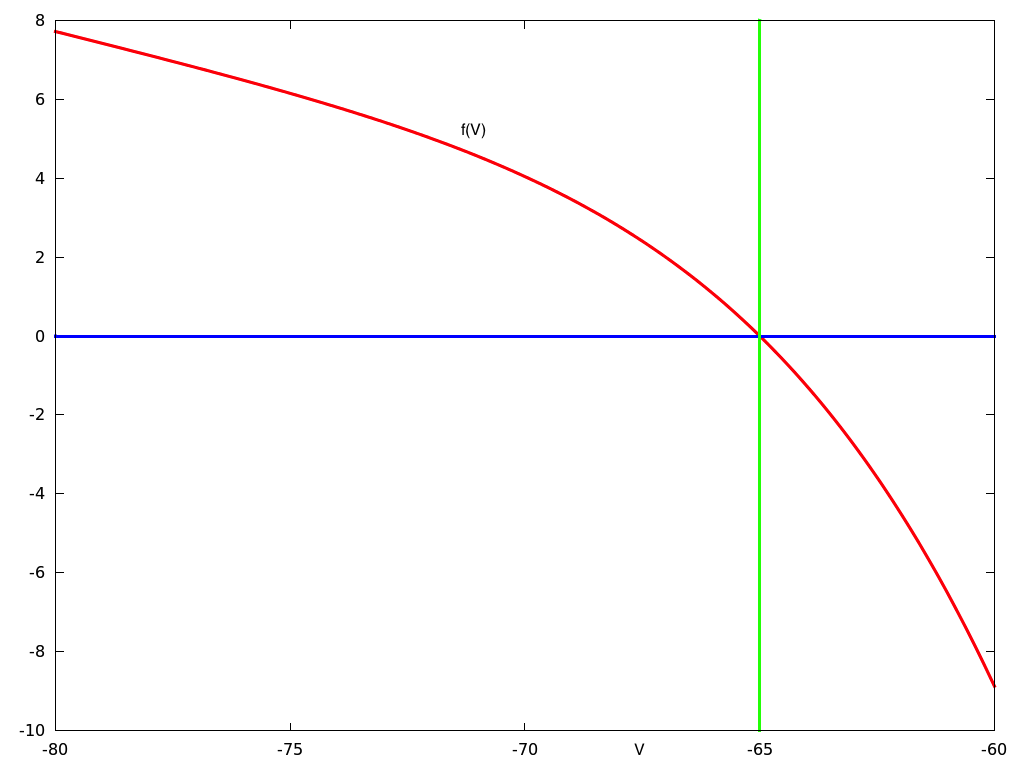}
\caption{Illustrations of $f(V)$ for $I=0$. Left: $V\in [-80,50]$, Right:$V\in [-80,-60]$. Equation $f(V)=0$ is satisfied for $V\simeq -65$ which corresponds to the stationary solution.}
\end{figure}
\begin{figure}
 \label{fig:AB2018}
 \includegraphics[height=5cm,width=5cm]{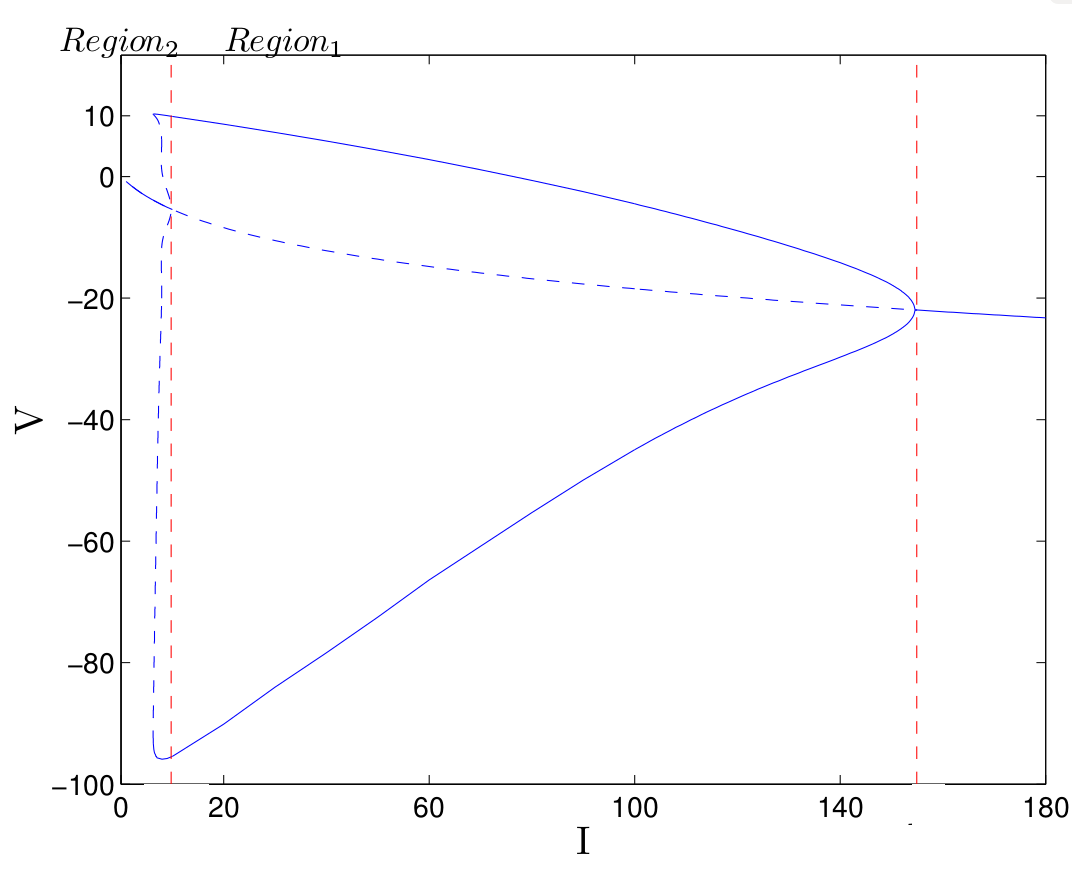}
\includegraphics[height=5cm,width=5cm]{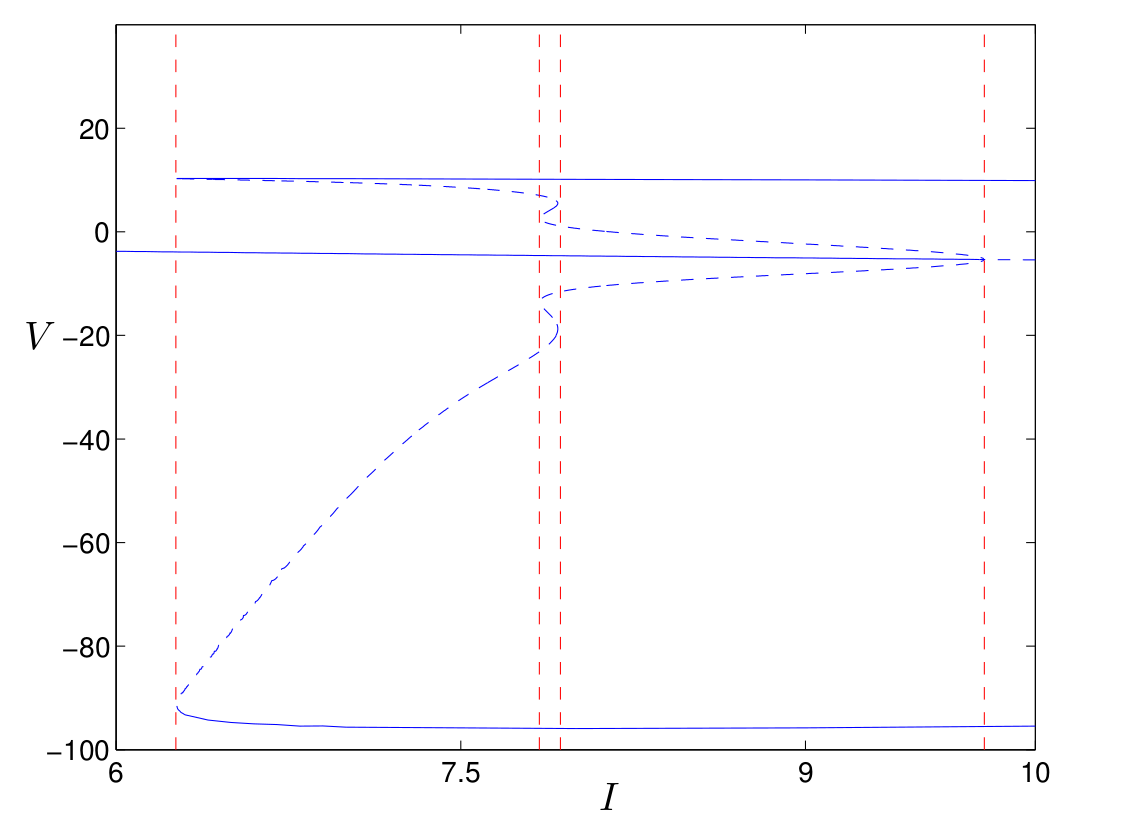}
\caption{Bifurcation Diagrams as the parameter I is varied, reproduced as in \cite{Bal-2018}, with permission of the authors. The figure at left indicates two bifurcating regions. One around I=10 and another around I=150.  The latter one corresponds to a super-critcial Hopf Bifurcation. The bifurcation around 10 is more complicated. A zoom is illustrated in the right figure. One can observe a sub-critical Hopf bifurcation around I=9.5, and bifurcation of cycles downwards. In particular, one can observe coexistence of a stationary stable point and a stable limit cycle for I approximately in the interval [6.3,9.5]  }
\end{figure}
\begin{figure}
 \label{fig:BifHH}
 \includegraphics[height=5cm,width=5cm]{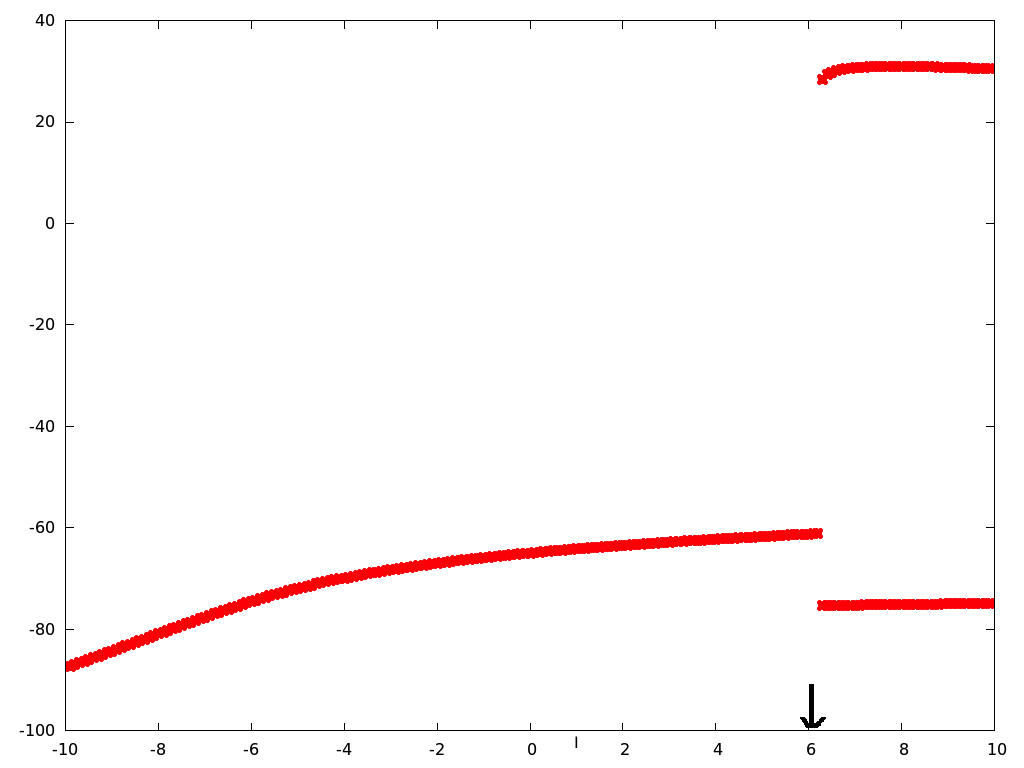}
\includegraphics[height=5cm,width=5cm]{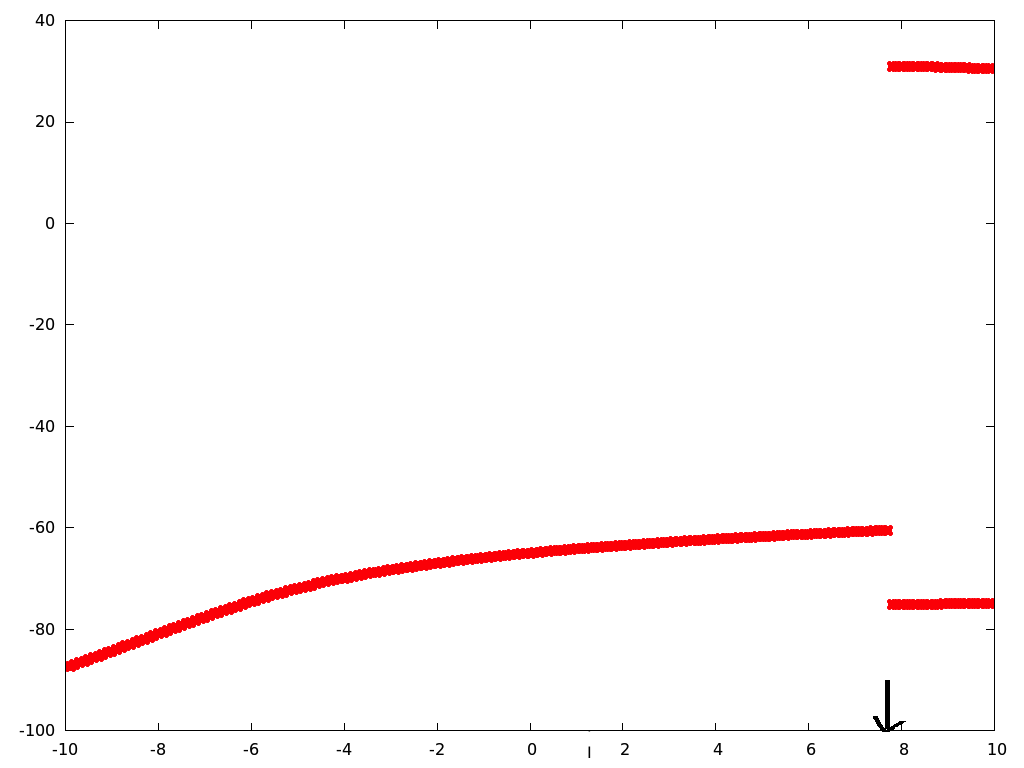}
\caption{Bifurcation Diagrams. Both figures plot $\limsup\limits_{t\rightarrow +\infty} V(t)$ and $\liminf\limits_{t\rightarrow +\infty}V(t)$ as a function of $I$ for fixed initial conditions. Left:$(V,n,m,h)(0)=(-50,0.5,0.5,0.5)$. Right: $(V,n,m,h)(0)=(-65,0.1,0.1,0.1)$. There is numerical evidence of coexistence of attractive limit cycle and stationary stable point for a region of $I$.}
\end{figure}
\begin{figure}
 \label{fig:Lc-Sta}
\includegraphics[height=5cm,width=5cm]{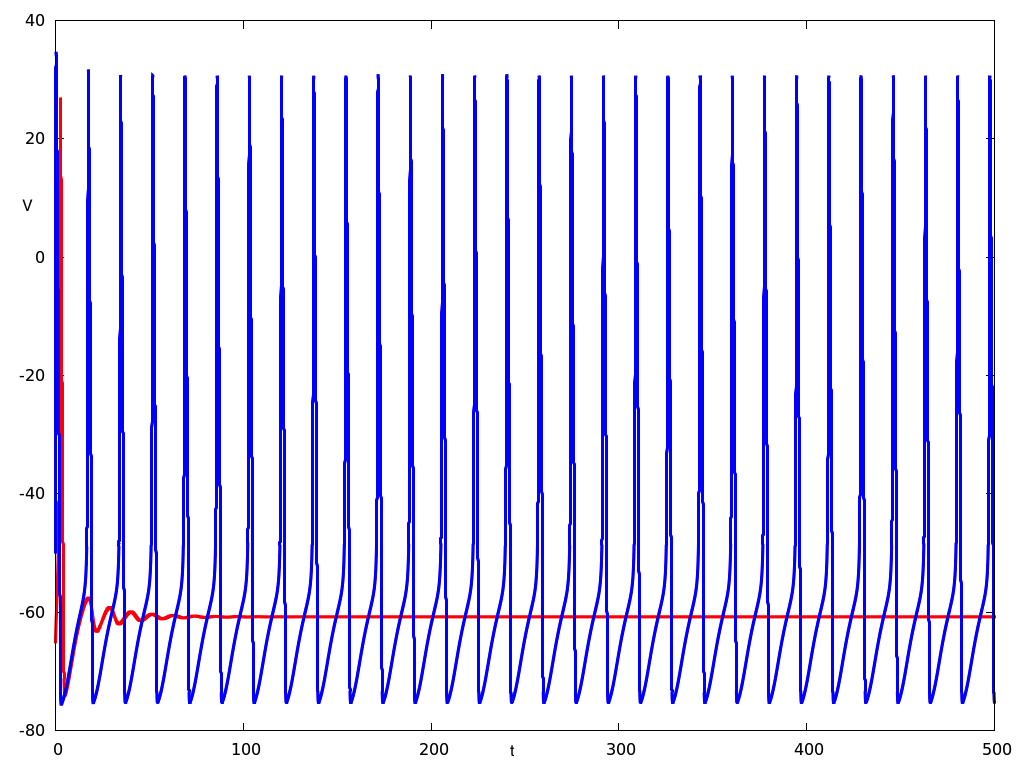}
\includegraphics[height=5cm,width=5cm]{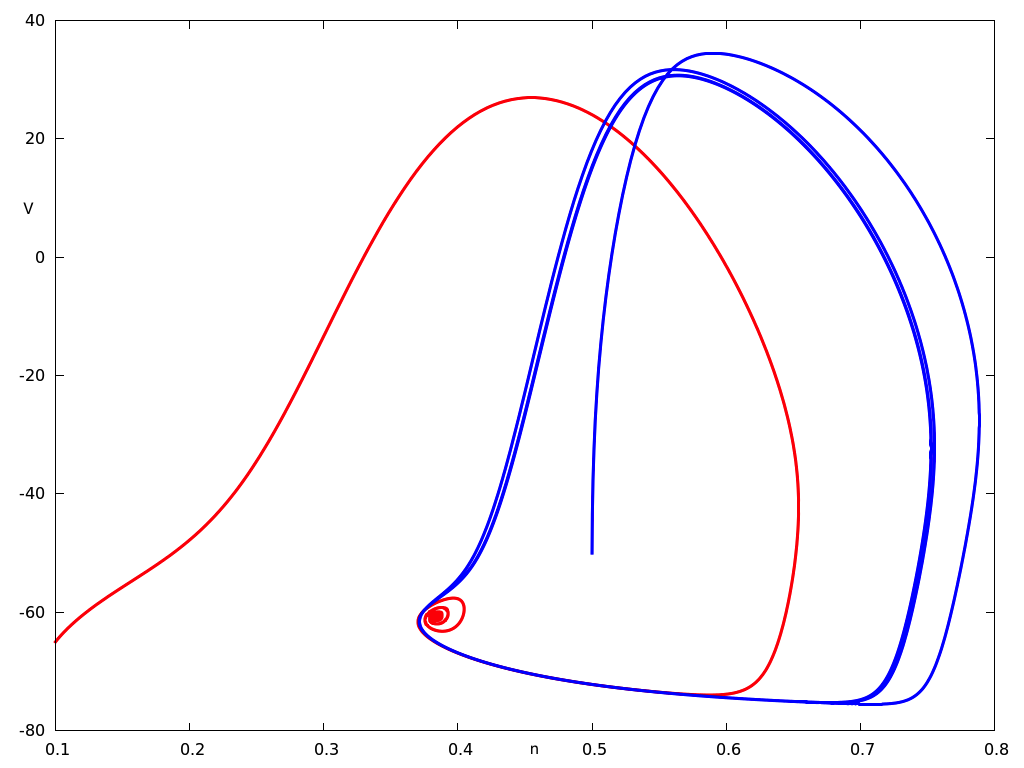}
\caption{Coexistence of attractive limit cycles and stationary point for $I=7$. Left: $V(t)$ for two initial conditions, $(V,n,m,h)(0)=(-65,0.1,0.1,0.1)$ in red, $(V,n,m,h)(0)=(-50,0.5,0.5,0.5)$ in blue. Right: analog representations in the $(n,V)$ projection plane. }
\end{figure}

\begin{figure}
 \label{fig:nmhvI=7}
\includegraphics[height=5cm,width=5cm]{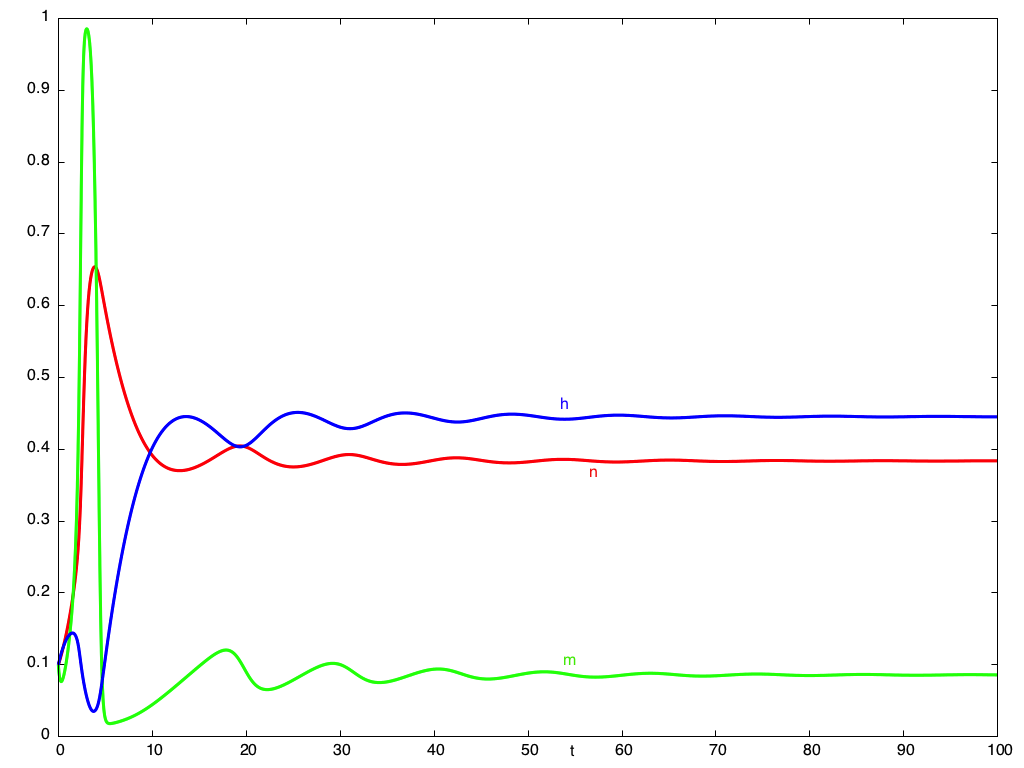}
\includegraphics[height=5cm,width=5cm]{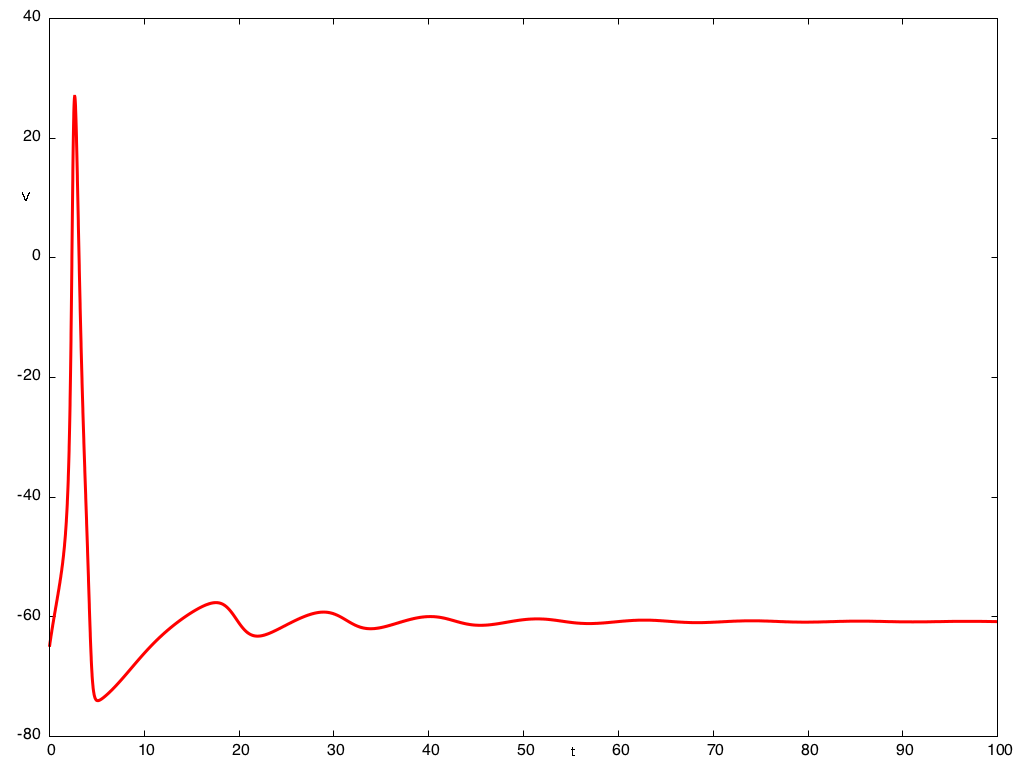}
\caption{Temporal evolution for $I=7$ and IC $(V,n,m,h)(0)=(-65,0.1,0.1,0.1)$ . Left: $n,m,h$ as functions of time, respectively in red, green and blue.   Right: $V$ as a function of time. These figures represent an action potential or spike. It is followed by a return to the stationary state. }
\end{figure}

\begin{figure}
 \label{fig:nmhvI=7p}
\includegraphics[height=5cm,width=5cm]{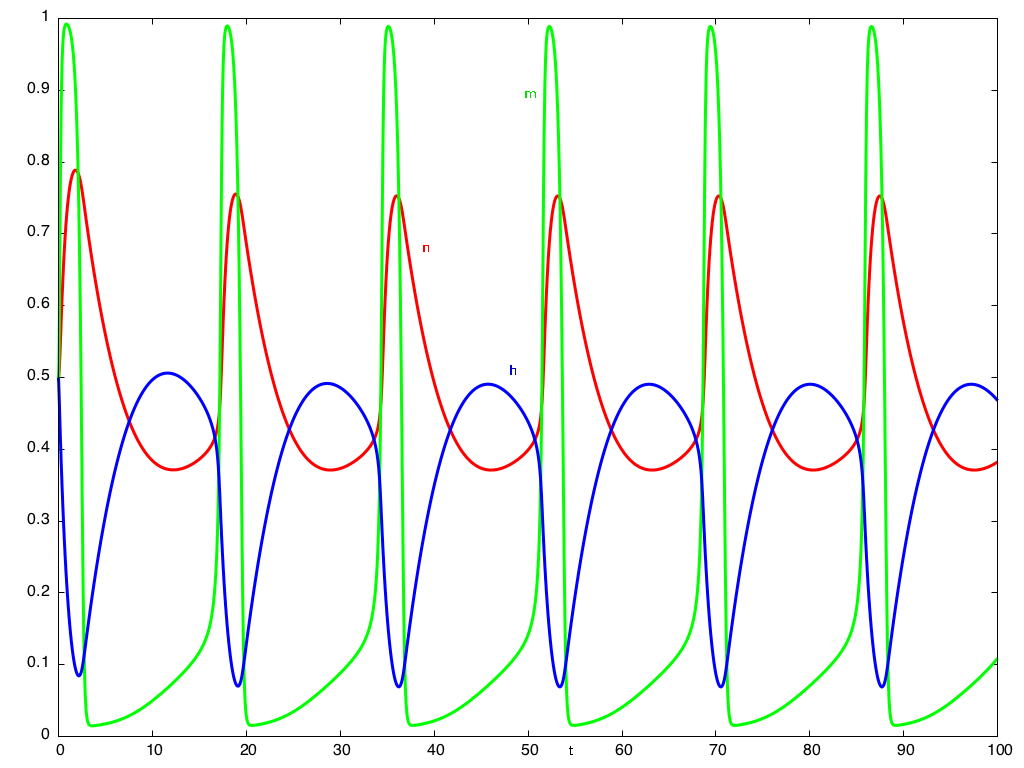}
\includegraphics[height=5cm,width=5cm]{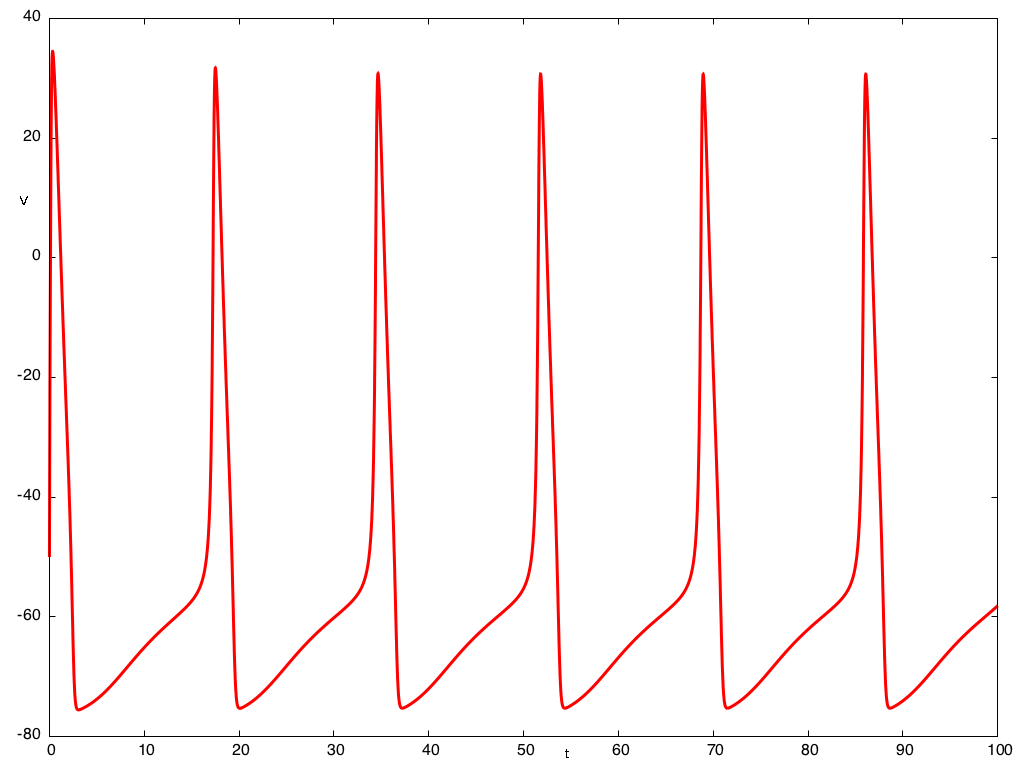}
\caption{Temporal evolution for $I=7$ and IC $(V,n,m,h)(0)=(-50,0.5,0.5,0.5)$ . Left: $n,m,h$ as functions of time.  Right: $V$ as a function of time. We observe a train of spikes. }
\end{figure}

\begin{center}
\begin{figure}
 \label{fig:nmhv-stationary}
\includegraphics[height=5cm,width=5cm]{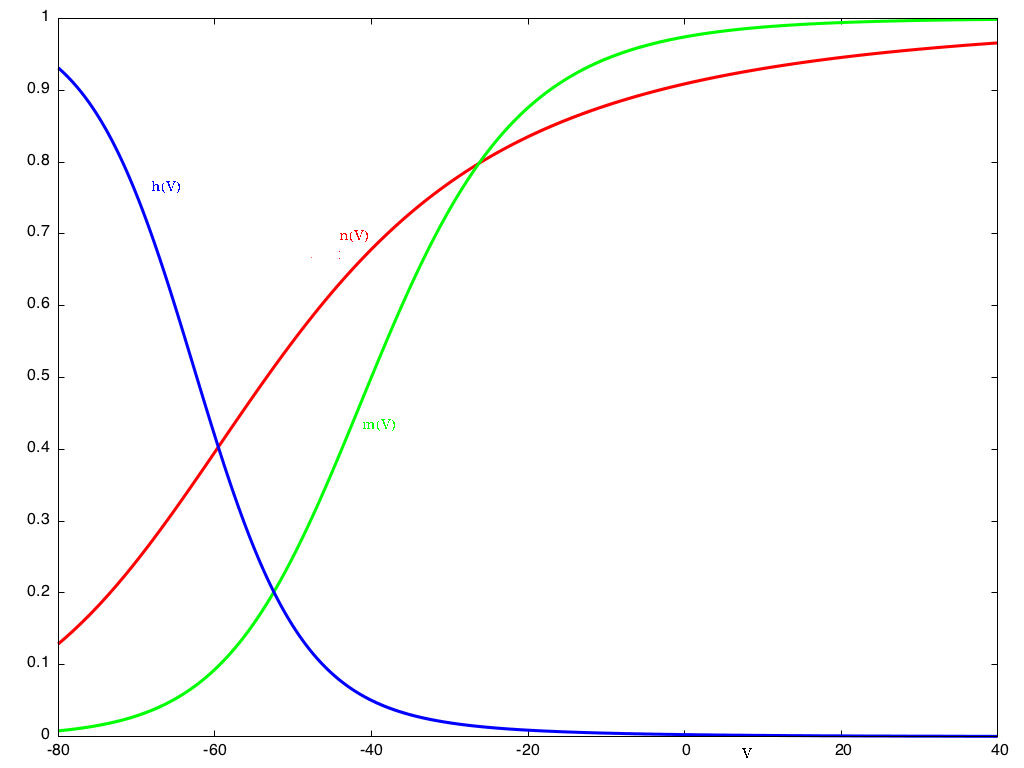}
\includegraphics[height=5cm,width=5cm]{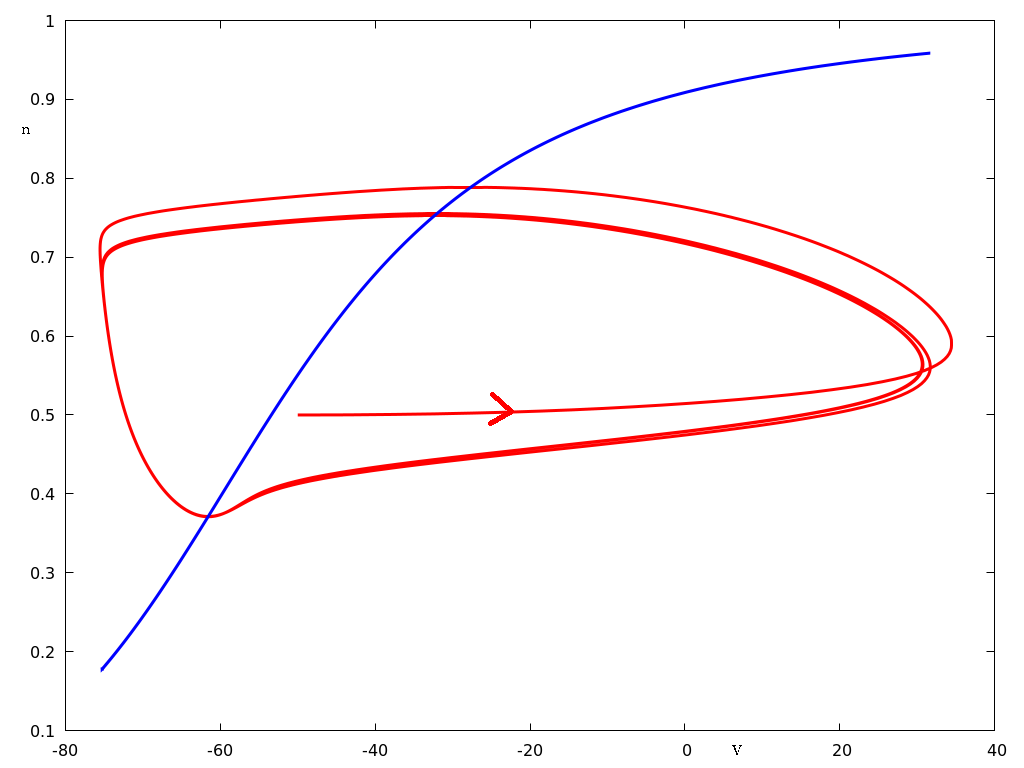}\\
\includegraphics[height=5cm,width=5cm]{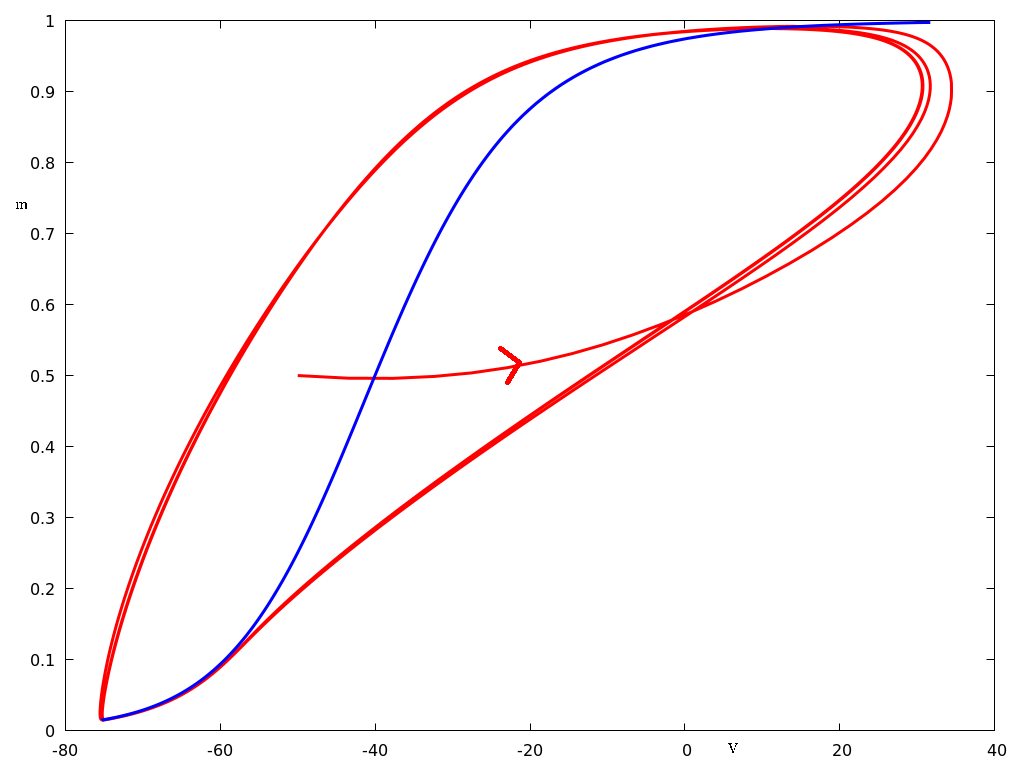}
\includegraphics[height=5cm,width=5cm]{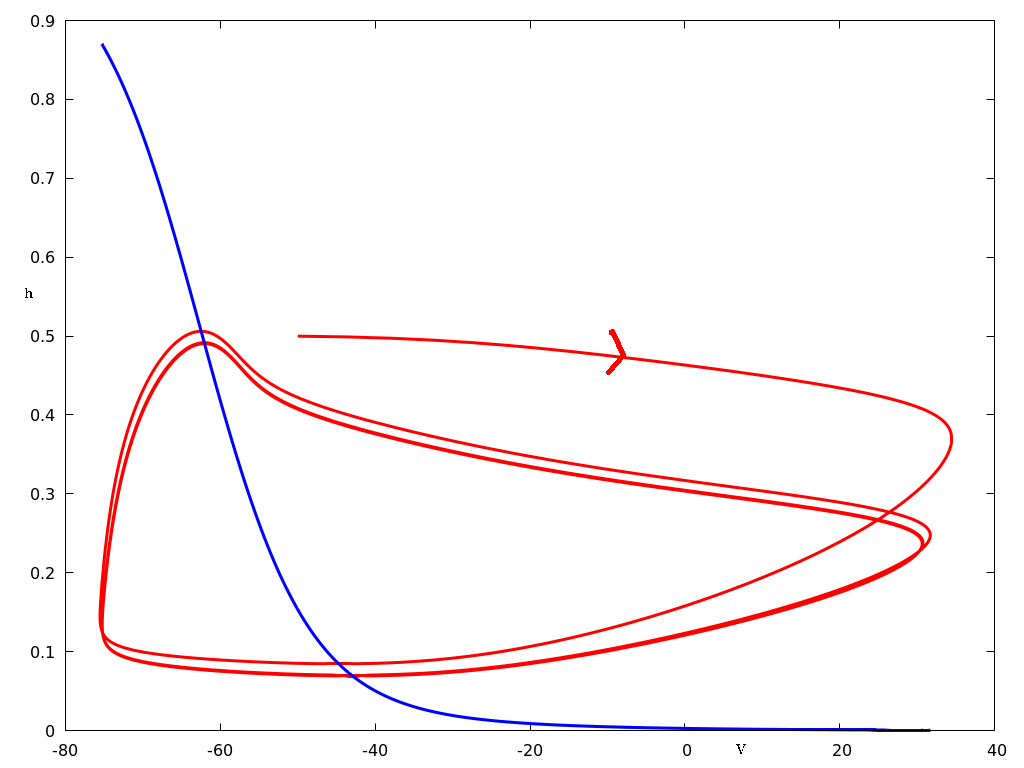}
\caption{Dynamics for $I=7$ and IC $(V,n,m,h)(0)=(-50,0.5,0.5,0.5)$. Top left panel: figure provides the nullclines of $n$ (red), $m$ (green) and $h$ (blue) as a function of $V$. One may rely on them to explain  the dynamics. The three others pictures, illustrate respectively the dynamics of $(V,n)$, $(V,m)$ and $(V,h)$, with their nullclines. These pictures highlight the dynamics of the HH model.   }
\end{figure}
\end{center}

\begin{figure}
 \label{fig:vnh3d}
\includegraphics[height=6cm,width=6cm]{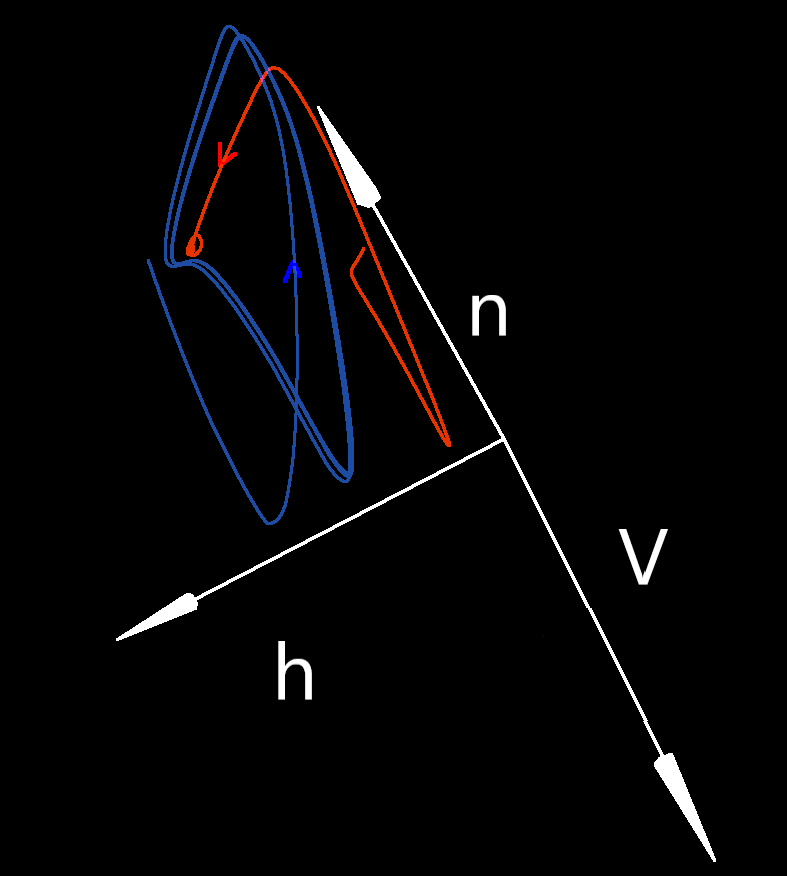}
\includegraphics[height=6cm,width=6cm]{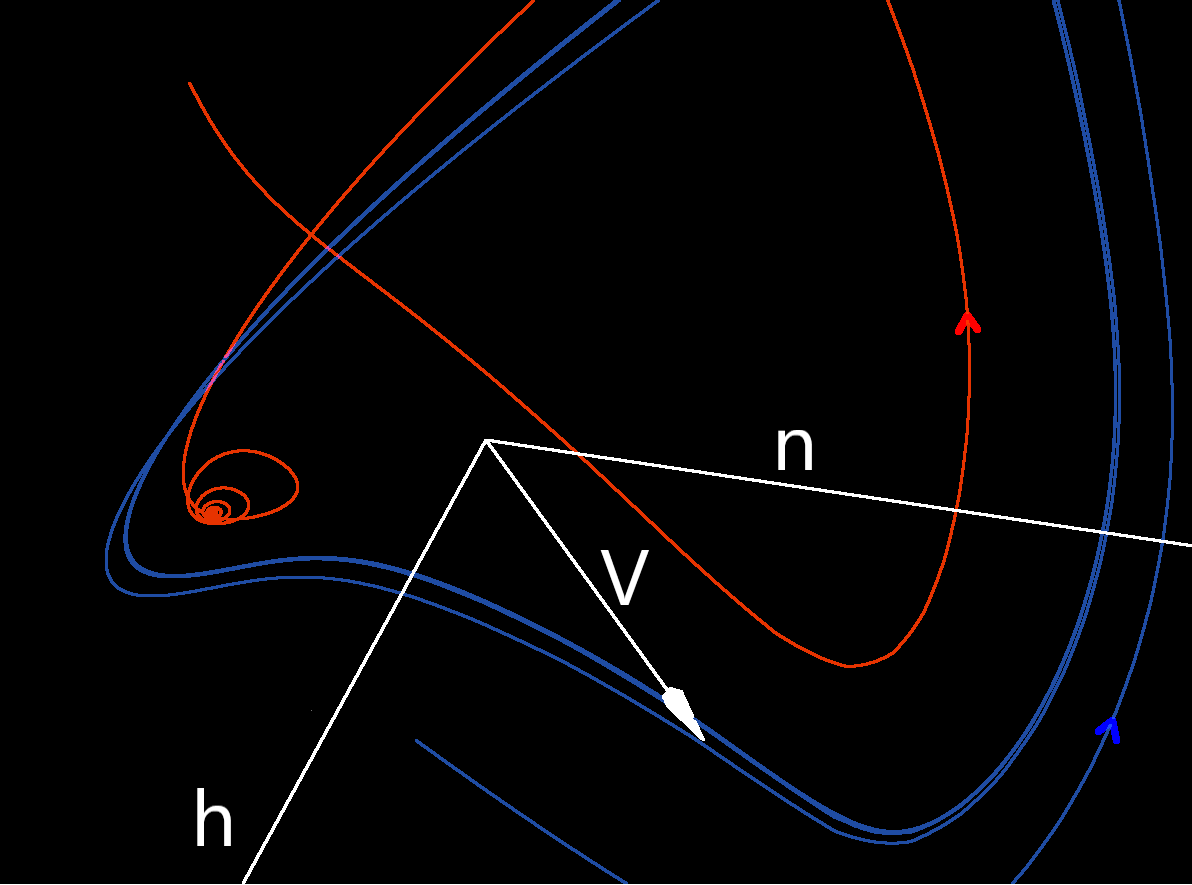}
\caption{Dynamics in the 3-dimensional phase space $(V,n,h)$ for $I=7$, for different IC. Blue curves correspond to IC $(V,n,m,h)(0)=(-50,0.5,0.5,0.5)$. Red curves correspond to $(V,n,m,h)(0)=(-65,0.1,0.1,0.1)$. There is numerical evidence of coexistence of attractive limit cycle and stationary stable point for a region of $I$.}
\end{figure}

\begin{figure}
 \label{fig:Batt}
\includegraphics[height=5cm,width=5cm]{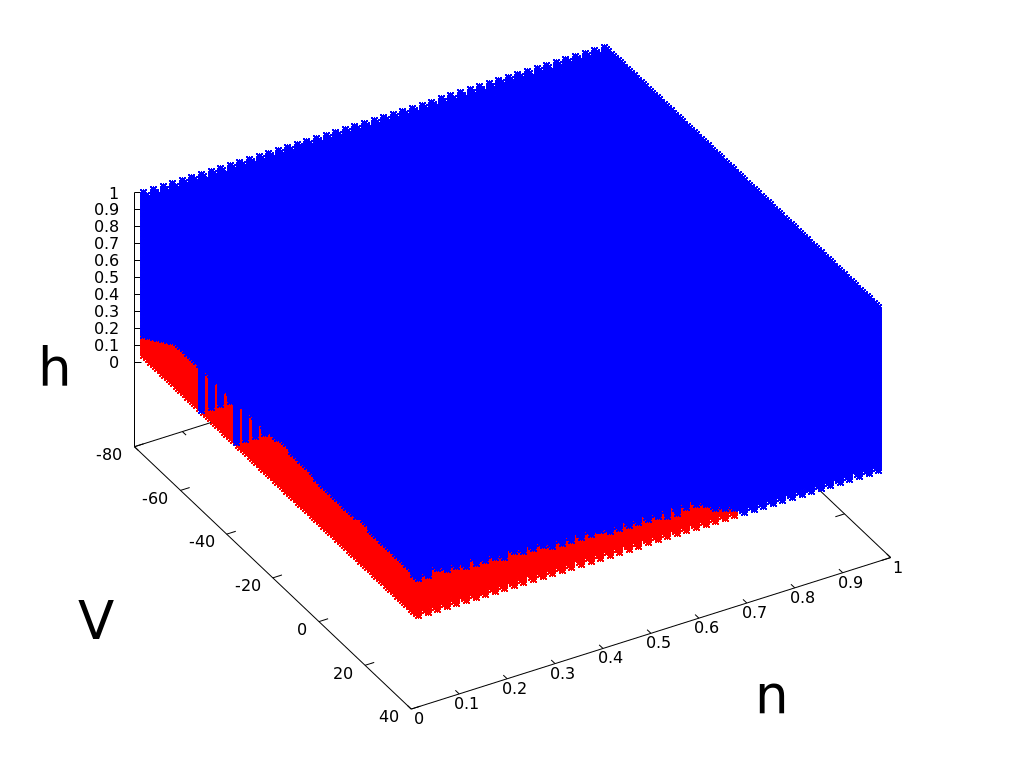}
\includegraphics[height=5cm,width=5cm]{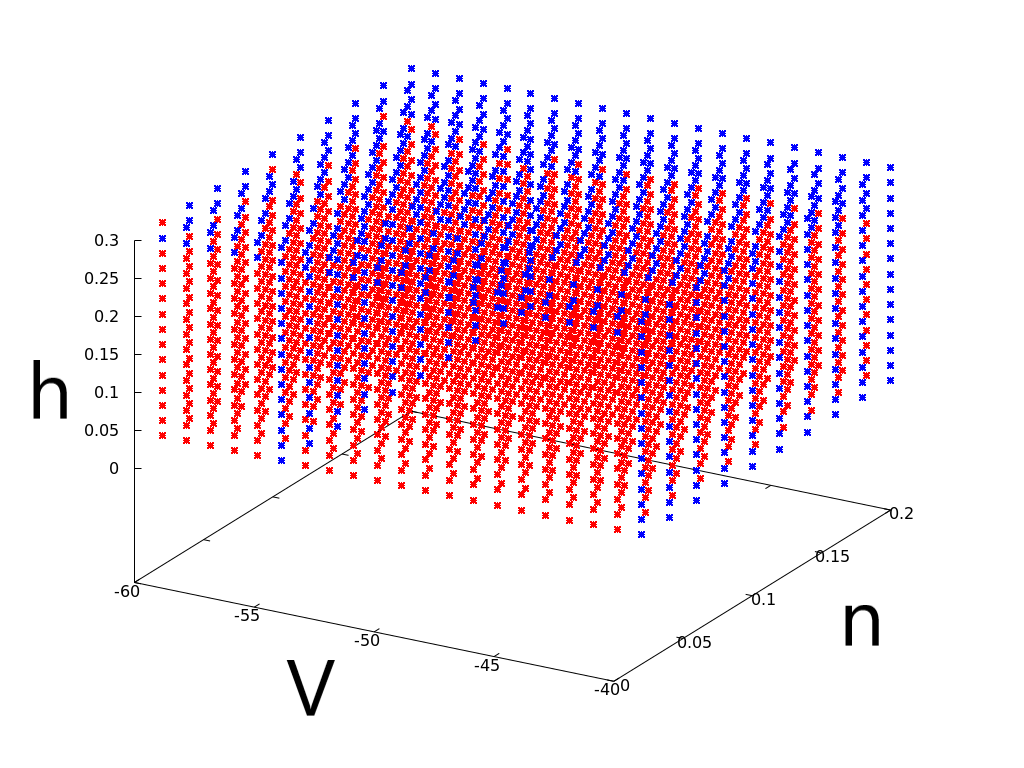}
\caption{This figure illustrates the basin of attraction for $I=7$, of the stationary solution and the limit-cycle. It is a projection in the $V,n,h$ plan.  In red, we plot IC which  evolve to the stationary point. In blue, we plot IC which evolve to the limit cycle.}
\end{figure}

\textbf{Action Potential and Excitability}
One of the reason of the success of the original HH paper, is the physiological based mechanism proposed to induce action potentials or spikes. From a dynamical modeling point of view, it corresponds to a large and fast excursion in the phase space, especially in the $V$ variable and away from the stationary stable point. This is also known as the excitability and illustrated in figures 7 and 8.  This feature of \eqref{eq:HH-single} is of fundamental importance for the next sections, for perturbations above a threshold will induce a spike.\\

In this section, we have revisited the dynamics of the HH equation. The key point for the next sections, is that varying the parameter $I$, or moving IC, the HH system is able to produces spikes. This is the key point for the network analysis because the spikes will determine the behavior of the network.

\section{One driven neuron}\label{Sec:single_HH-sto}
In this section, we focus on the dynamics of one driven neuron. We assume that $I=0$. This implies that if there is no drive inputs, as sketched in paragraph $1$ the system evolves toward the steady state. The equation writes in this case
\begin{equation}\label{eq:HH-onedriven}
\left\{
\begin{array}{rcl}
V_t  &=&\overline{g}_{Na} m^3 h(E_{Na}-V)+\overline{g}_{K} n^4(E_{K}-V)+\overline{g}_{L}(E_{L}-V)\\
& &+g_{E}(E_{E}-V)\\
n_t  &=& \alpha_{n}(V)(1-n)-\beta_{n}(V)n\\
m_t  &=& \alpha_{m}(V)(1-m)-\beta_{m}(V)m\\
h_{t}  &=& \alpha_{h}(V)(1-h)-\beta_{h}(V)h\\
\tau_E g_{Et}  &=& -g_E+S^{dr}\sum_{s\in \mathcal{D}}\delta(t-s)
\end{array}
\right.
\end{equation}
where  $\mathcal{D}$ refers to the set of times at which the neuron  receives kicks from the stochastic drive.
In comparison with system \eqref{eq:HH-single}, we have added one equation. This last equation accounts for external drive. Mathematically, we assume that jumps of amplitude $\frac{S^{dr}}{\tau_E}$ occur in the time course of $g_E$ as a realization of a Poisson process of rate $\lambda$. 
%
 Ordering the elements of the set $ \mathcal{D}$ by increasing order, we denote
\[\mathcal{D}=(s_i)_{i\in \N}\] 
and the time intervals $(s_{i+1}-s_i)$, are fixed at the beginning as realization of exponential laws of parameter $\lambda$ in order to generate a realization of a Poisson process.

\textbf{Biological interpretation} \\
The introduction of these spikes stand in the model as the external drive, which is of fundamental role in applications, since electrical activity of neurons results from the interaction of the external drive and recurrent inputs. Mathematically, in this  model, the external drive is Poissonian and generates Diracs. The recurrent inputs correspond to the coupling terms coming from the network. In this paragraph, we focus on the external drive.

This last formulation, dividing the interval $[0,T]$ into subdivision  $\{t_0,t_1,...\}$ is well adapted to our framework. Note that the derivative $g_{Et}$ has to be taken in the  sense of distributions. A classical computation leads to:
 \[g_E(t)=g_E(t_i)e^{-\frac{1}{\tau_E}(t-t_i)}\]
  on the time interval $[t_i,t_{i+1})$, then at time $t_{i+1}$ a kick arrives and, 
 \[g_E(t_{i+1})=g_E(t_{i})e^{-\frac{1}{\tau_E}(t_{i+1}-t_i)}+\frac{S^{dr}}{\tau_E}\]
  Therefore $g_E(t)$ is a discontinuous function with jumps,  and is locally bounded.
The aim of  this section is, given the input drive, to look at the behavior resulting from an increase in the parameter $S_{dr}$.  
\subsection{Some analytical results}
Before going into that, we state some propositions which clarify the mathematical framework. 
\begin{proposition}
\label{prop:HH-onedriven}
The trajectories of the stochastic process generated by equation \eqref{eq:HH-onedriven} are defined on $(0,+\infty)$  and piecewise-$C^1$; Furthermore, $V$ is continuous, $n,m,h$ are $C^1$, $g_{E}$ has jumps. Furthermore, the set 
\[(E_K,E_{Na})\times (0,1)^3 \mbox{ is  positively invariant for } (V,n,m,h) \]
\end{proposition}
\begin{proof}
Existence and uniqueness on each interval $[t_i,t_{i+1})$ follows from the Cauchy theorem. At time $t_{i+1}$ a jump occurs, which determines the value of $g_E$ in the next time interval. It follows that the solution $g_E$ is  $C^{\infty}$ in each interval $(t_i,t_{i+1})$.  Next, we deal with the boundedness of trajectories.
The proof of theorem \ref{th:posinv} remains valid, with the specific assumption that $I=0$. There are jumps on $V_t$ but $V$ is continuous. The derivative $V_t$ is negative if $V$ is above $E_{Na}$ and positive if $V$ is below $E_{K}$. This implies the result. 
\end{proof}

\begin{remark}
Note that the value of $g_E$ after the jump is therefore given by the following recurrence equation.
\begin{equation*}
g_E(t_{i+1})=g_E(t_i)\exp(-\frac{t_{i+1}-t_i}{\tau_E})+\frac{S^{dr}}{\tau_E}. 
\end{equation*}

\end{remark}

The following proposition follows from the above recurrence equation.
\begin{proposition}
We assume $g_E(0)=0$. Then, for $t\in [t_i,t_{i+1})$,  $g_E$ is given by the following expression:
\[g_E(t)=\frac{S^{dr}}{\tau_E}\sum_{k=1}^i\exp(\frac{t-t_k}{\tau_E})\] 
\end{proposition}

Since the $t_{k+1}-t_k$ are independents exponential laws, we can compute the value of the mean $E[g_E(t_i)]$ (just before kick). Computations lead to the following proposition.
\begin{proposition} 
Under the above assumptions, the following expression holds:
\[E(g_E(t_i))=\frac{S^{dr}}{\tau_E}\sum_{k=1}^{i-1}\frac{(\lambda \tau_E)^k}{(\lambda\tau_E+1)^k}=\frac{S^{dr}}{\tau_E}\frac{r-r^{i}}{1-r} \]
where \[r=\frac{\lambda \tau_E}{1+\lambda \tau_E}.\]
And,
\begin{equation}
\label{eq:drive}
\lim_{i\rightarrow+\infty}E[g_E(t_i)]=S^{dr}\lambda
\end{equation}
\end{proposition}

\textbf{Biological interpretation} \\
Note that equation \eqref{eq:drive} gives a quantitative simple information about the drive. Roughly speaking, it says  that the mean value of $g_E$ after kicks and exponential decay is the product of the amplitude $S^{dr}$ and the frequency of inputs  per ms $\lambda$.

\subsection{Varying $S^{dr}$}
Next, we set values of $\lambda=0.9$ and $\tau_E=2$, and vary $S^{dr}$. When increasing $S^{dr}$, we reach a threshold at which the driven neuron starts to spike. Increasing more $S^{dr}$ increases the spiking rate. This is illustrated in figure 12. This numerical analysis serves as a basis for the study of the network for which we set $\rho_E=0.9$, $\rho_I=2.7$, $\tau_E=2$ and $S^{dr}=0.04$. For theses values, one single driven neuron, exhibits multiple spikes.

\begin{figure}
 \label{fig:BifHHDriven}
 \includegraphics[height=2.3cm,width=2.3cm]{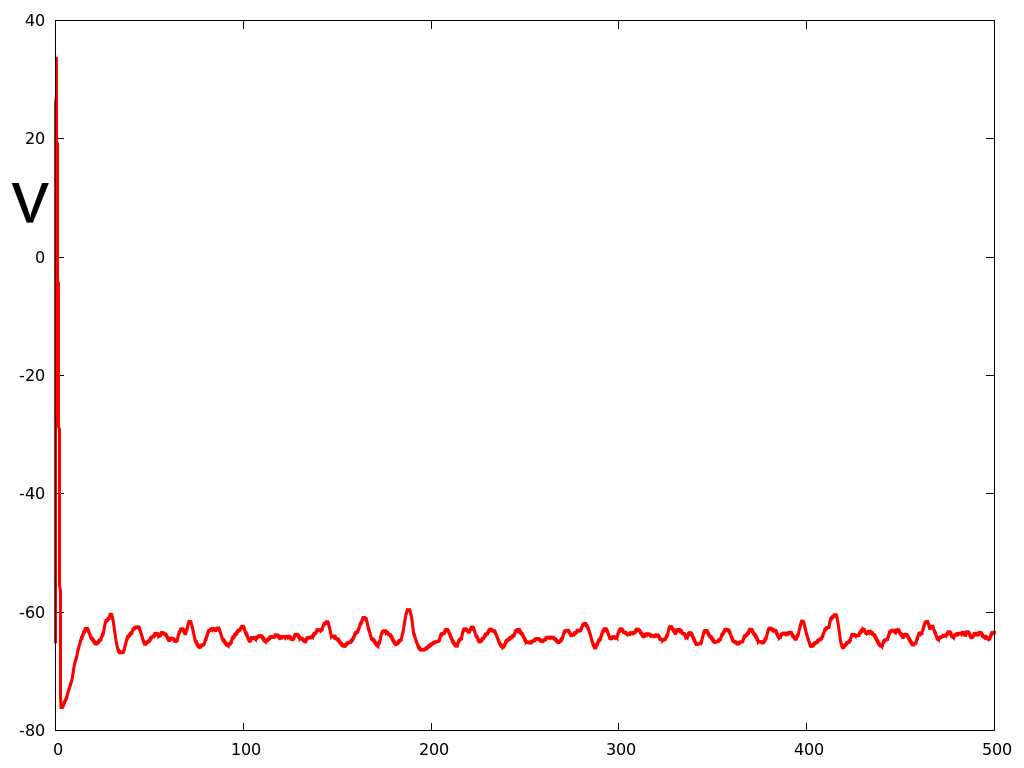}
\includegraphics[height=2.3cm,width=2.3cm]{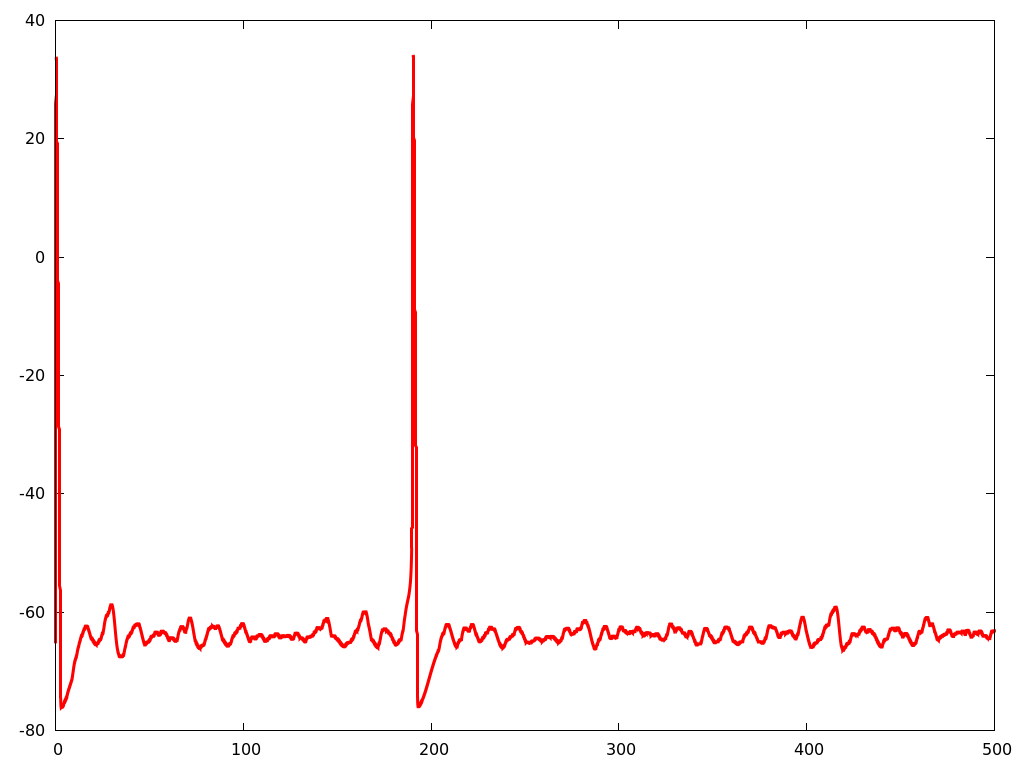}
\includegraphics[height=2.3cm,width=2.3cm]{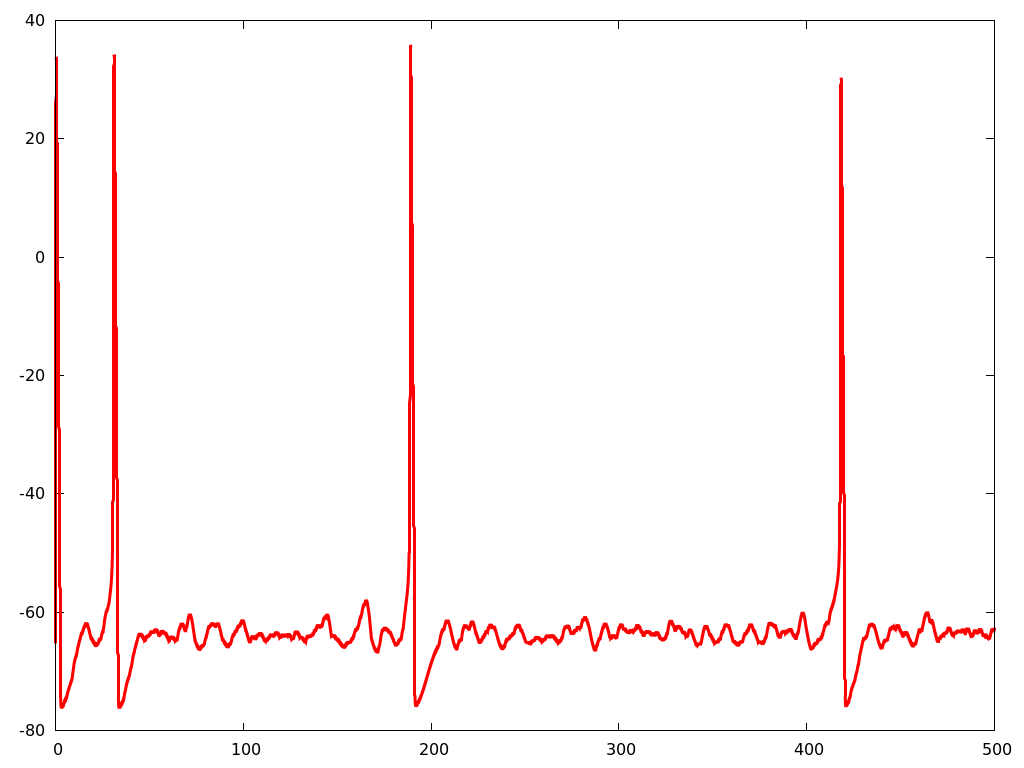}
\includegraphics[height=2.3cm,width=2.3cm]{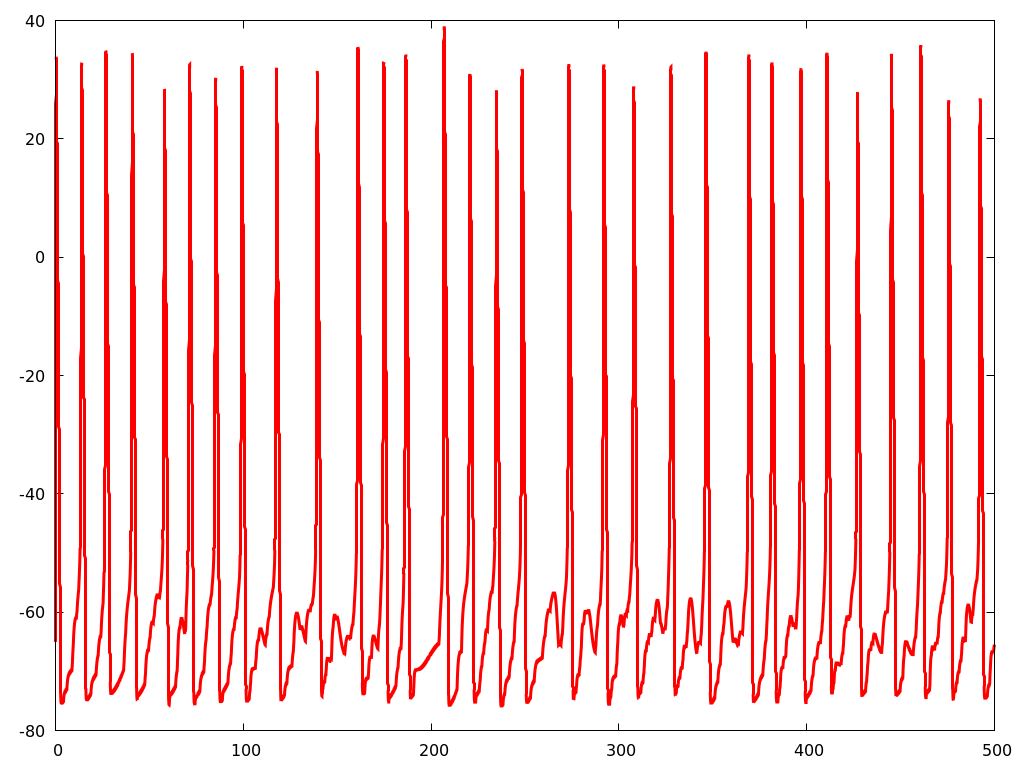}
\includegraphics[height=2.3cm,width=2.3cm]{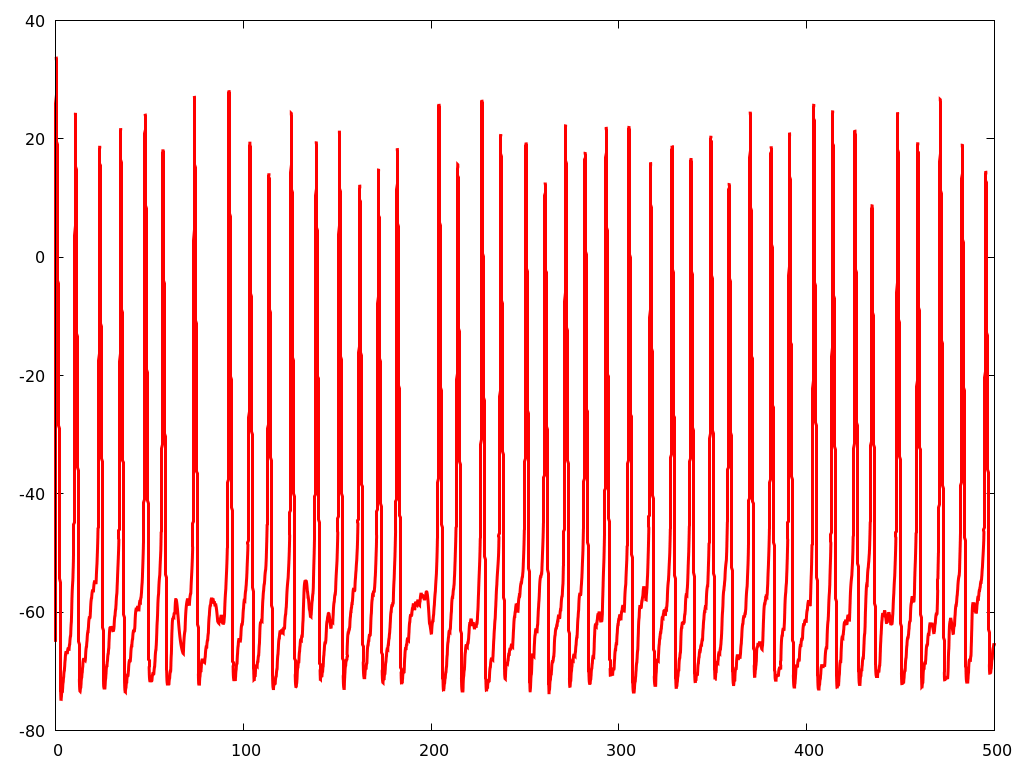}\\
\includegraphics[height=2.3cm,width=2.3cm]{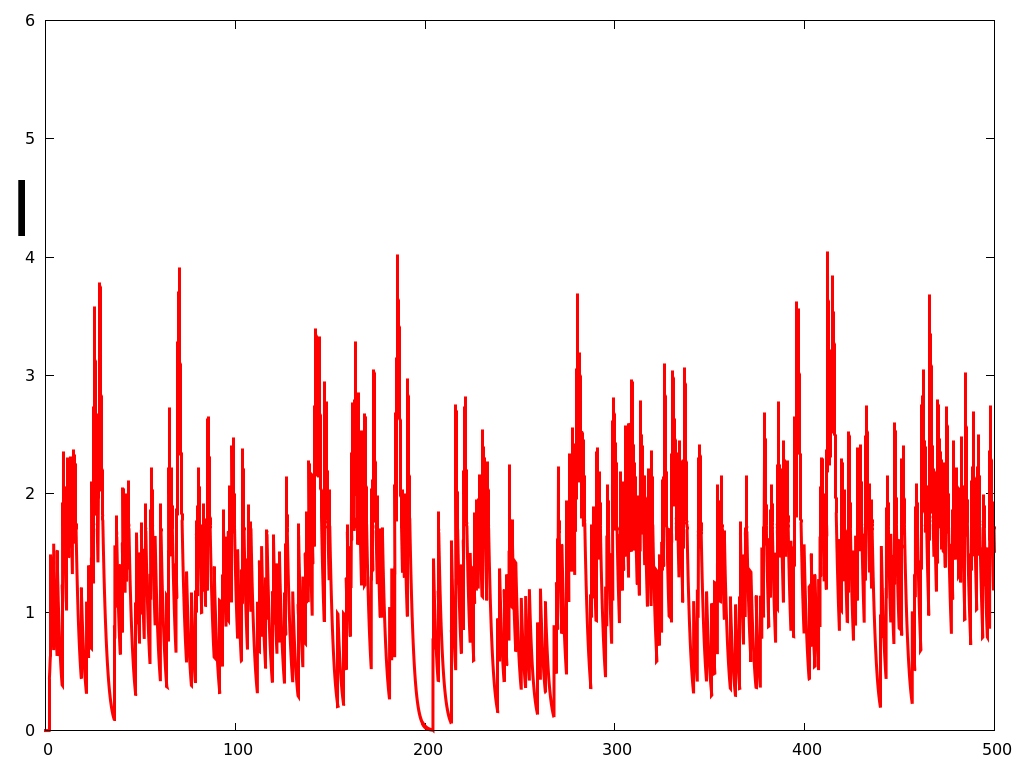}
\includegraphics[height=2.3cm,width=2.3cm]{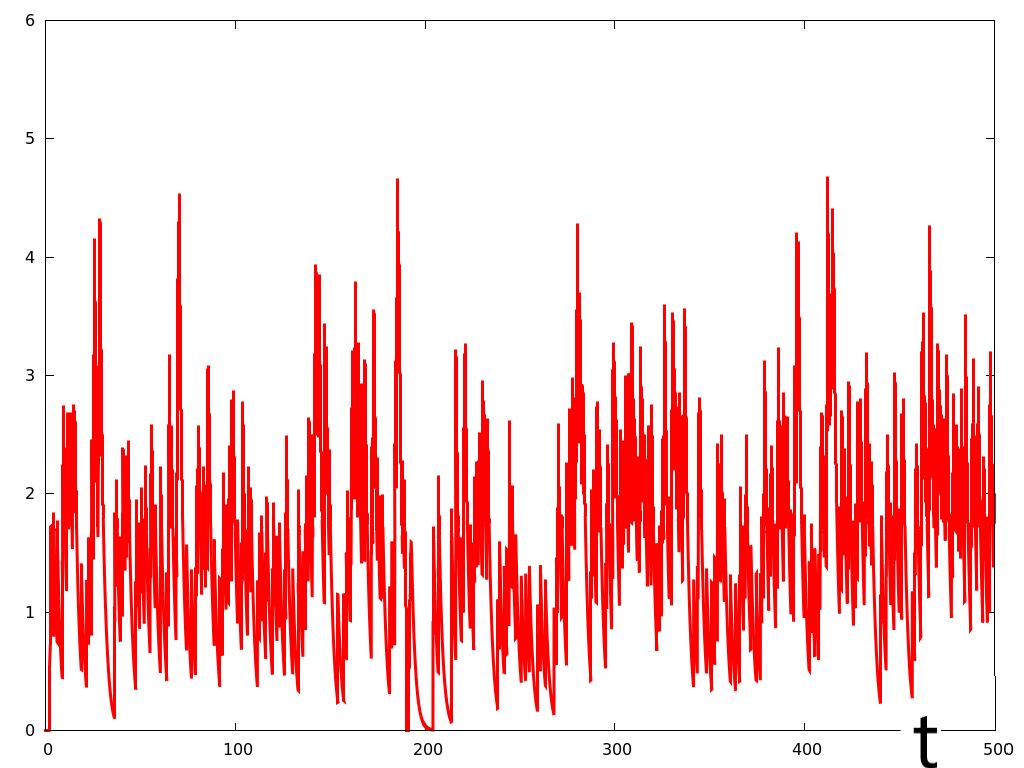}
\includegraphics[height=2.3cm,width=2.3cm]{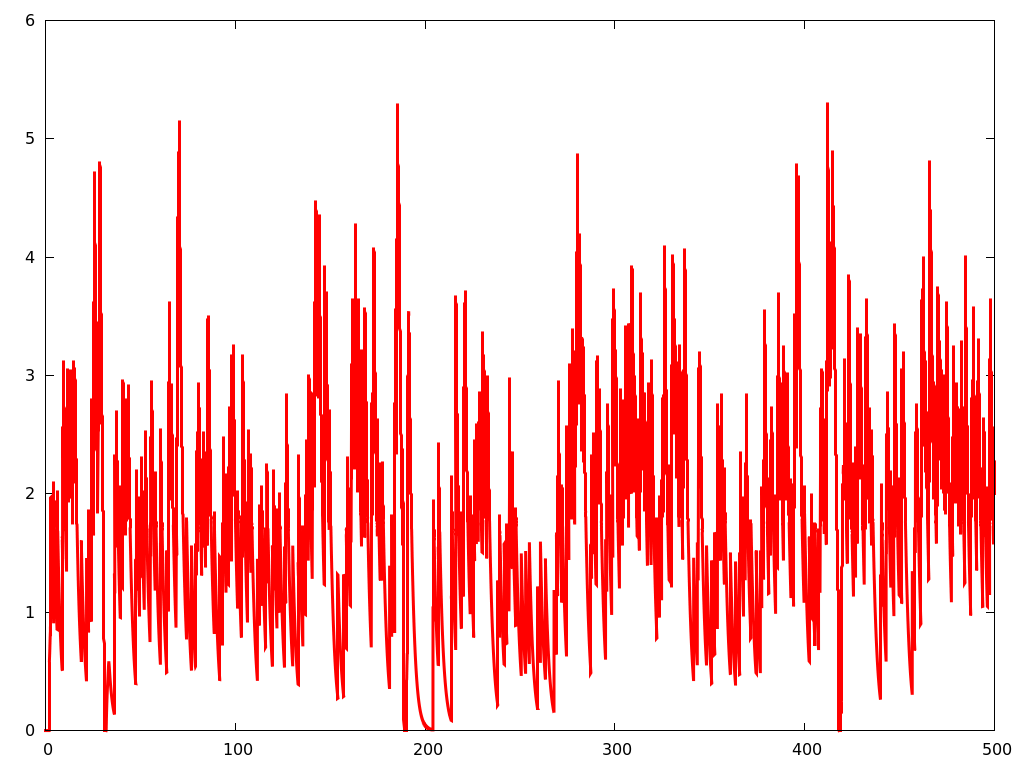}
\includegraphics[height=2.3cm,width=2.3cm]{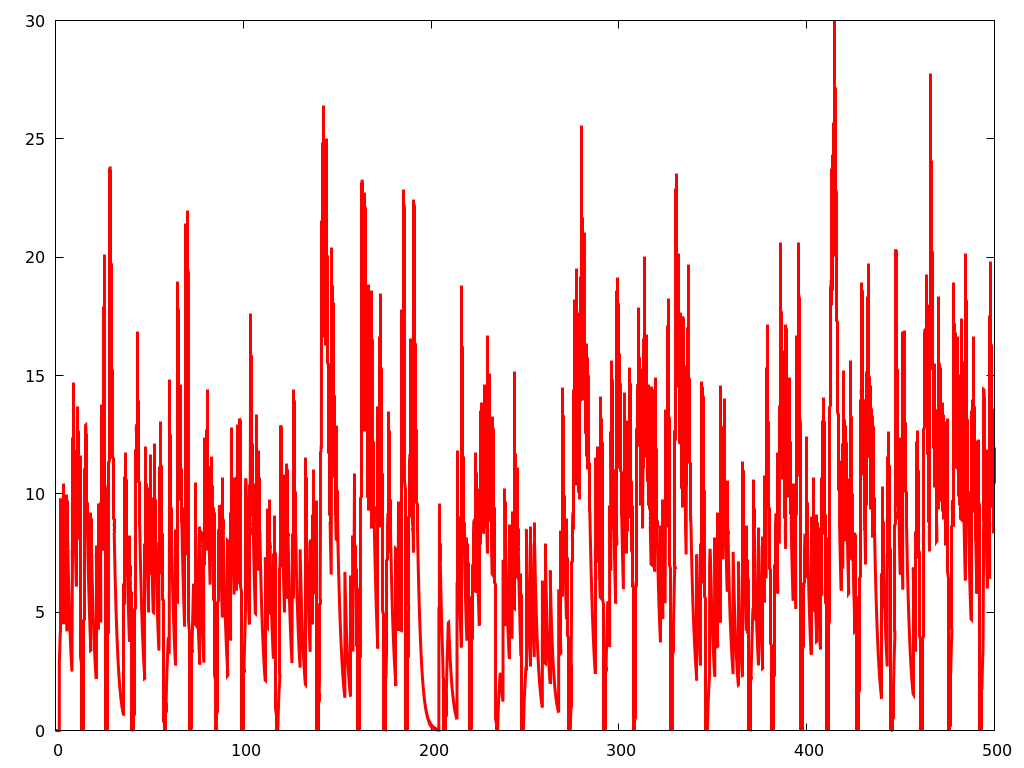}
\includegraphics[height=2.3cm,width=2.3cm]{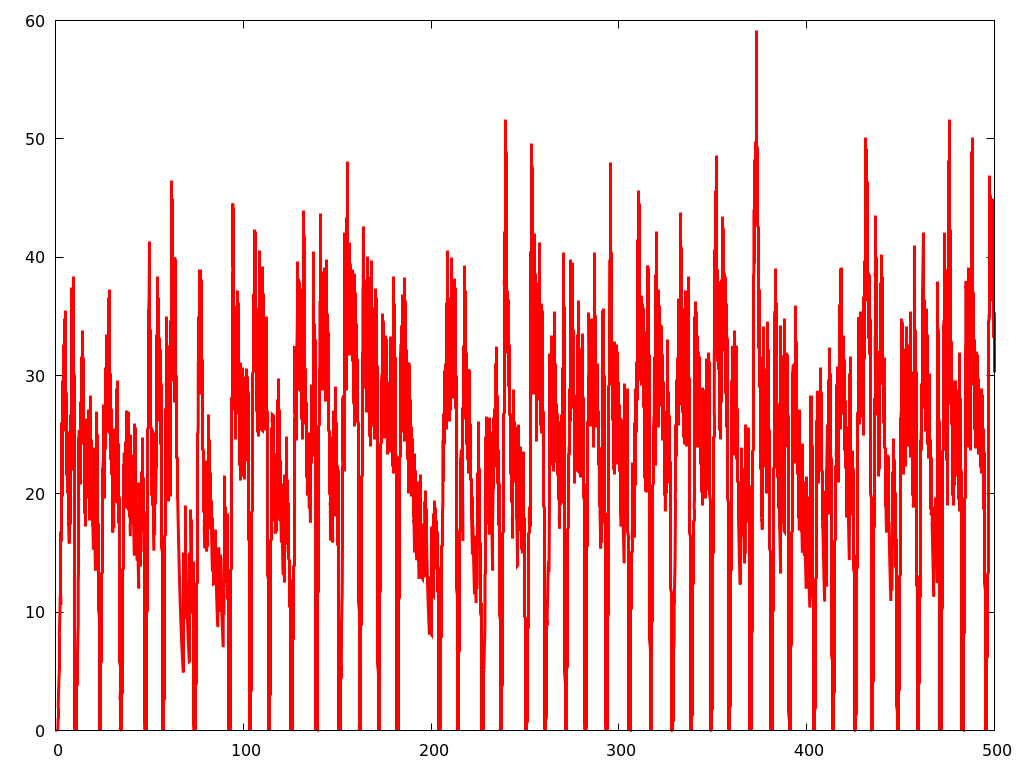}\\
\includegraphics[height=2.3cm,width=2.3cm]{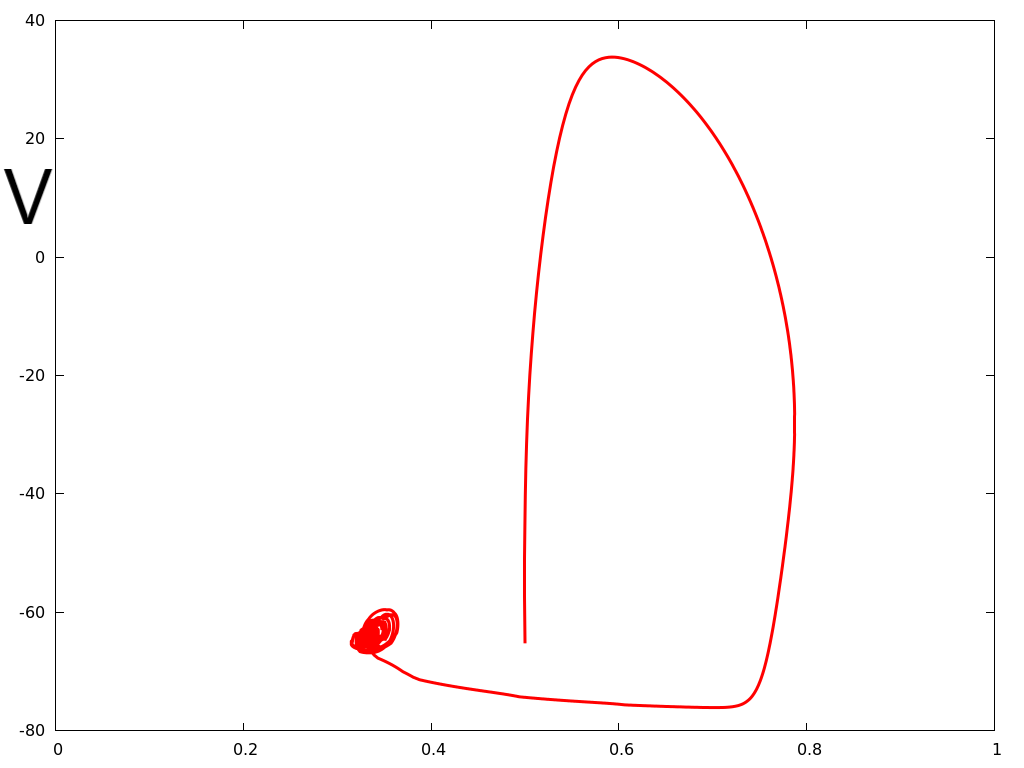}
\includegraphics[height=2.3cm,width=2.3cm]{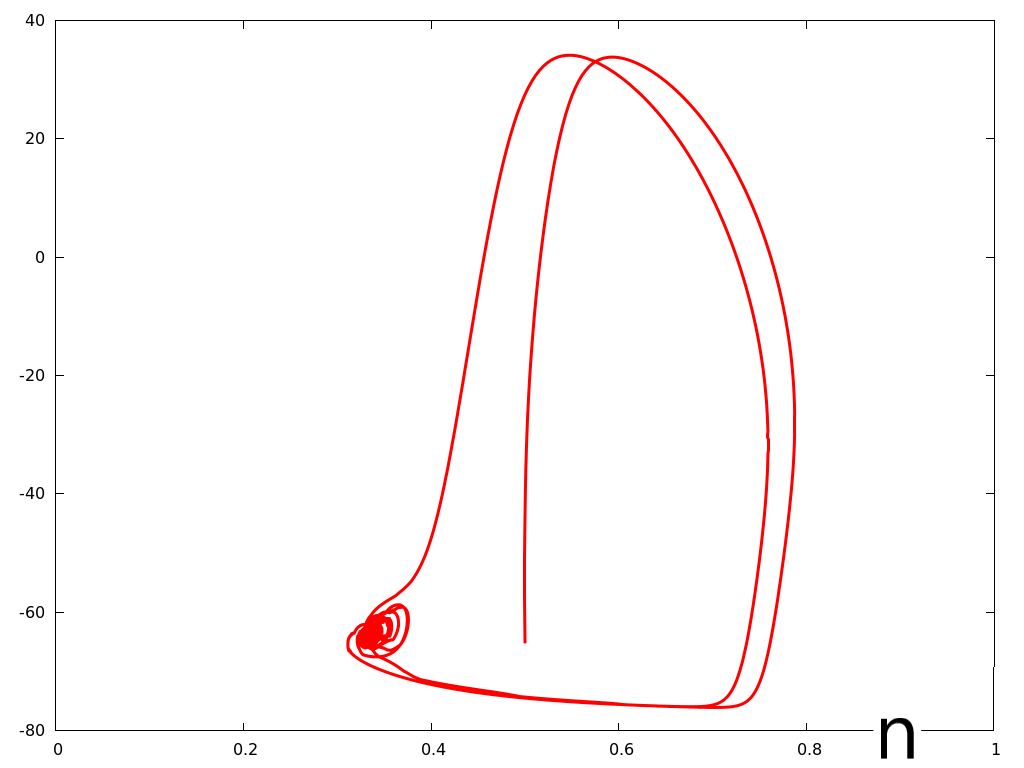}
\includegraphics[height=2.3cm,width=2.3cm]{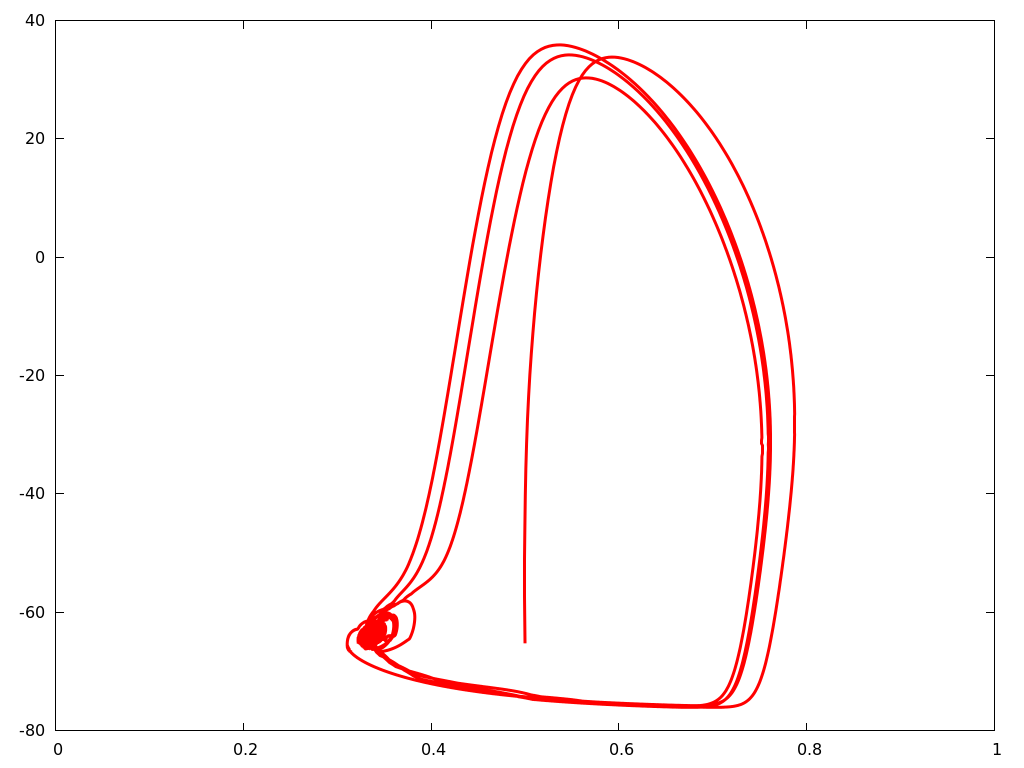}
\includegraphics[height=2.3cm,width=2.3cm]{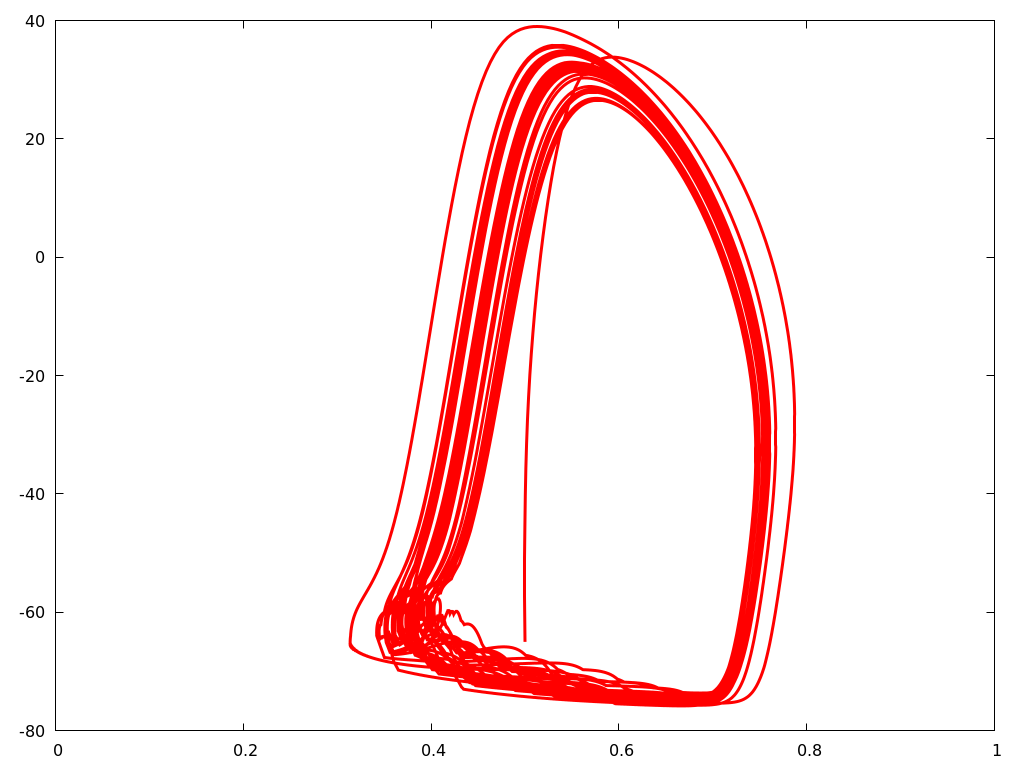}
\includegraphics[height=2.3cm,width=2.3cm]{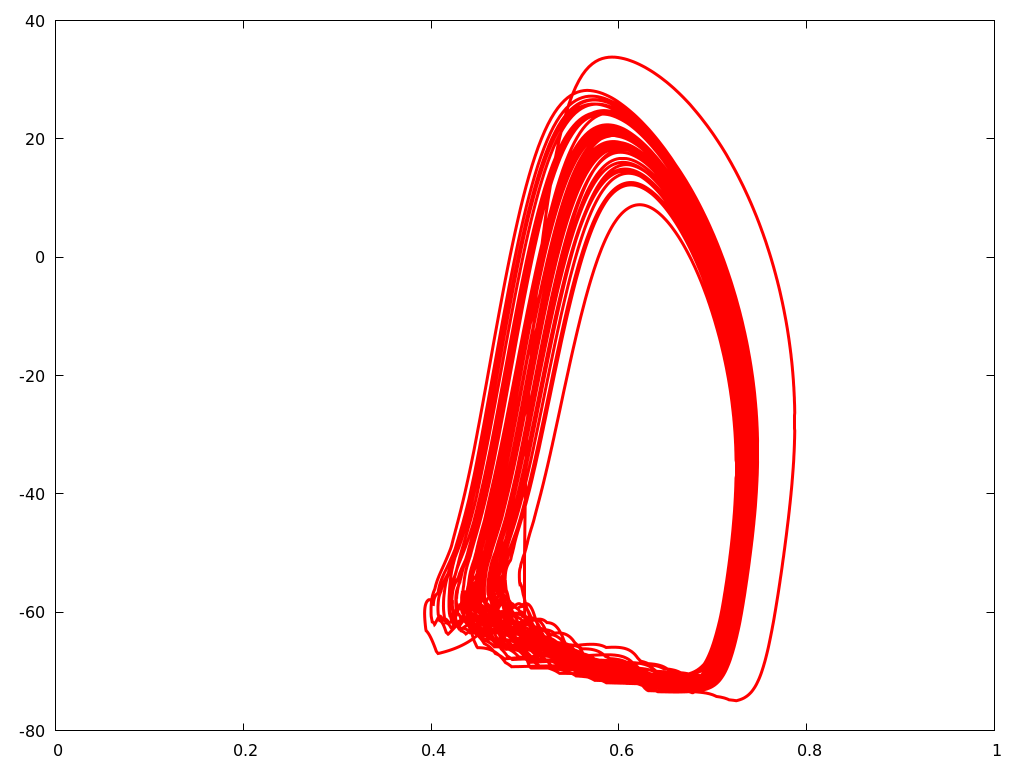}
\caption{Simulations of equation \eqref{eq:HH-onedriven}, for $S^{dr}\in\{0.006,0.007,0.008,0.04\}$ are illustrated. We observe    a bifurcation between no spike and spiking regime.  The first row represents $V(t)$, the second row $-g_EV(t)$ which corresponds to a $I(t)$ for \eqref{eq:HH-single}. The last row illustrates the projection in the $n-V$ plane. The  first spiking regime occurs for $S^{dr}=0.007$ (mean of $-Vg_E\simeq 1.5$). The value for $\lambda$ is $0.9$. For $S^{dr}=0.006$ (mean of$ -Vg_E\simeq 1.3$), there is no spikes, while for $S^{dr}=0.008$, we obtain $6$ spikes per second. Note that each spike for $V$  occurs after an increase signal in $-g_EV$.  For $S^{dr}=0.04$ (mean of$ -Vg_E\simeq 7.4$), there is $60$ spikes per second (for comparison, for $I=7$ in the deterministic case, $i.e.$ for equation \eqref{eq:HH-single}, we had $600$ spikes per second). The last column correspond to $S^{dr}=0.04$ and $\lambda=2.7$. We observe a frequency of $84$ spikes per second (mean of $ -Vg_E\simeq 22$). The values of parameters for the two last columns are those set to the input drive for $E$ and $I$-neurons.}
\end{figure}
\begin{remark}

 It is worth noting that for the last two columns, neurons exhibit a quite high frequency regime. Those values of parameters characterizing the stochastic drive are those which will be set for the $E$ and $I$-neurons. Next, we will focus on the network, which dynamics reflects a balance between the stochastic drive studied in this paragraph, and  the effects of recurring inputs coming from the network.
\end{remark}

\section{Emergent properties in a stochastically driven network}\label{Sec:network_analysis}
We first mention the following theorem which provides a theoretical framework before delving into the numerical analysis. The proof is analog to the one of proposition \ref{prop:HH-onedriven}
\begin{theorem}

The trajectories of the Stochastic Process generated by equation \eqref{eq:HH-Network} are defined on $(0,+\infty)$  and piecewise-$C^1$; For every $i\in \{1,...,N\}$, $V_i$ is continuous, $n_i,m_i,h_i$ are $C^1$, $g_{Ei},g_{Ii}$ have jumps. Furthermore, the set 
\[(E_K,E_{Na})\times (0,1)^3 \mbox{ is  positively invariant for } (V_i,n_i,m_i,h_i) \]

\end{theorem}
In this section, our aim is to illustrate how variation of parameters leads to emergent properties in the network. This section results from a large set of numerical simulations. We have started our exploration from:
\[S^{EI}=S^{IE}=S^{II}=S^{EE}=0.01\]
and then moved each one of the parameters within the range $[0.002;0.03]$. Among the parameters under consideration,  tuning the value of the parameter $S^{EE}$ appears to be the most effective way to  synchronization in the network. Varying this parameter allows to identify a path from stochastic homogeneity to synchronization. According with our numerical simulations, we will further discuss and illustrate the apparition of the following phenomena in the network:
\begin{itemize}
\item path from homogeneity to partial  synchronization and  synchronization; 
\item correlation between $g_E$ and $g_I$;
\item emergence of the $\gamma$ rhythm. At some point the network has his own rhythm of oscillation consistent with the so called gamma frequency, and which may be different from individual neuronal rhythms;
\item we will also discuss the effect of the parameters variation on the mean spiking rate of E and I neurons.
 
\end{itemize}
This section relies strongly on figures 14 and 15.
\subsection{A path from random-homogeneity toward partial  synchronization and  synchronization}
In this part, we focus on the following set of  parameters:
\[S^{EI}=S^{IE}=S^{II}=0.01\]
and
\[S^{EE} \mbox{ is varied from }0.01 \mbox{ to } 0.03.\]
This variation of $S^{EE}$ provides a path along which the system goes from stochastic homogeneity toward partial  synchronization and  synchronization. We describe hereafter the dynamical behavior corresponding to three distinct values of $S^{EE}$ at which those typical states are observed. The main tool used to characterize these states is the rasterplot: for each time, we plot the neurons which are in a spiking state. The rasterplots relevant for this section are illustrated in the first row of the figure 14\\
\subsubsection{Random-homogeneity}
In this paragraph, we consider the parameter values
\[S^{EI}=S^{IE}=S^{II}=S^{EE}=0.01.\]
For these values, the network exhibits a behavior for which no specific pattern seems to emerge, see top left panel of figure  14. We call this state randomly homogeneous since it looks like the spikes of neurons might be chosen to spike randomly and independently. Note that the mean value of spikes per second is around $11.49$ for $E$-neurons and $48.48$ for $I$ neurons, see table 2. A simple and relevant indicator of the total excitability of the network is given by the mean  value of $V$ across the network over the time. In figure 13, we have therefore plotted the mean value of $V$ over all  neurons as a function of time. For illustrative comparison, and to characterize this random-homogeneous state, we have also plotted in figure 13, the following output. At each time step (here the time step is set $dt=0.01 ms$), a spike shaped by analogy with the typical  HH-V spike is generated with a probability $p$. We choose the probability $p$ to solve the following equation:
\[\frac{1}{dt}\times 100\times p \simeq 11.48\times 375+48.48\times 125.\]
 Which means that the mean number of spikes generated by this simple stochastic iterative equation is equal to the men number of spikes in the network. For illustrative purpose, the signal is then divided by $N=500$ to obtain the mean value per neuron. We observe that the two plots are similar, which suggests that the appellation 'randomly-homogeneous' is relevant.\\
 \textbf{Biological interpretation} \\
 From a Neuroscience point of view, this behavior is consistent with the so-called background activity. Note that the coupling has dramatically decreased the number of spikes per neuron, in comparison with simulations done in the previous section with neurons stimulated only by Poissonian drives. This emphasizes the inhibitory effect of I-neurons in the network dynamics: for the parameters considered here, the recurrent inputs have a global inhibitory effect.

\subsubsection{Partial Synchronization}
As $S^{EE}$ is increased, the synchronization phenomenon emerges.  By synchronization, we refer here to a state of the network where all the neurons of the network will fire within a short interval of time, see rasterplot of figure 14, first row and last column. Between the random-homogeneous state and the state of synchronization, a partial synchronization is observable: i.e. a state in which only a portion of the population will fire during an identified event. 
In this paragraph, we describe this state which is typically observed, for the following values of parameters:
\[S^{EI}=S^{IE}=S^{II}=0.01 \mbox{ and }S^{EE}=0.017\]
The rasterplot, for these parameters is reported in figure  14, first row, second column. The observation of the rasterplot has to be combined with other figures. In figure 15, second row, second column, we report the number of $E$ and $I$ spikes occurring in the time interval $[425,450]$ which corresponds to an identified event. We observe that only around $190$ spikes from $E-$neurons have been recorded  during this interval. Around $190$ spikes from $I-$neurons have also been recorded during the considered interval (some $I-$neurons have therefore spiked several times).  The panel in the second row, second column of figure 14 shows moreover the time evolution of the $V-$ potential of a given E-neuron. It shows  2 spikes over the interval [400,500], even though 4 events are identified in the network dynamics.  The panel in the second row, second column of figure 15 shows a specific $I-$neuron  which spikes 6 times during the same interval. Finally, we refer also to figure 16 which provides a 3 dimensional visualization of the phenomenon. In this figure, neurons that spike during the time interval appear highlighted in comparison with those which do not spike.

\subsubsection{Synchronization}
When $S^{EE}$ is increased above, the synchronization occurs. We refer to columns 3 and 4 of figures  14 and 15 for observation of this state. These  columns correspond respectively to the following values of the parameters:
\[S^{EI}=S^{IE}=S^{II}=0.01 \mbox{ and }S^{EE}=0.02\]
and
\[S^{EI}=S^{IE}=S^{II}=0.01 \mbox{ and }S^{EE}=0.03\].
Rasterplots are the most illustrative representation for the aforementioned activity and are reported in the first line of figure 14. The first row of figure 15 illustrates that in this case a regime all the neurons will spike during a given event. See also  
the second row of figure 14 which illustrates a E-neuron which spikes  at each event. Synchronization in this context can be compared to excitation waves spreading over the whole network in short time intervals.
\subsection{$g_E$ and $g_I$ correlation}
The values of $g_E$ and $g_I$ appear to be clearly correlated. This can be observed in rows 3 in both figures 14 and 15.   Recall that, according to equations, for a given neuron,  $g_E$ results from the number of presynaptic $E$-spikes received, while $g_I$ results from the number of $I$-spikes received. Consequently, for a specific neuron, when an augmentation/decrease in both its presynaptic E and I-neurons spiking rate occur at the same time,  we  observe a correlation. This is typically the case during synchronization. It is also observed during partial synchronization, see figure 15, row 3, column 2. In our network, this correlation could be used to detect partial synchronization. This correlation found in networks of E and I neurons has been used in \cite{Amb-2020} to build a two dimensional able to describe typical wandering brain rhythms. Note that evidence of $g_E$ and $g_I$ correlation have intensevily been investigated in experiments, see for example \cite{Ata-2009,Oku-2008,Shu-2003,Tan-2004}.  
\subsection{Gamma frequency and neurons frequencies}
When partial synchronization and synchronization occur, events arise in the network at a rhythm (frequency) of $40 Hz$, see figure 14, first row, which is a typical oscillation in the gamma regime. Note that the frequency of the network is different from the frequency of each $E$ and $I$ neurons in the state of partial synchronization. Note also, that the variation of $S^{EE}$ has a strong effect on the frequency of individual $E$-neurons but limited effect on the frequency of I-neurons, see table 2 and third rows of figures 14 and  15. We must also recall here that each neuron if it was not connected would have a higher frequency: 60 for E-neurons and 84 for I-neurons. For the case considered here, the network activity pushes down the spiking activity of each neuron, and in some regimes allows a gamma rhythm oscillation for the network.  
\subsection{Waves of excitation}
We want to emphasize here, that each individual cell, if there was no connexions in the network, would be in a high frequency spiking state. The network activity drastically brings down this activity. In states of partial synchronization and synchronization, spikes  spreads trough the network thanks to the network connexions. It would be of dynamical interest to follow these paths of excitation through the network, for the specific topology considered here. For example, to compare with other topologies or continuous equations, see for example \cite{Amb-2009,Amb-2016}.

\subsection{I(t) and spikes}
As sketched in section 2,  the spiking activity of a specific neuron depends on the value of $I$. In the network model \eqref{eq:HH-Network}, the parameter $I$ corresponds to the currents induced by excitation and inhibition fluxes. This means that the dynamics of a specific neuron in the network are the same that the dynamics of a single neuron modeled by equation \eqref{eq:HH-single} with a corresponding $I(t)$ equal to the excitatory and inhibitory fluxes. According with that, we denote:
\[I(t)=g_E(t)(V_E-V)+g_I(t)(V_I-V).\]
This quantity is plotted in figures  14 and 15, row 5. Some qualitative properties of individual neurons are remarkable and may be seen as resulting for the corresponding $I(t)$.  For example, we remark that, individual neurons may exhibit a kind of mixed mode oscillations (see fig 14,15, rows 3 and 6). We also remark, that if $I$ is varied while the neuron has already started to spike, it has a little effect on the dynamics. We refer here for example to the column 3 of figure 14.  Just after the second spike of the $E$-neuron, the current $I$ is high, but this  does  not generate another spike. Note finally, that the shape of the  spike may be qualified as attractive and stable in the sense that it remains roughly the same at each spike. A study of the attractor of the non-homogeneous HH single equation would be of interest regarding this aspect. 
\subsection{Statistics of spikes}
Finally, an usual output to monitor the activity of such a network is the mean value of $E$ and $I$ spikes per second per neuron. We denote these outputs respectively by $Ess$ and $Iss$. We have reported those outputs in tables 2 to 5. We observe, that for the network and range of parameters  considered here:
\begin{itemize}
\item increasing $S^{EE}$ has a strong effect on $Ess$ but little effect on $Iss$.
\item increasing $S^{IE}$ has a strong effect on $Iss$ but little effect on $Ess$.
\item increasing $S^{EI}$ has a notable  effects both on $Iss$ and $Ess$.
\item increasing $S^{II}$ has a strong effect on $Iss$ but little effect on $Ess$.
\end{itemize}

\section{Conclusion}\label{Sec:concl-persp}
In this article, we have considered a network of inhibitory and excitatory neurons, where each cell was modeled by the HH equations. The topology chosen for the network was inspired by prior work on the visual cortex $V1$.  After reviewing and revisiting the dynamics of a single HH model, we have considered a single HH cell driven by an external Poissonian input. Our theoretical results and preliminary numerical calculations suggested a framework for studying the network by way of appropriate values for the parameters in the model. Continuing with this approach, we analyzed numerically the network by varying the coupling parameters $S^{ab}, a,b\in\{I,E\}$. From there, we identified $S^{EE}$ as the most effective parameter (in the range of parameters considered) for reaching synchronization, and illuminated the transition from random homogeneity towards partial synchronization and synchronization. We also illustrated emergent phenomena such as:  $g_E$ and $g_I$ correlation and gamma-oscillations. A striking point of our study is that rhythms emerge as a property of the network activity itself: for example, in the synchronization regime, the network is oscillating at a gamma rhythm of 40 hz, even though each individual cell features a natural oscillation frequency of 60 (for E-neurons) and 80 hz (for I) in isolation. Finally, we would like to discuss a few two-dimensional models of interest with respect to some aspects treated in the paper. In the homogeneous stochastic regime, every neuron within its class, plays the same role and any synchronized effect in time emerges. Therefore, the dynamics could be described by two units: one E neuron and one I neuron. To fit the network outputs, these two neurons should be fed with stochastic inputs corresponding to the one observed in the network and the outputs should match the inputs: this is a kind of fixed point problem. In the synchronized regime, the E and I conductances are highly correlated. In this case the idea would be to construct coupled equations which match inputs, outputs with strong interaction between the E and I conductances (in contrast with the previous case). We refer to \cite{Amb-2020}, for a two dimensional model of conductances (whithout membrane potential and ionic fluxes) related to this case. The partial synchronization regime would be modeled as the latter, with different input/outputs to match. 
\begin{table}
\begin{tabular}{|c|c|c|c|}
\hline
$S^{EE}$&Ess&Iss\\
\hline
0.001&10.35&48\\
\hline
0.01&11.4933&48.48\\
\hline
0.02&36.51&49.12\\
\hline
0.03&40.11&48.56\\
\hline
\end{tabular}
 \label{ta:SEE}
\caption{Variation of $S^{EE}$ and its effect on the mean value of $E$ and $I$-spikes per second per neuron.}
\end{table}

\begin{figure}
\begin{center}
 \label{fig:Nho-Mean}
\includegraphics[height=4cm,width=3.8cm]{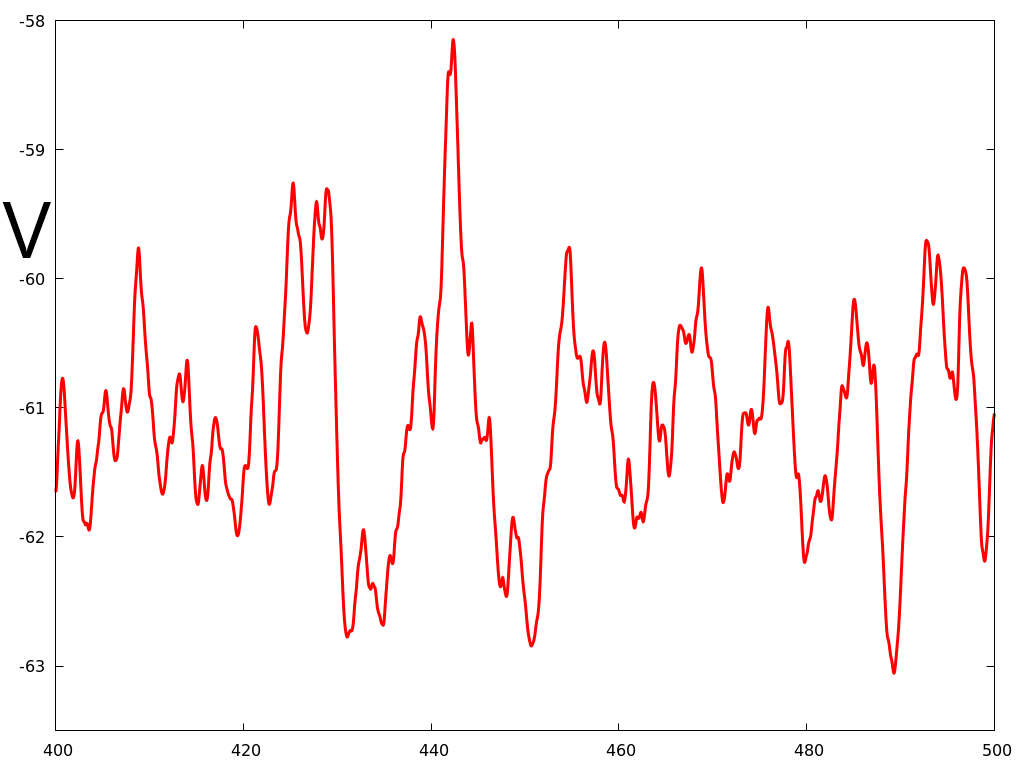}
\includegraphics[height=4cm,width=3.8cm]{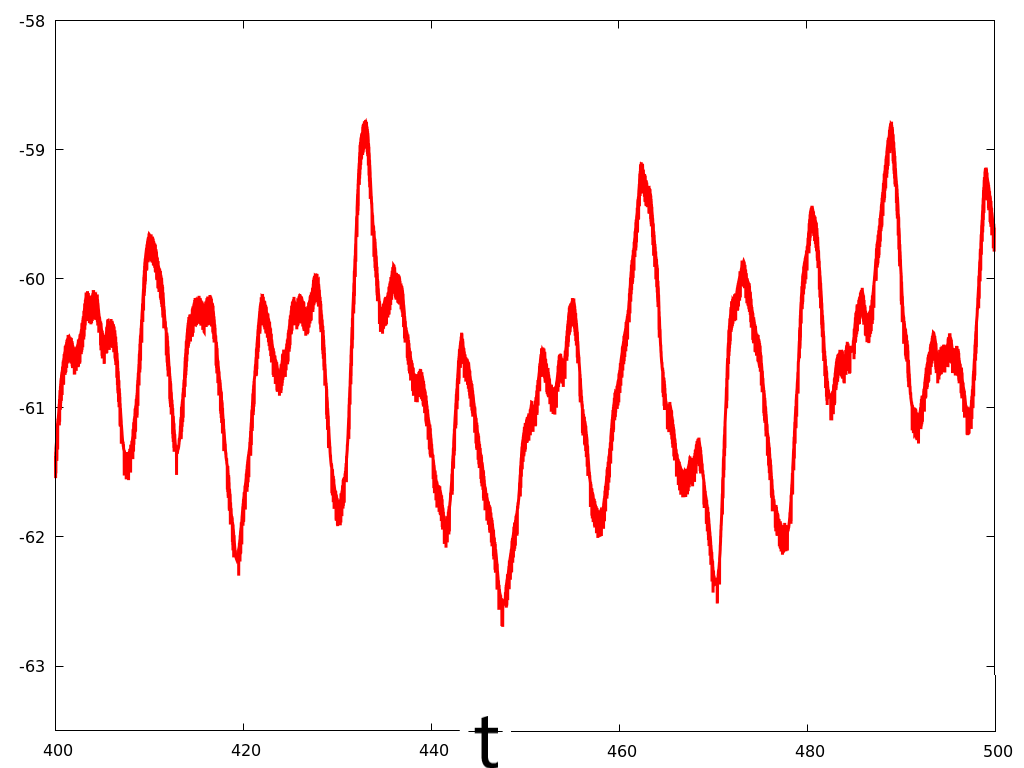}
\includegraphics[height=4cm,width=3.8cm]{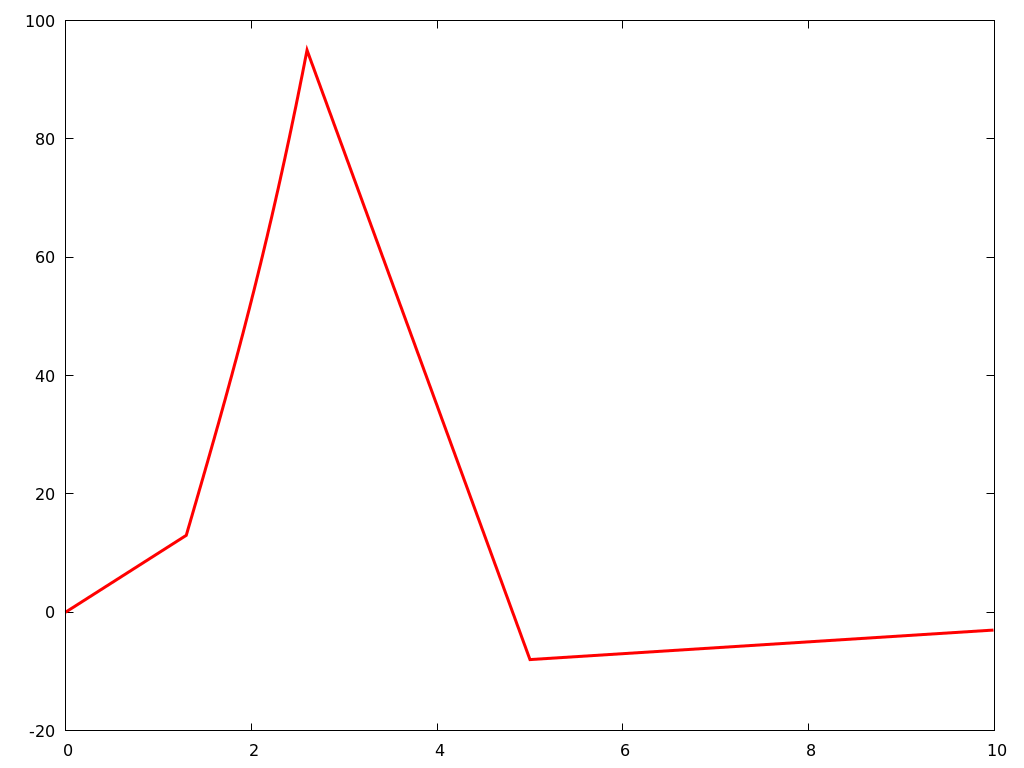}
\caption{Left panel: simulation of equation \eqref{eq:HH-Network} for $SII=SEI=SIE=SEE=0.01$. This plot shows the evolution of the mean value over $(V_i)_{i\in \{1,..,N\}}$ as a function of time. Right panel: for comparison, we have simulated a schematic temporal signal roughly mimicking the HH action potential over 10 ms. Middle panel: at each time, this signal was generated randomly with a mean of $10.375$ spikes per ms. This gives in mean, 10375 spikes over one second. This choice was made to approximate the $11.5\times 375+48.5\times 125=10375$ spikes occurring in the network, see table 2. Then, the signal was divided by 500 (to obtain a mean per neuron), and plotted.}
\end{center}
\end{figure}

\begin{figure}
\begin{center}
 \label{fig:NhotoSync1}
\includegraphics[height=2.5cm,width=2.9cm]{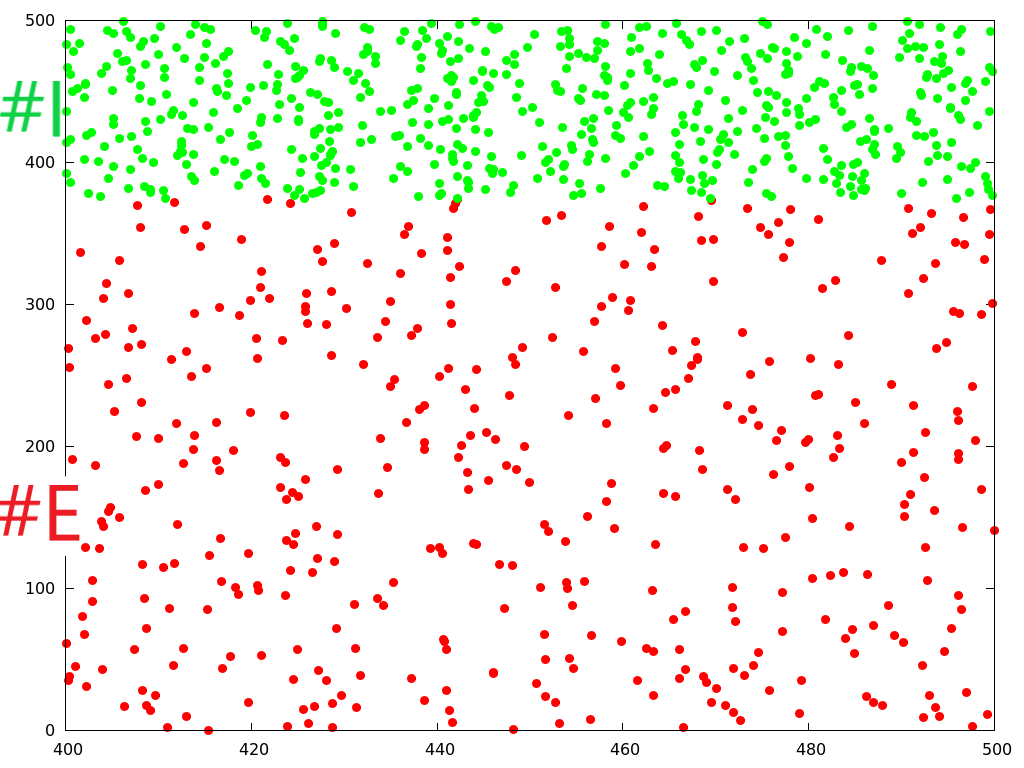}
\includegraphics[height=2.5cm,width=2.9cm]{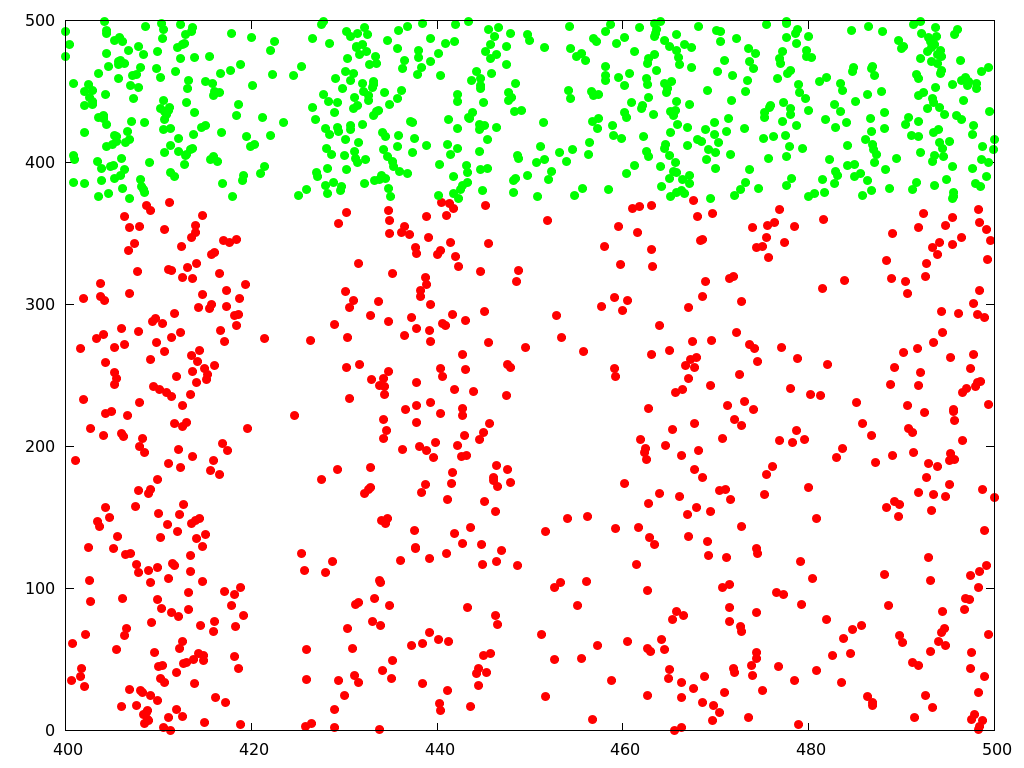}
\includegraphics[height=2.5cm,width=2.9cm]{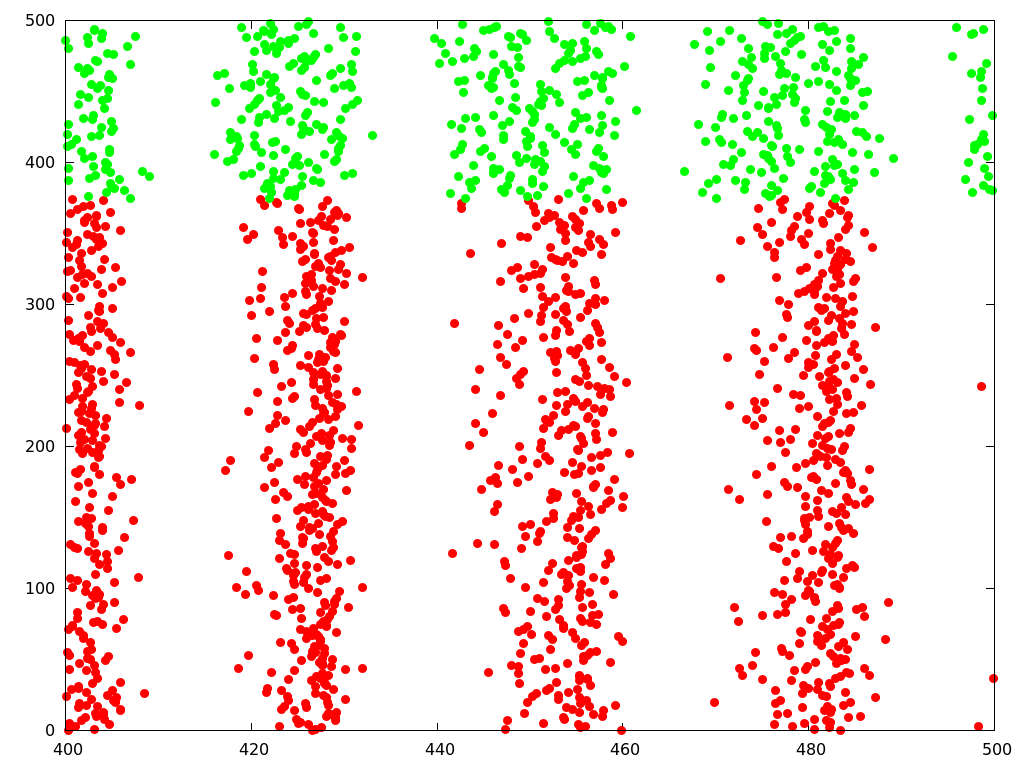}
\includegraphics[height=2.5cm,width=2.9cm]{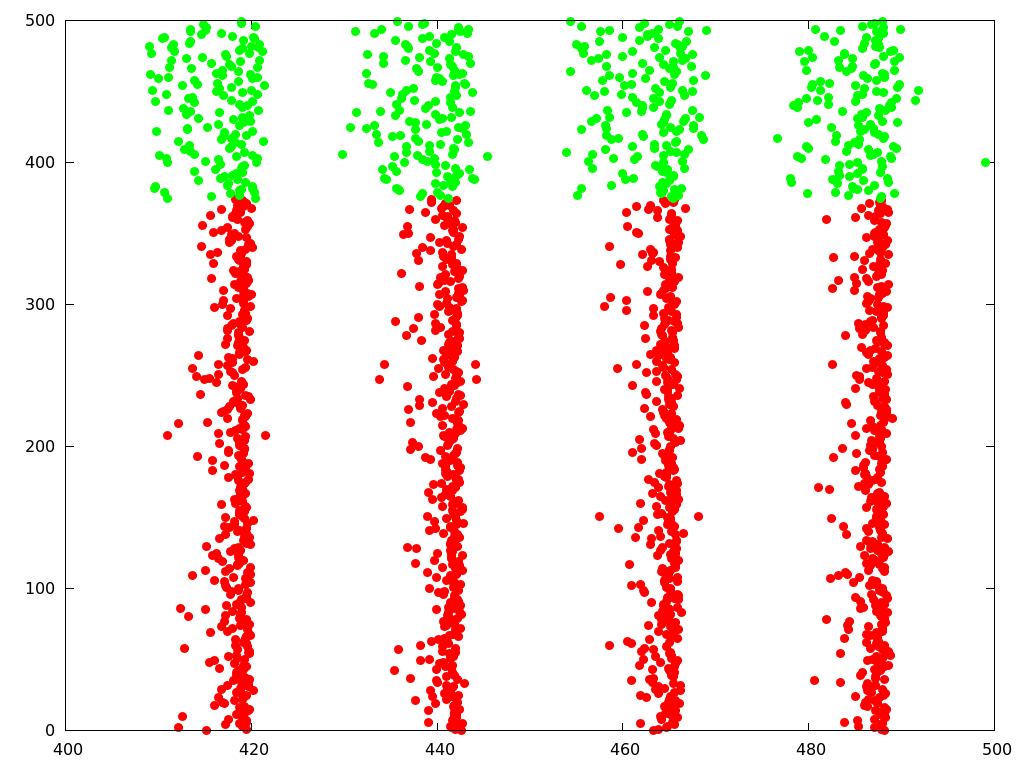}\\
\includegraphics[height=2.5cm,width=2.9cm]{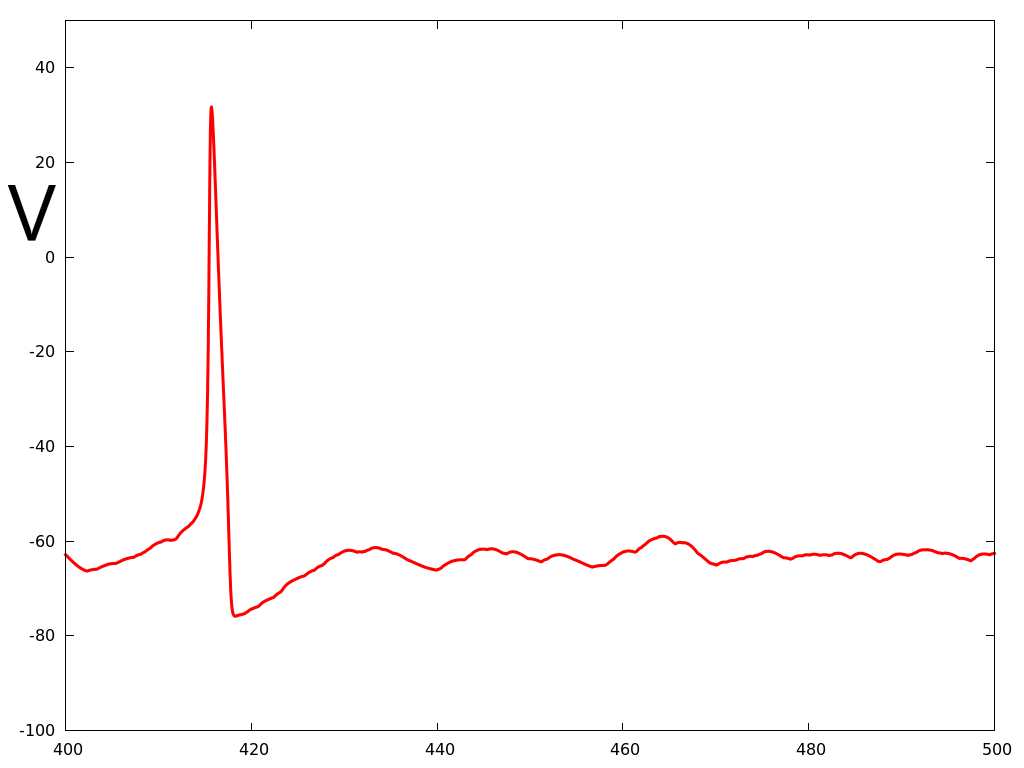}
\includegraphics[height=2.5cm,width=2.9cm]{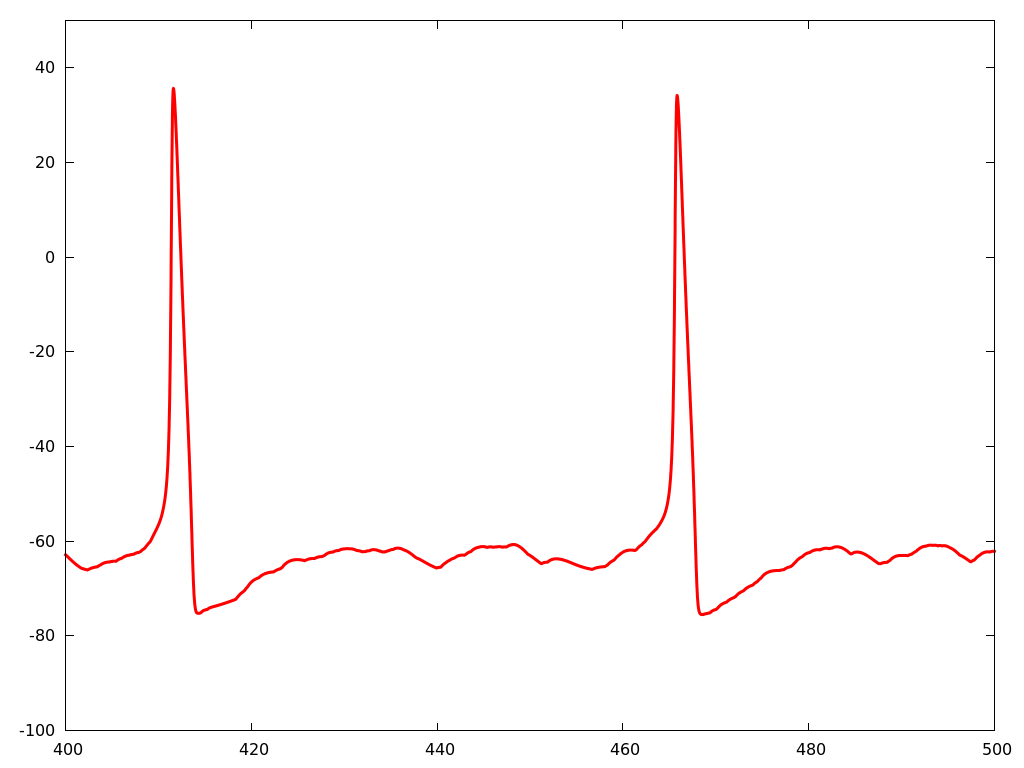}
\includegraphics[height=2.5cm,width=2.9cm]{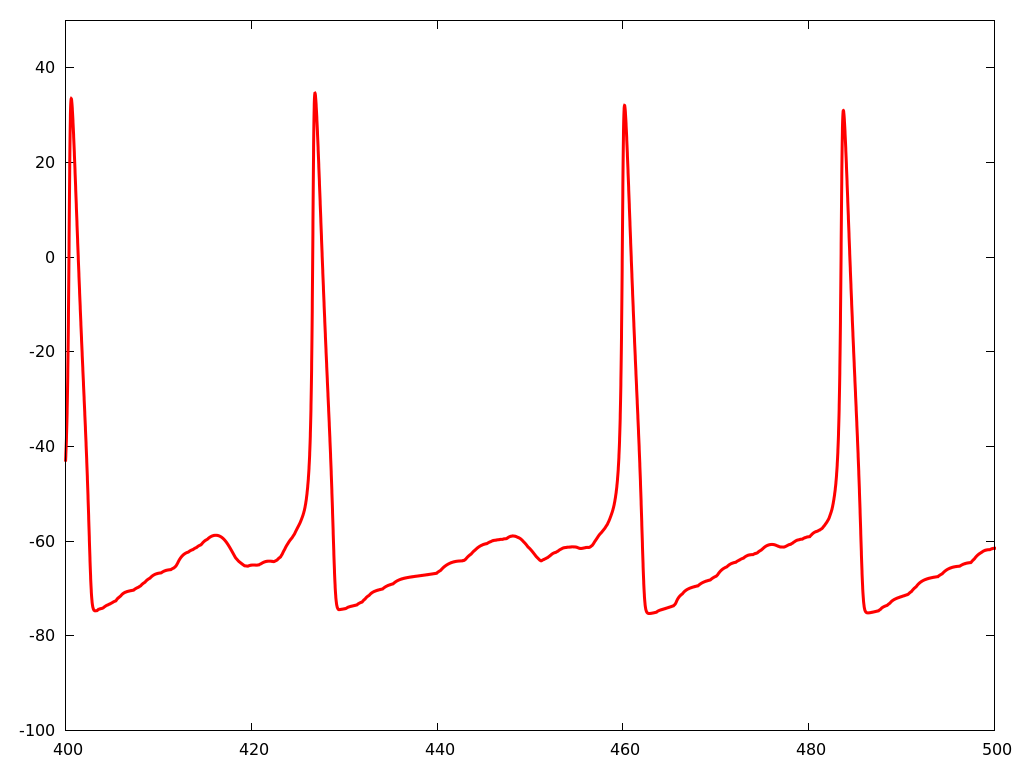}
\includegraphics[height=2.5cm,width=2.9cm]{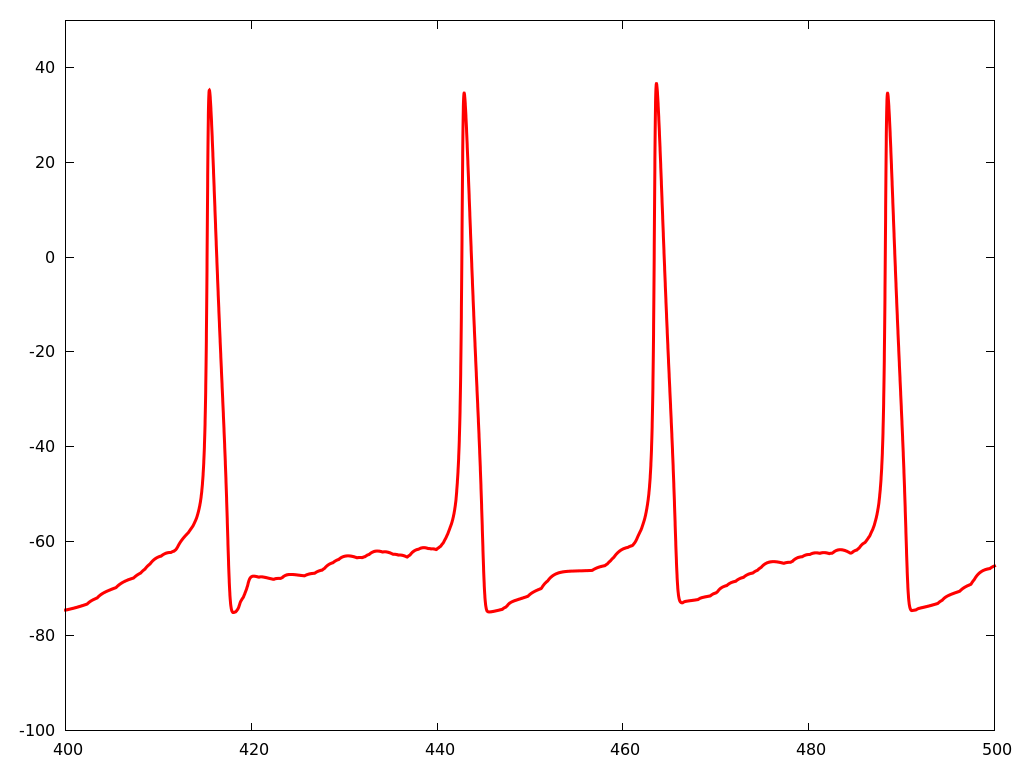}\\
\includegraphics[height=2.5cm,width=2.9cm]{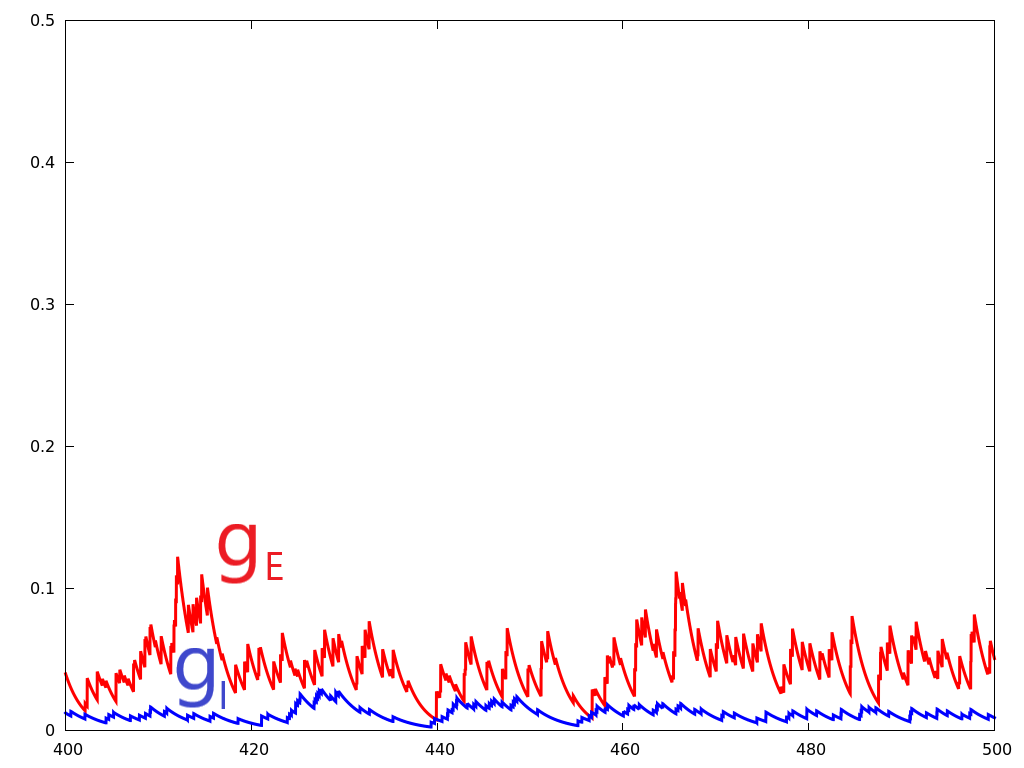}
\includegraphics[height=2.5cm,width=2.9cm]{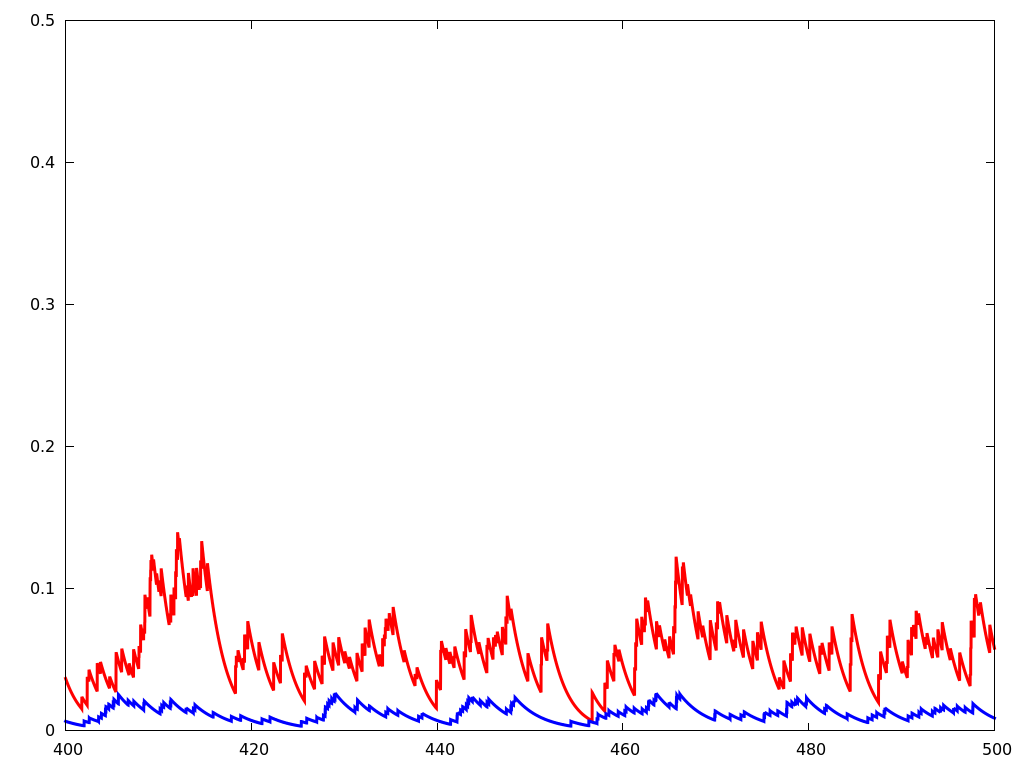}
\includegraphics[height=2.5cm,width=2.9cm]{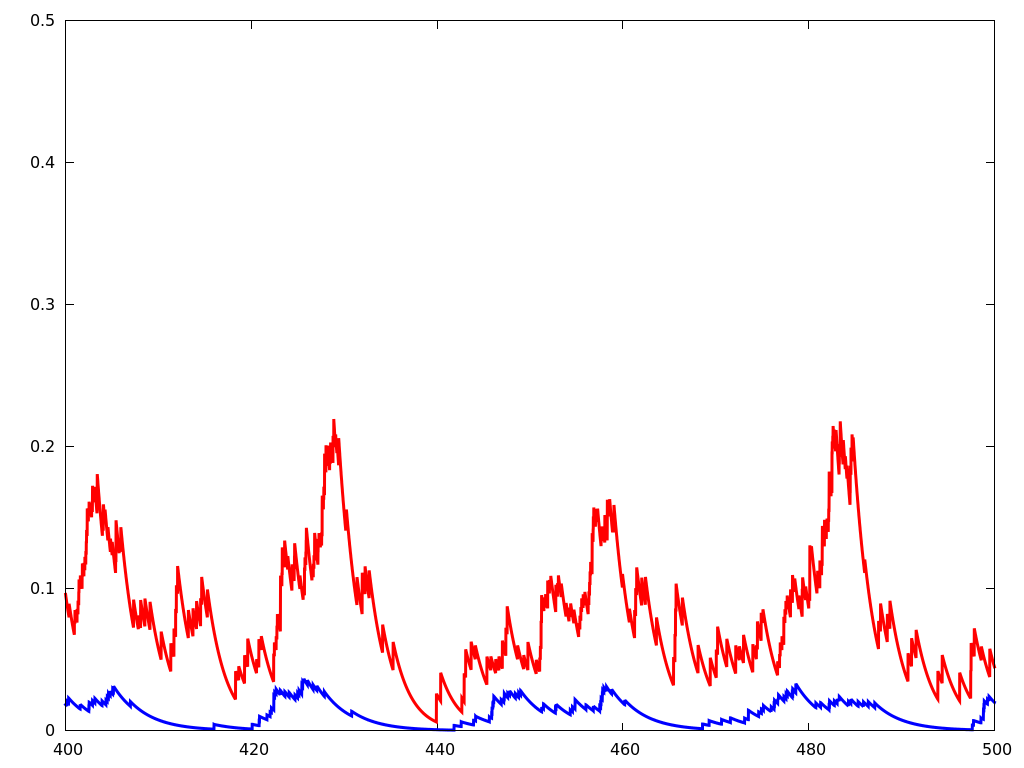}
\includegraphics[height=2.5cm,width=2.9cm]{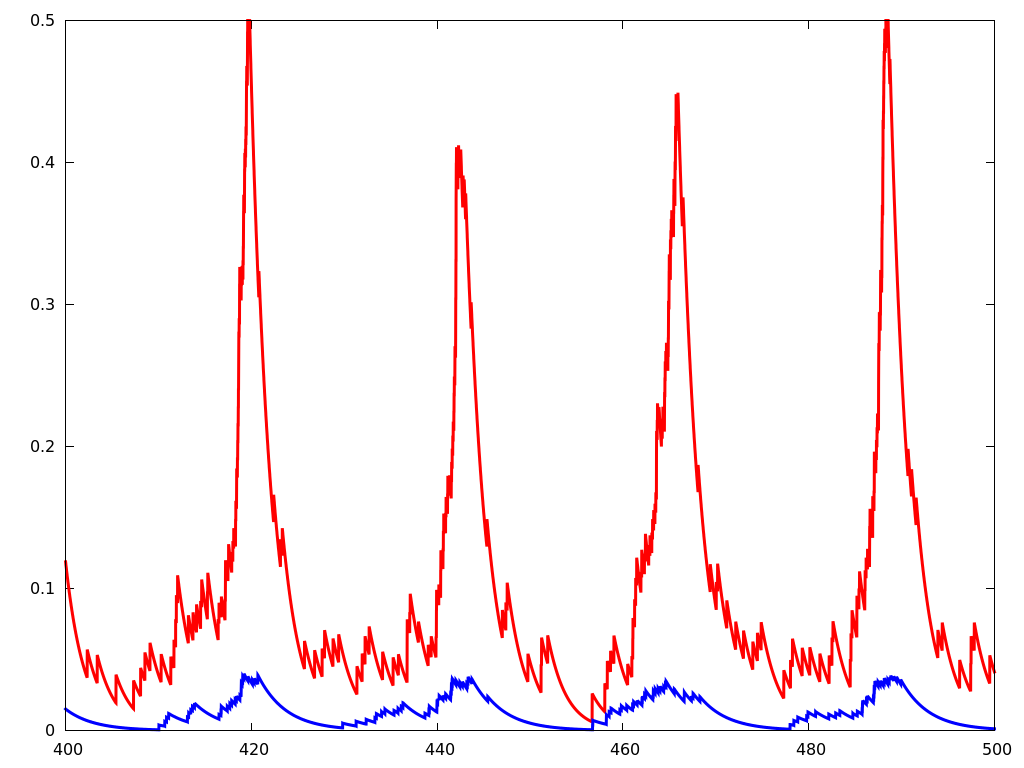}\\
\includegraphics[height=2.5cm,width=2.9cm]{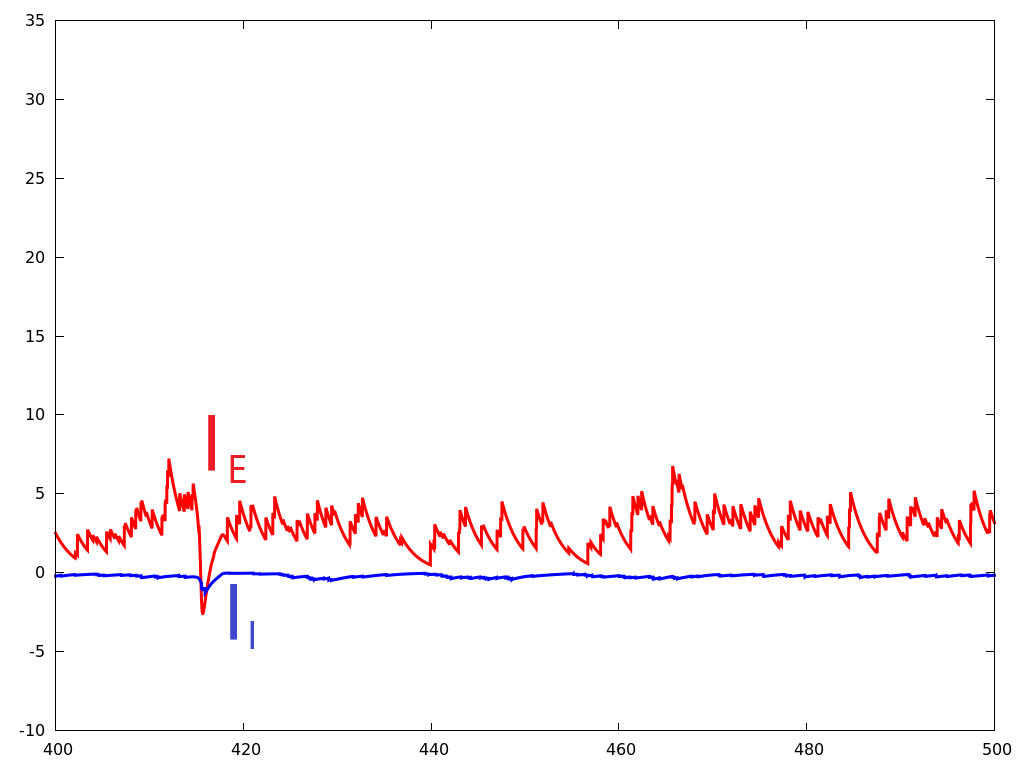}
\includegraphics[height=2.5cm,width=2.9cm]{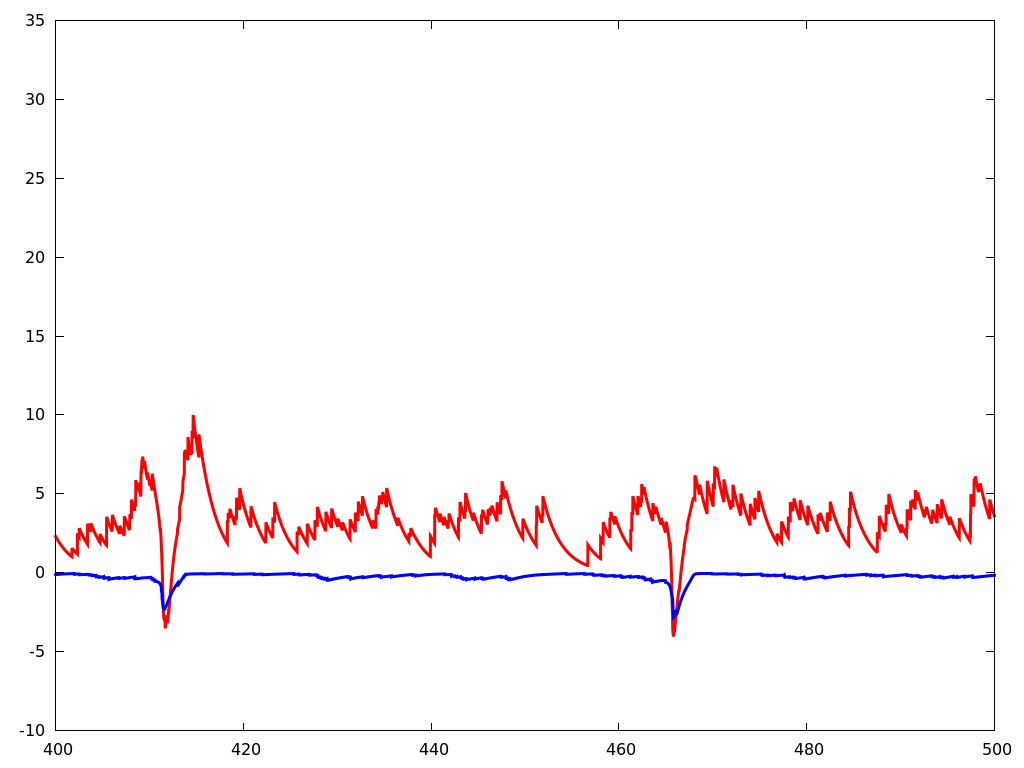}
\includegraphics[height=2.5cm,width=2.9cm]{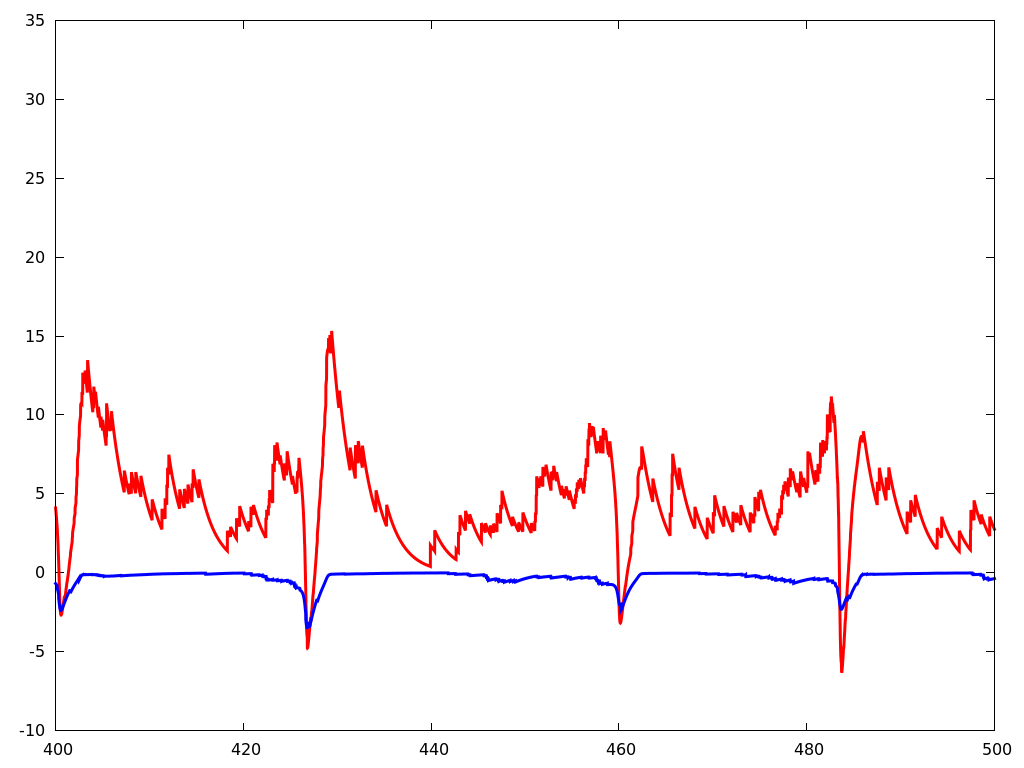}
\includegraphics[height=2.5cm,width=2.9cm]{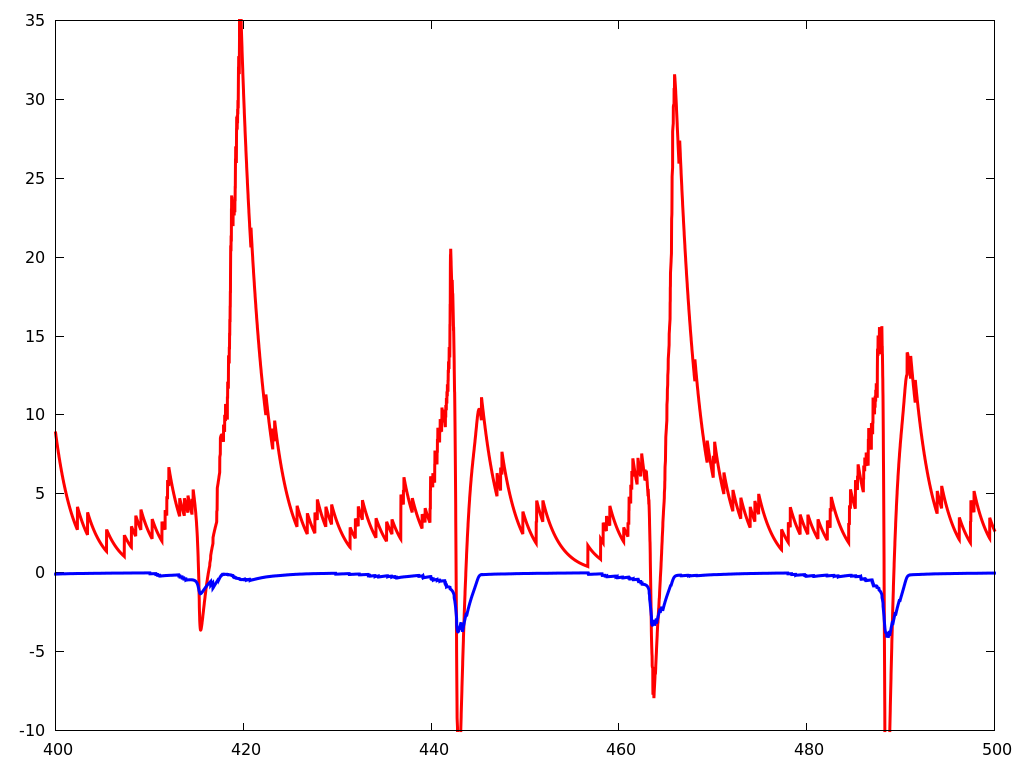}\\
\includegraphics[height=2.5cm,width=2.9cm]{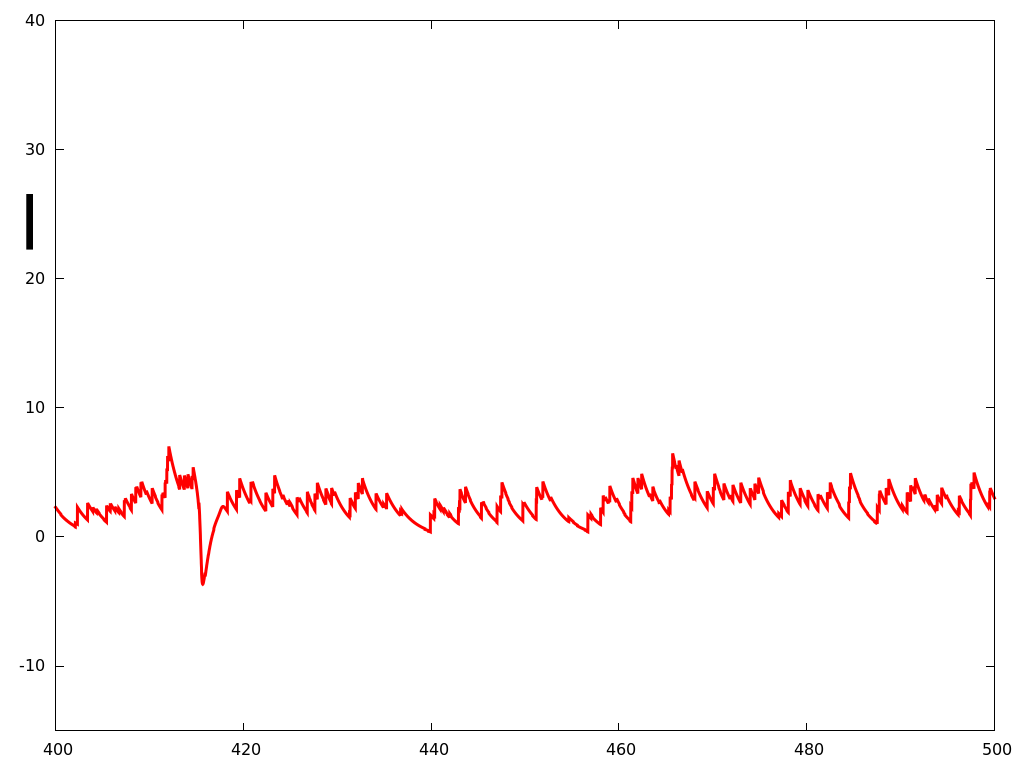}
\includegraphics[height=2.5cm,width=2.9cm]{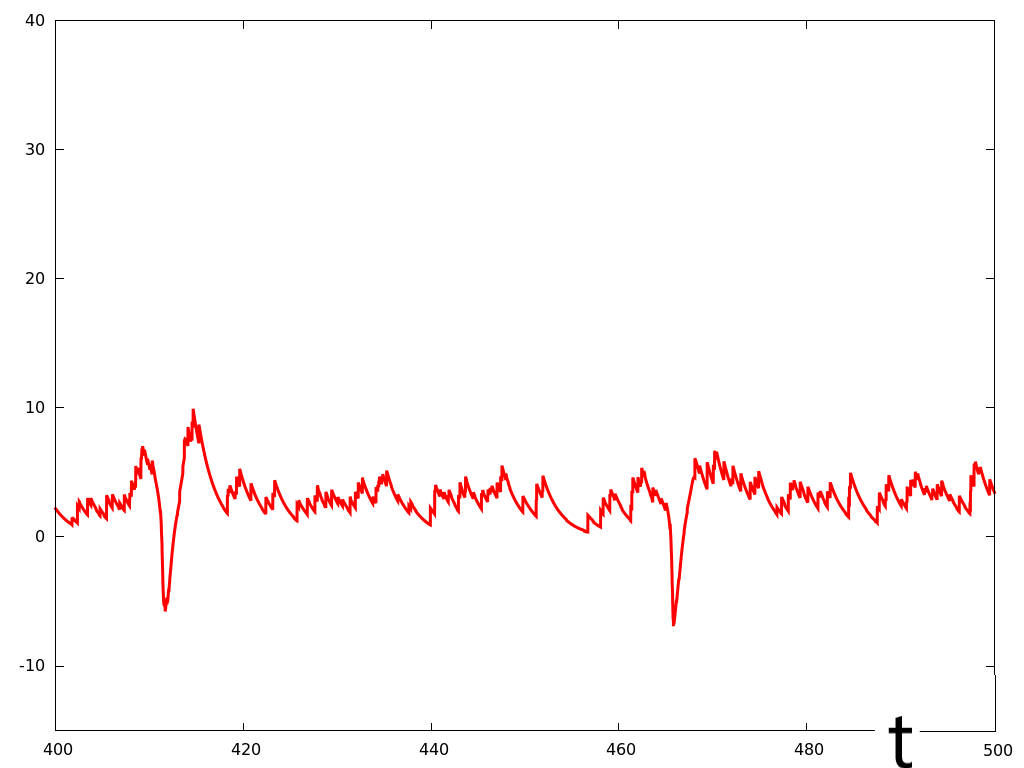}
\includegraphics[height=2.5cm,width=2.9cm]{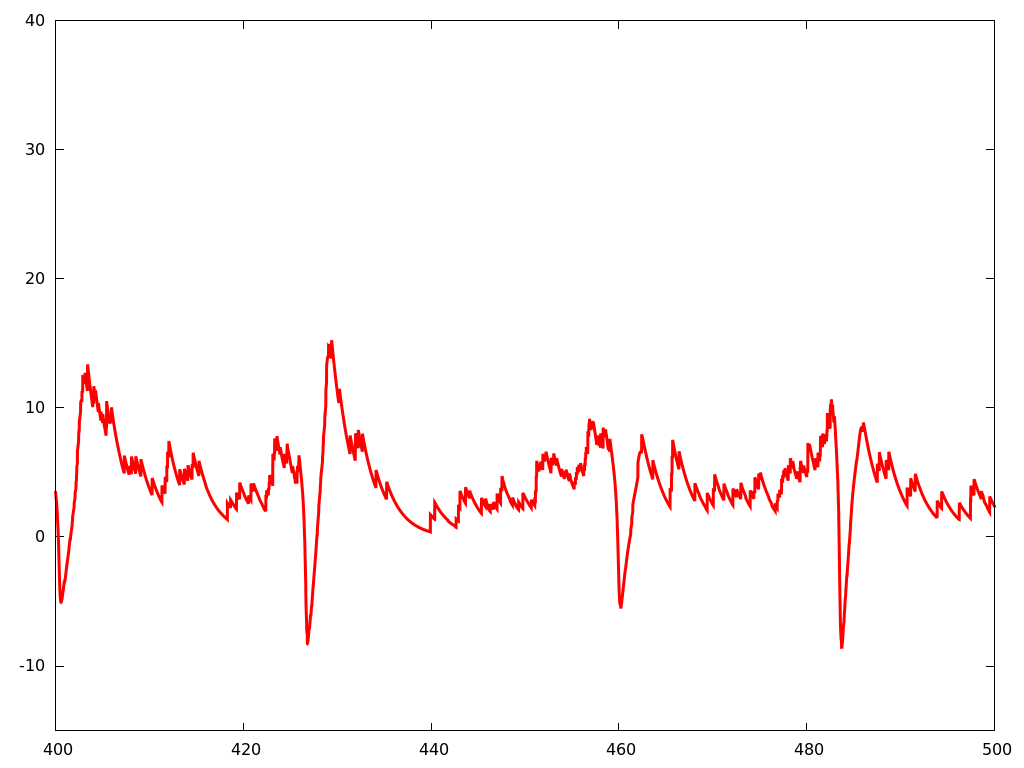}
\includegraphics[height=2.5cm,width=2.9cm]{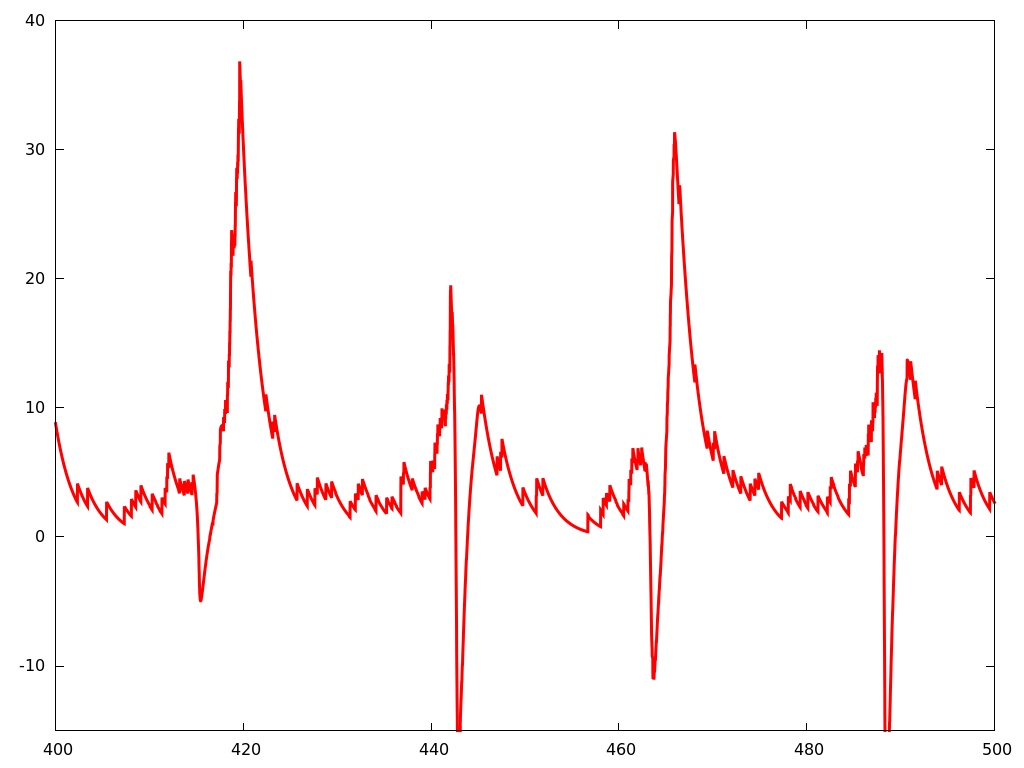}\\
\includegraphics[height=2.5cm,width=2.9cm]{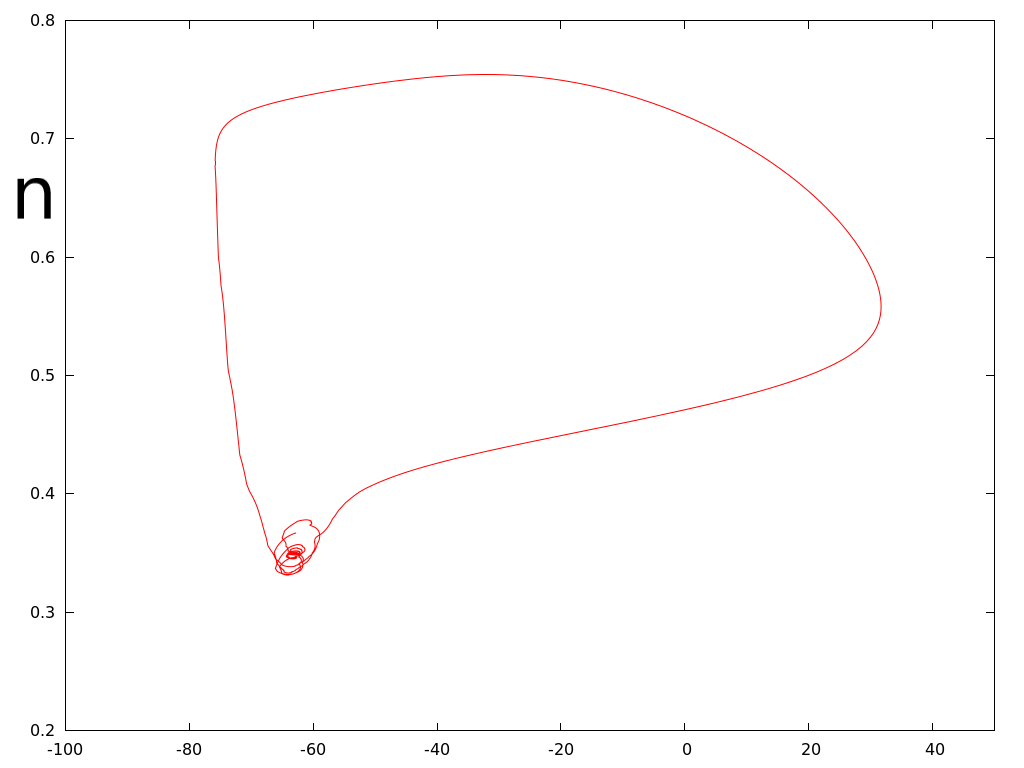}
\includegraphics[height=2.5cm,width=2.9cm]{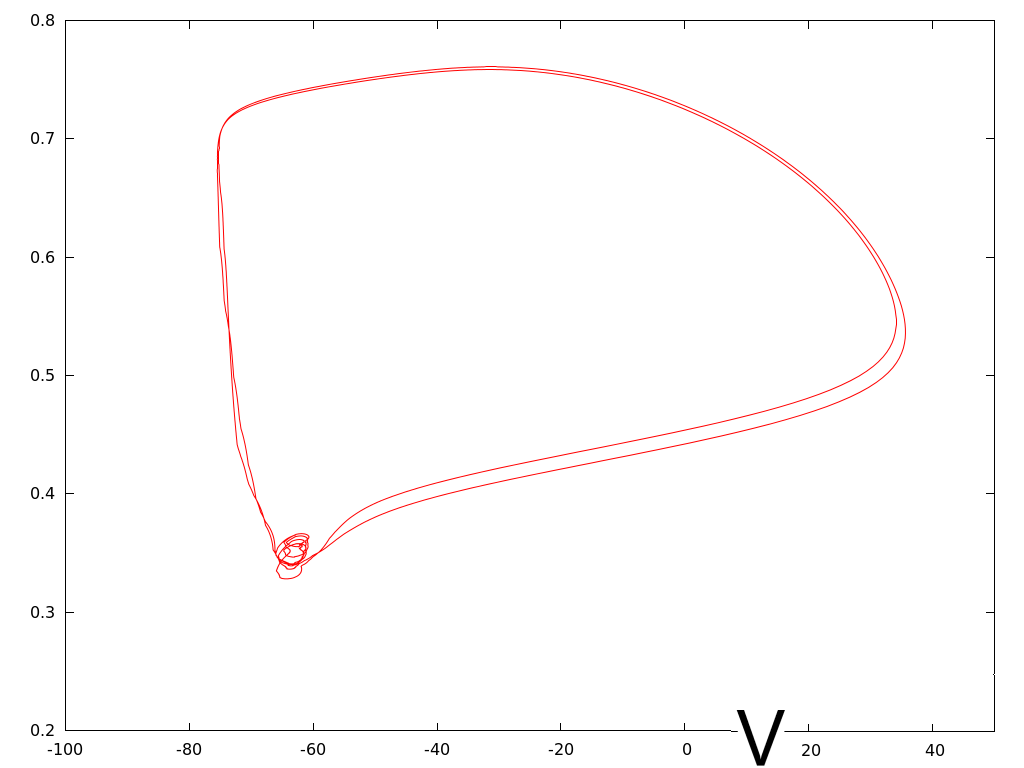}
\includegraphics[height=2.5cm,width=2.9cm]{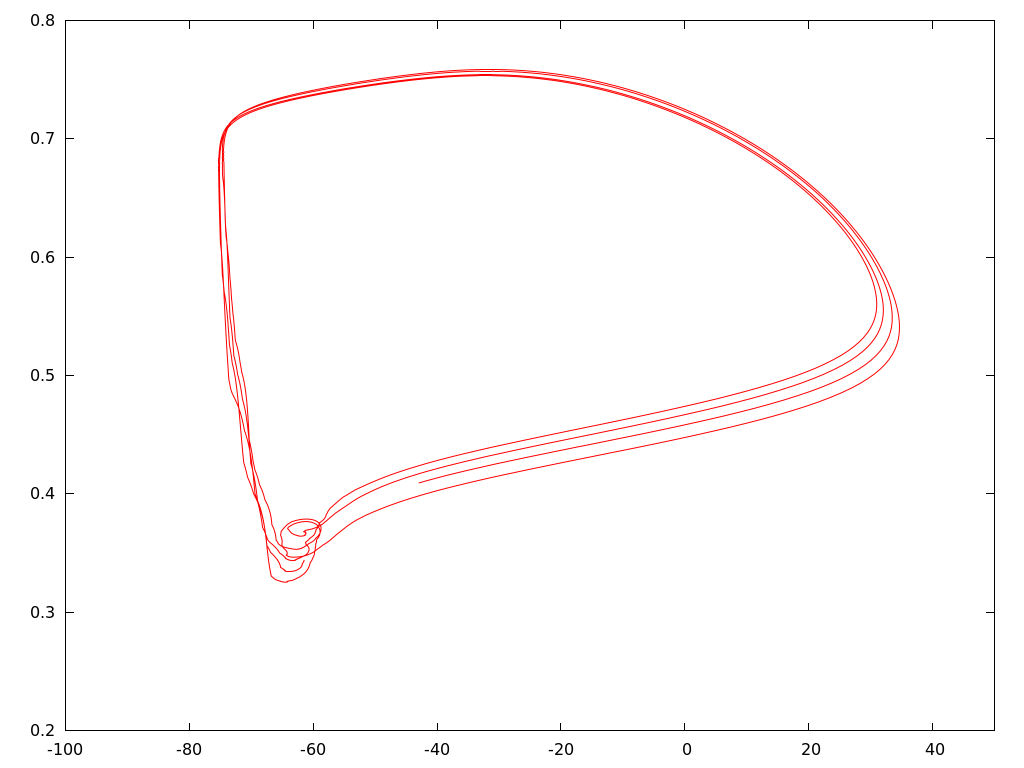}
\includegraphics[height=2.5cm,width=2.9cm]{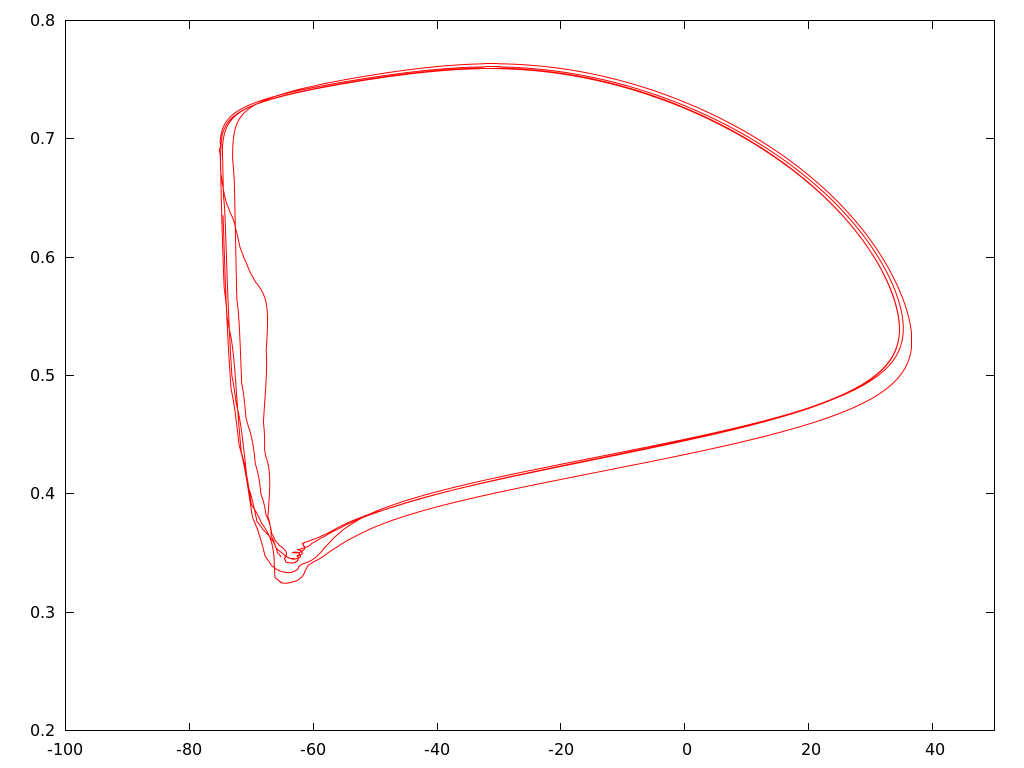}
\caption[width=.95\linewidth]{Simulation of system \eqref{eq:HH-Network}. This figure illustrates a path from random homogeneity to synchronization as the parameter $S^{EE}$ is increased. In this picture, the parameters $SII=SEI=SIE=0.01$ are set and each column from left to right corresponds to a specific value of $S^{EE}$. Respectively: $SEE=0.01, 0.017, 0.02$ and $0.03$.  The first row represent the rasterplot:  at each time the spiking neurons are represented by a point. On the top of the figure, in green, the $I$-neurons are plotted. $E$-neurons are plotted in red below. At left, the rasterplot illustrates a state where any event appears to be distinguishable. We call it random homogeneous activity. Increasing $S^{EE}$ induces synchronization. The second row represents the potential $V_1$ of neuron $\#1$ as a function of a time. The third row, the $E$-conductance $g_E$ in red and the $I$-conductance $g_I$ in blue for the same neuron.  The fourth row, the $E$-current denoted by $I_E$ in red and the $I$-current denoted by $I_I$ in blue. The fifth row, the sum of $E$ and $I$ currents which plays the role of $I(t)$ in a single HH equation.  The last row illustrates the projection of the trajectory of neuron $\#1$ in the $(V,n)$ phase space.}
\end{center}
\end{figure}

\begin{figure}
 \label{fig:NhotoSync2}
\begin{center}
\includegraphics[height=2.9cm,width=2.9cm]{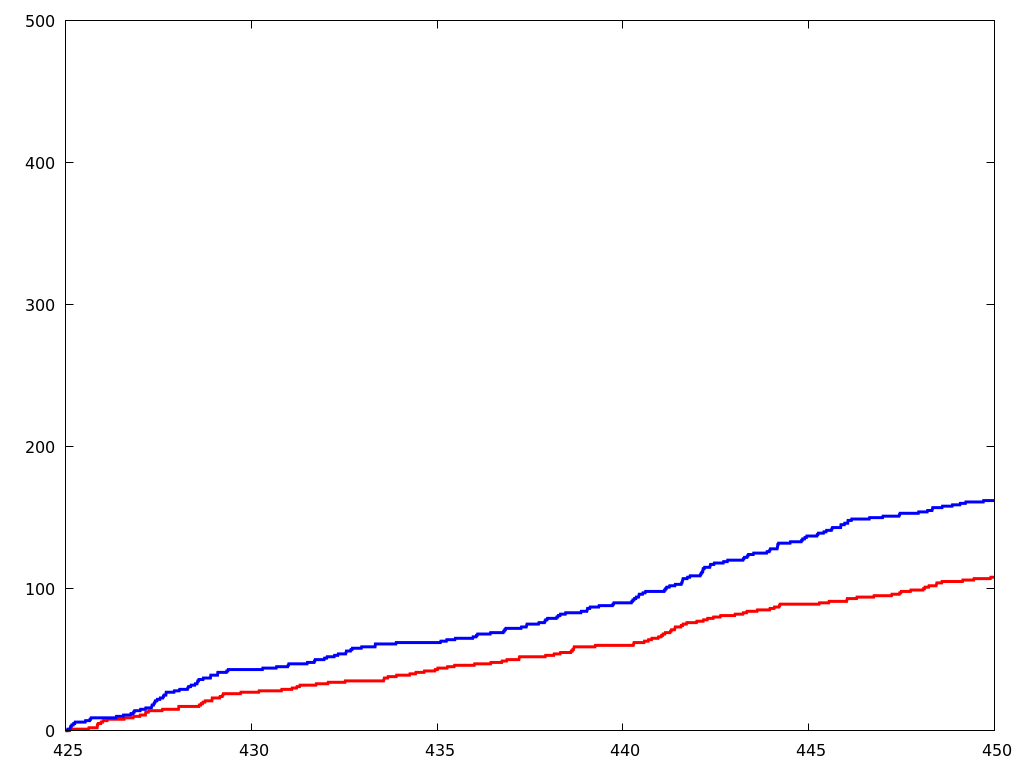}
\includegraphics[height=2.9cm,width=2.9cm]{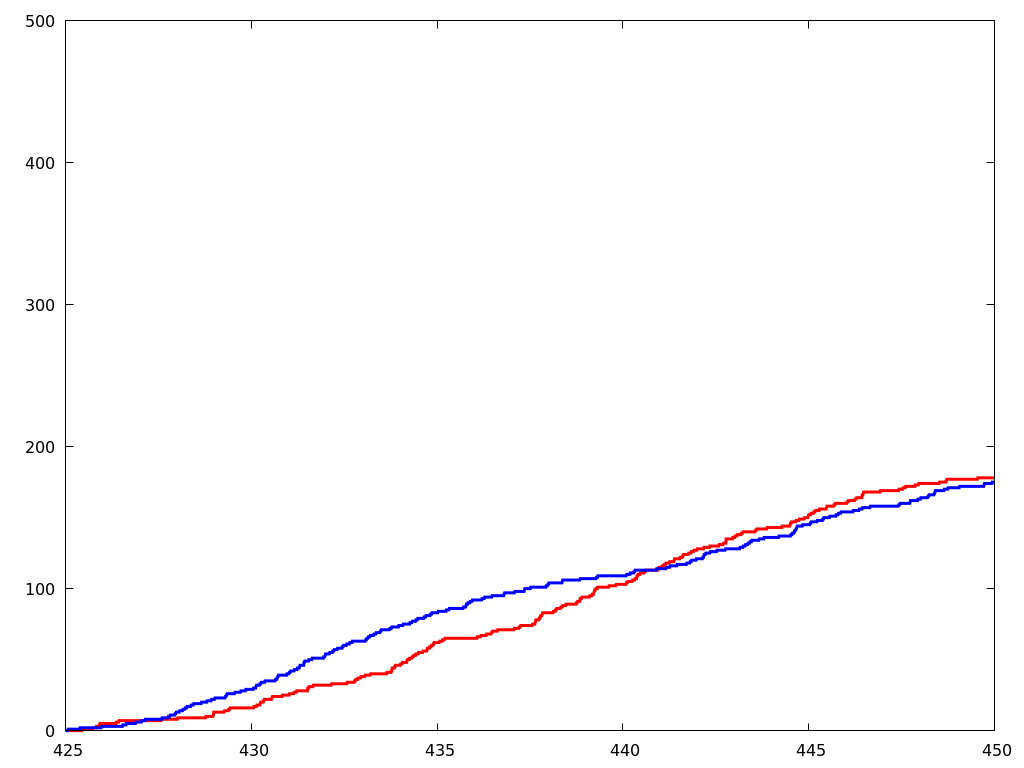}
\includegraphics[height=2.9cm,width=2.9cm]{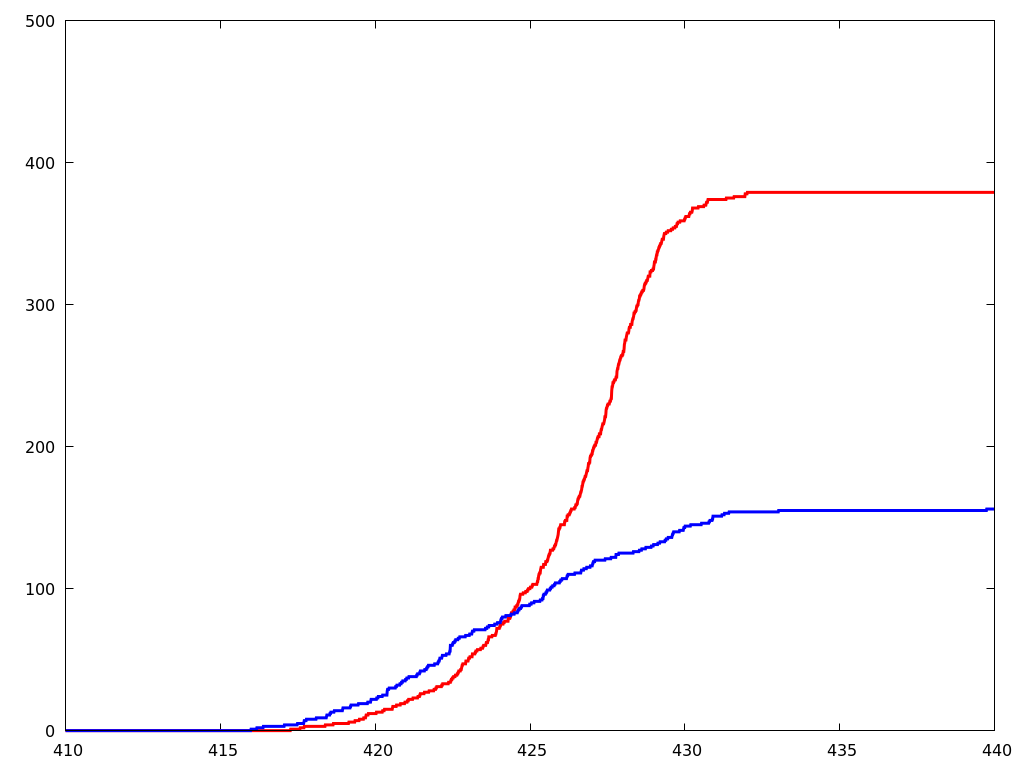}
\includegraphics[height=2.9cm,width=2.9cm]{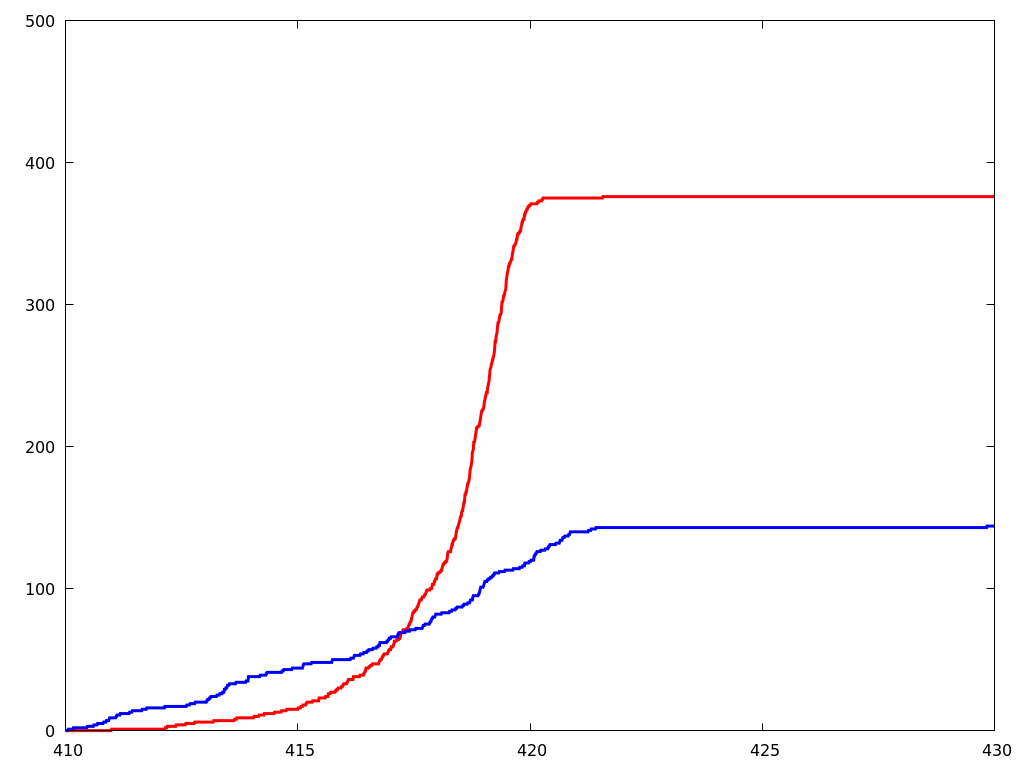}\\
\includegraphics[height=2.9cm,width=2.9cm]{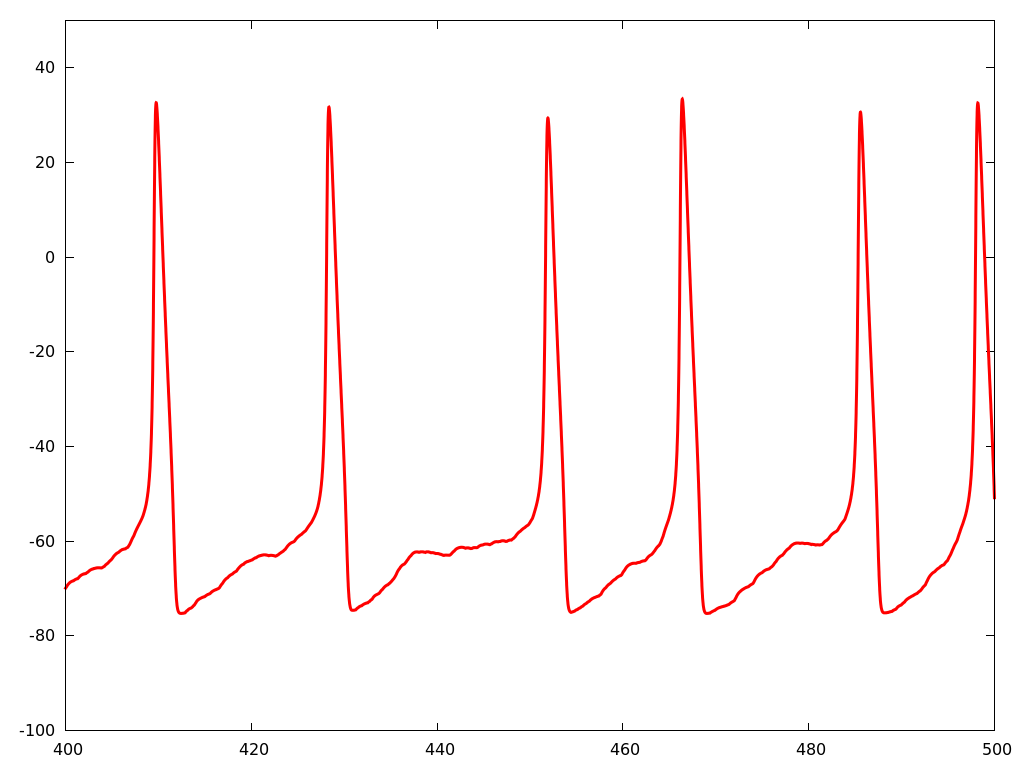}
\includegraphics[height=2.9cm,width=2.9cm]{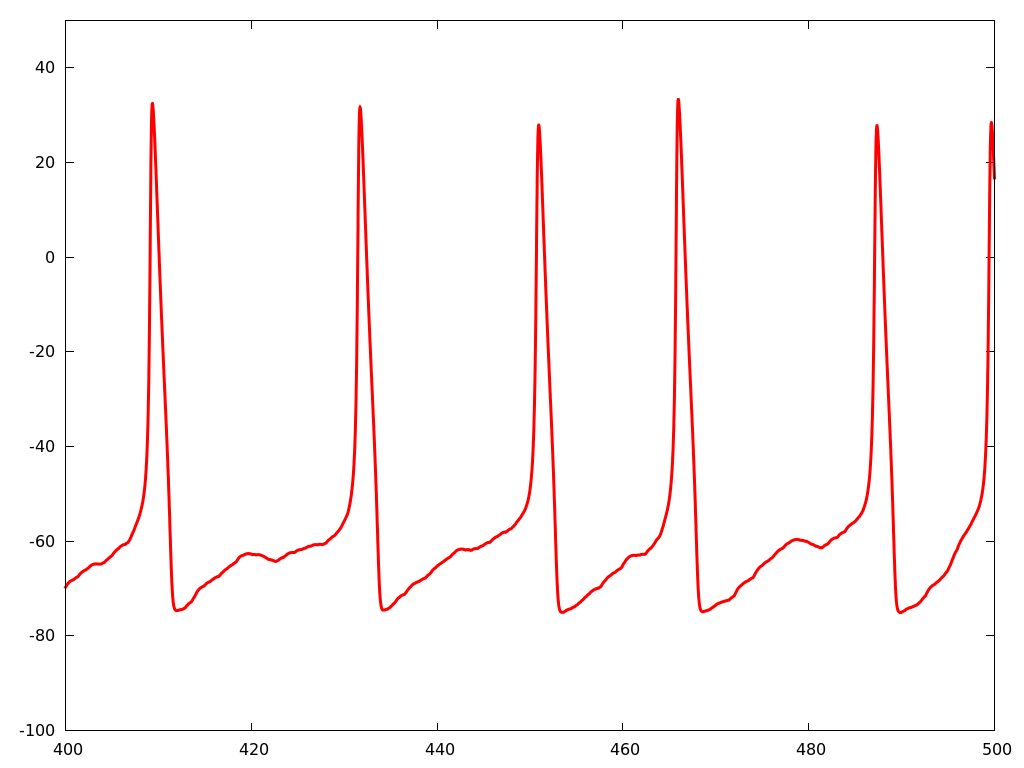}
\includegraphics[height=2.9cm,width=2.9cm]{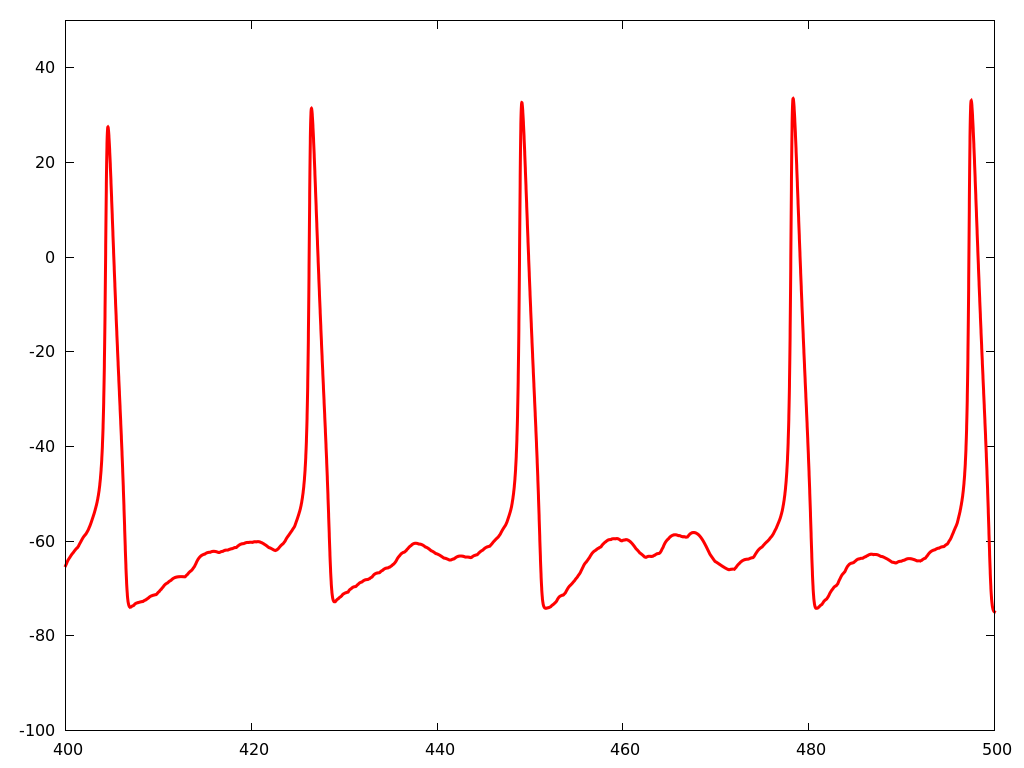}
\includegraphics[height=2.9cm,width=2.9cm]{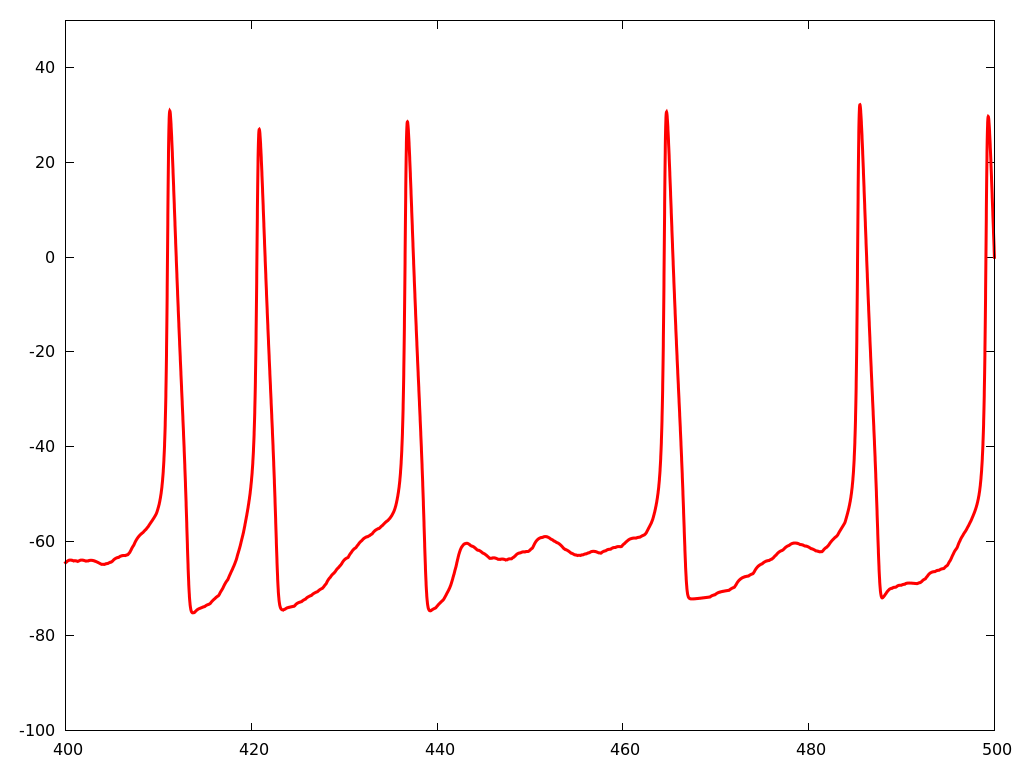}\\
\includegraphics[height=2.9cm,width=2.9cm]{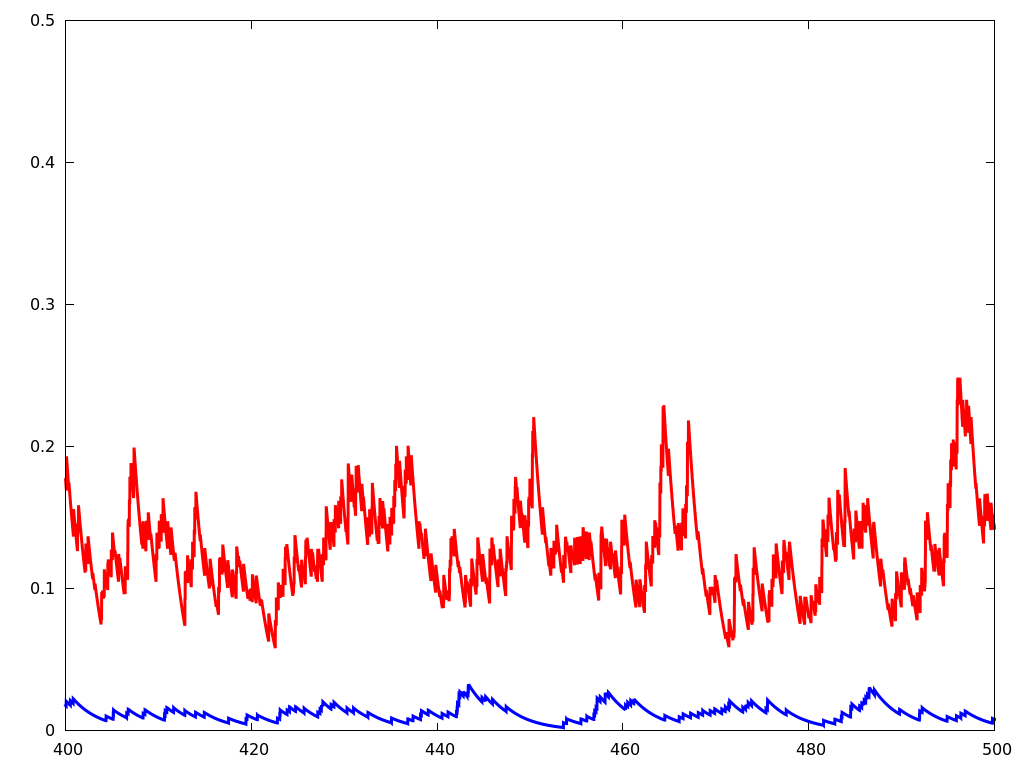}
\includegraphics[height=2.9cm,width=2.9cm]{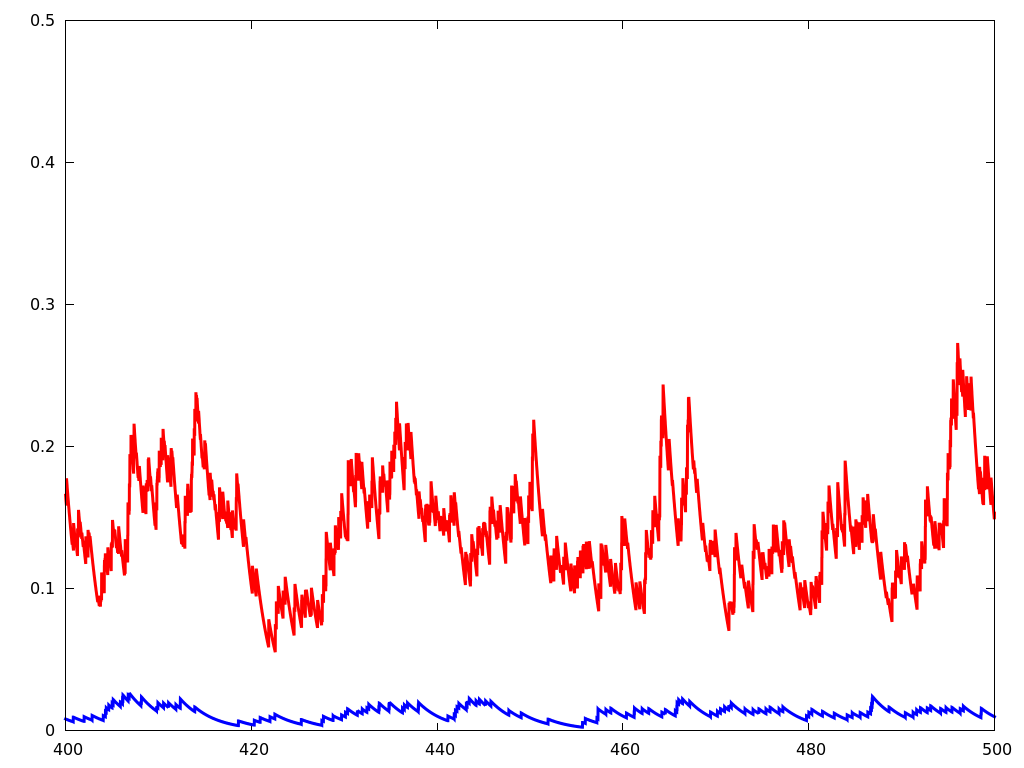}
\includegraphics[height=2.9cm,width=2.9cm]{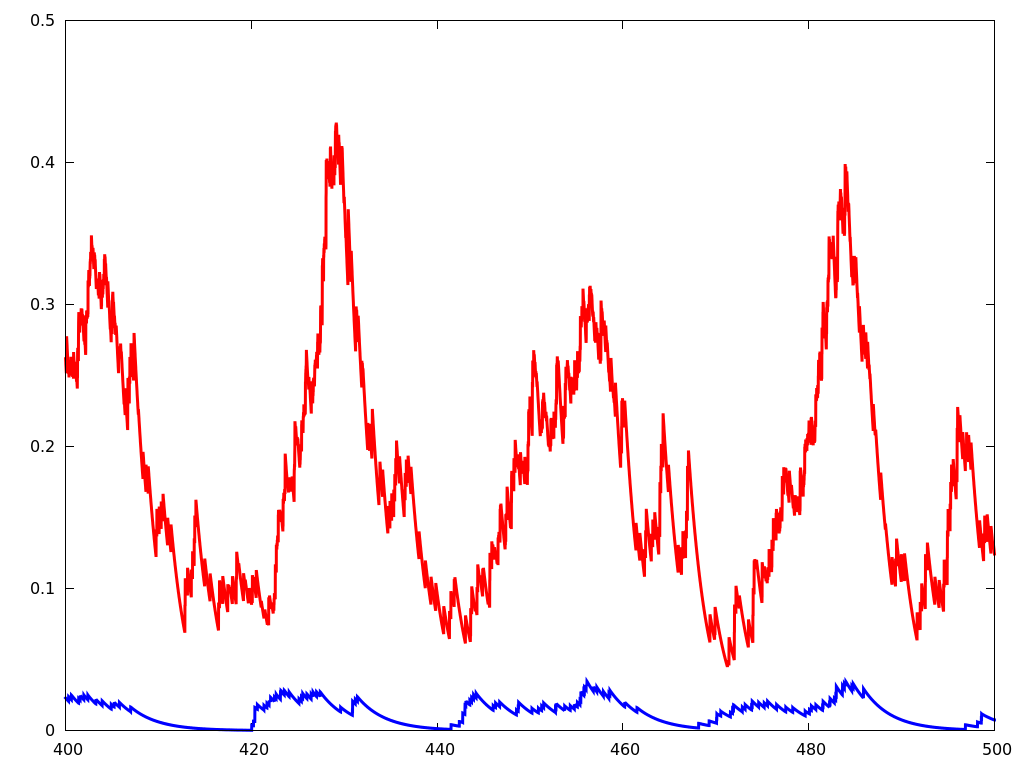}
\includegraphics[height=2.9cm,width=2.9cm]{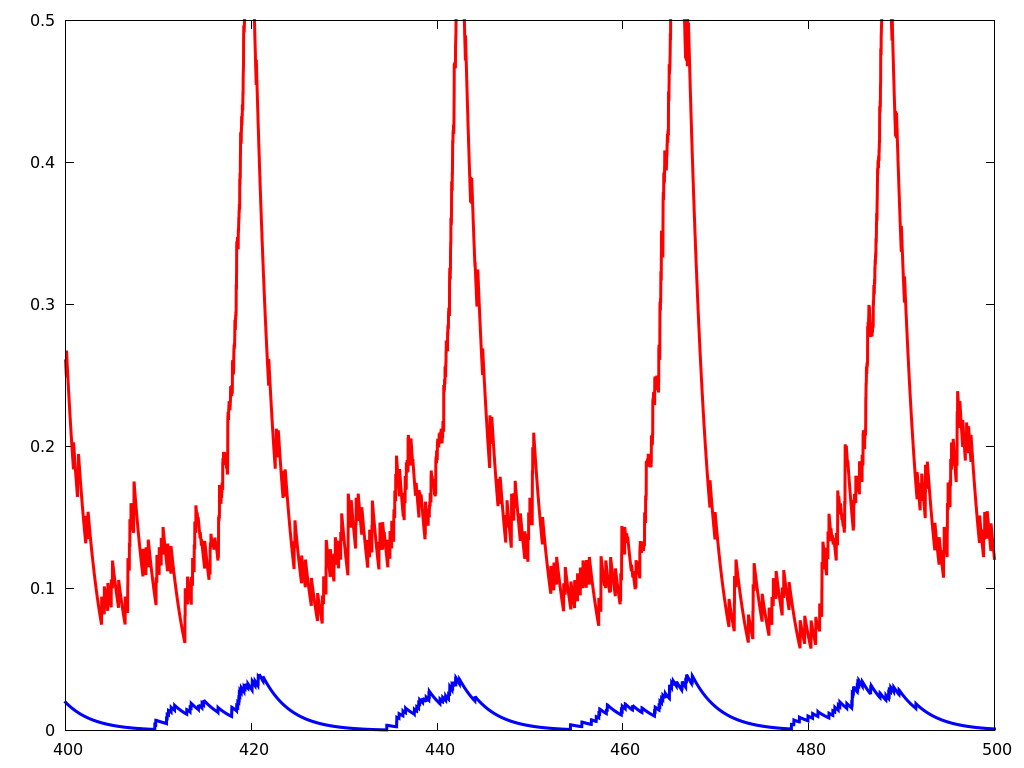}\\
\includegraphics[height=2.9cm,width=2.9cm]{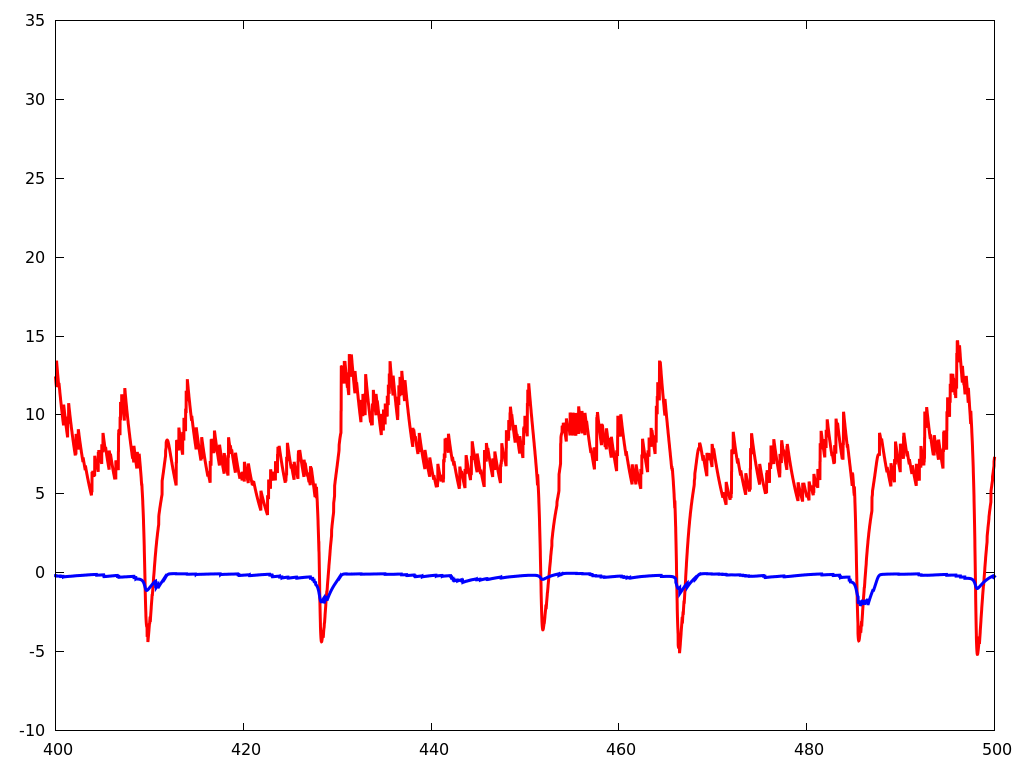}
\includegraphics[height=2.9cm,width=2.9cm]{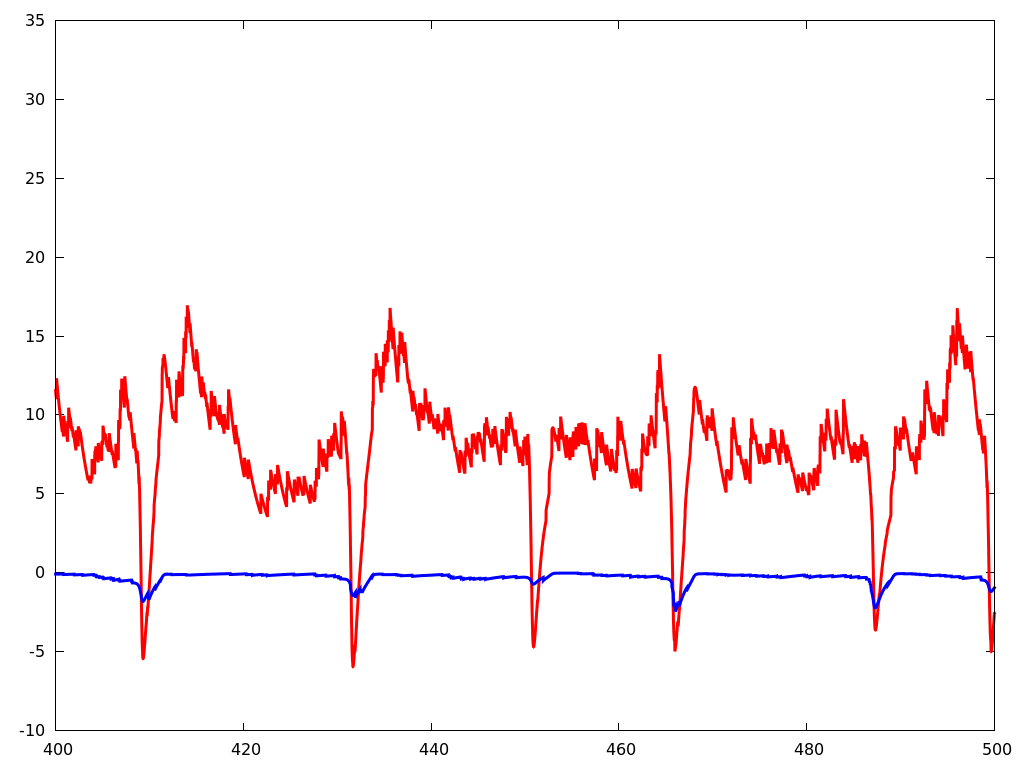}
\includegraphics[height=2.9cm,width=2.9cm]{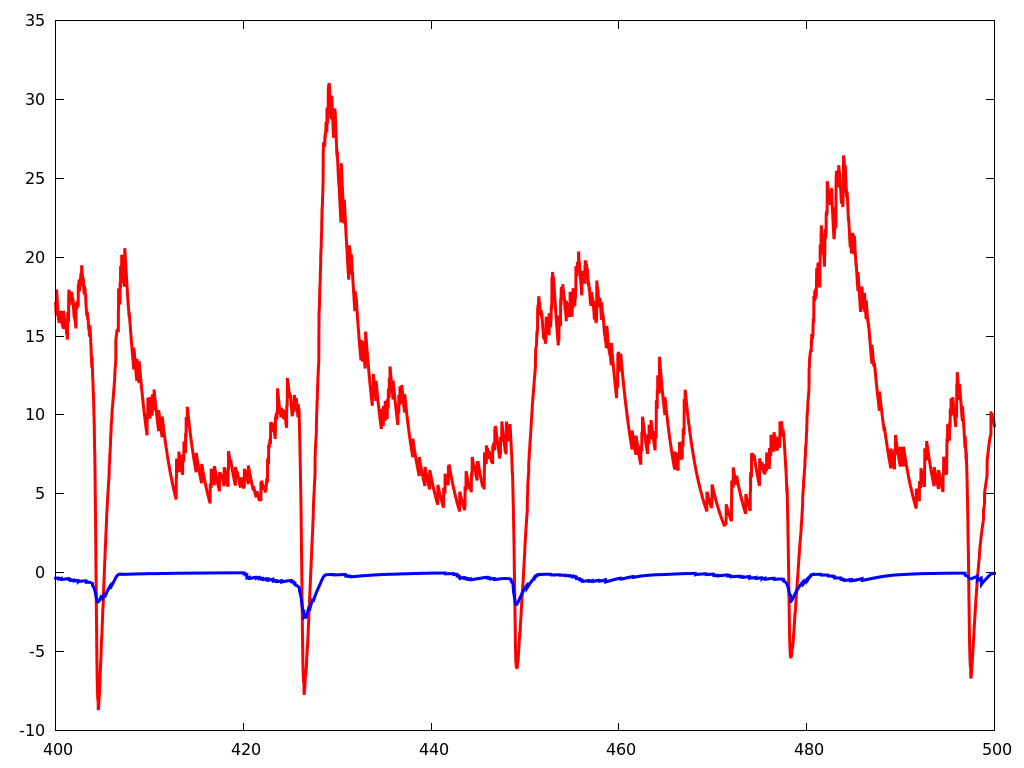}
\includegraphics[height=2.9cm,width=2.9cm]{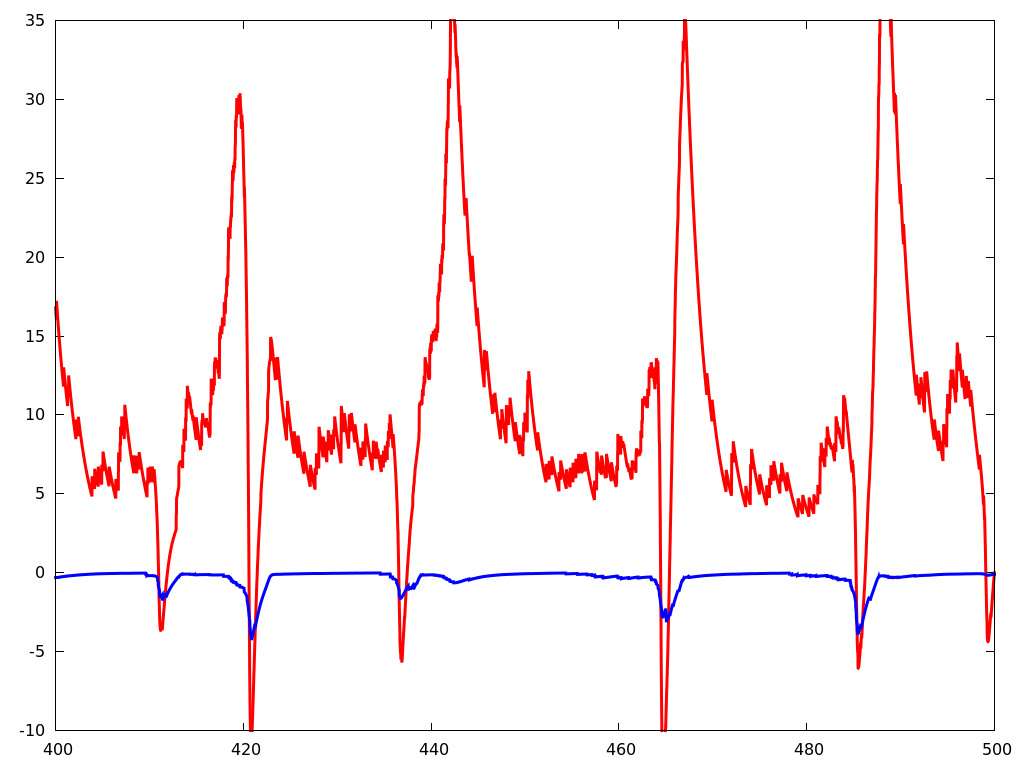}\\
\includegraphics[height=2.9cm,width=2.9cm]{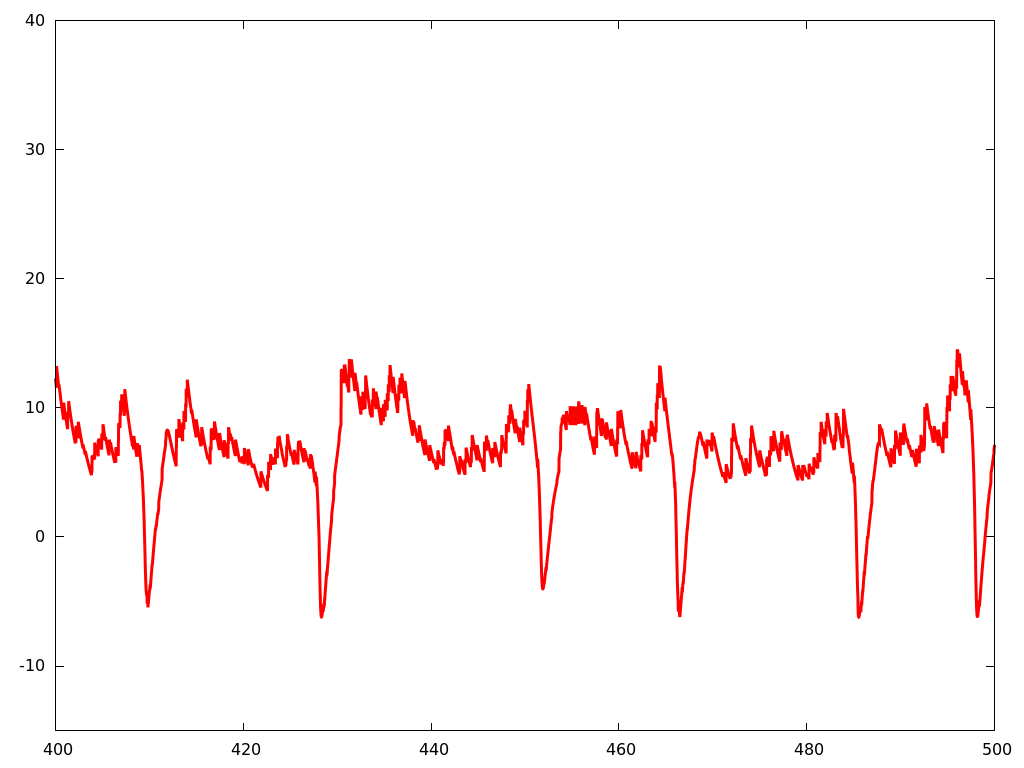}
\includegraphics[height=2.9cm,width=2.9cm]{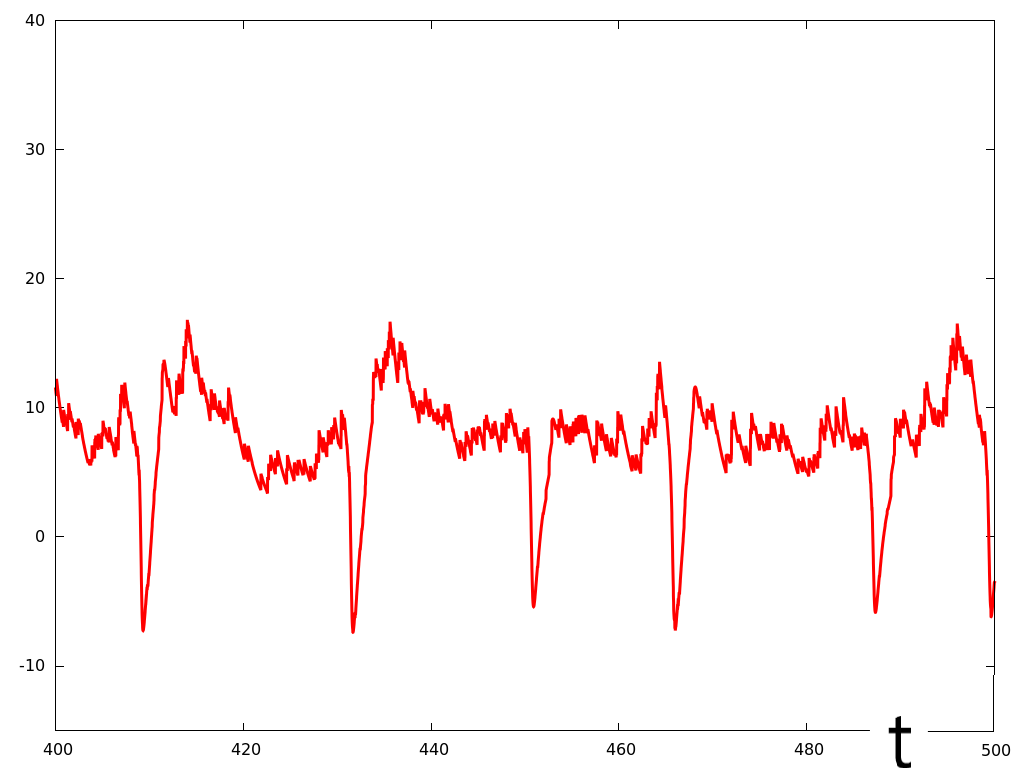}
\includegraphics[height=2.9cm,width=2.9cm]{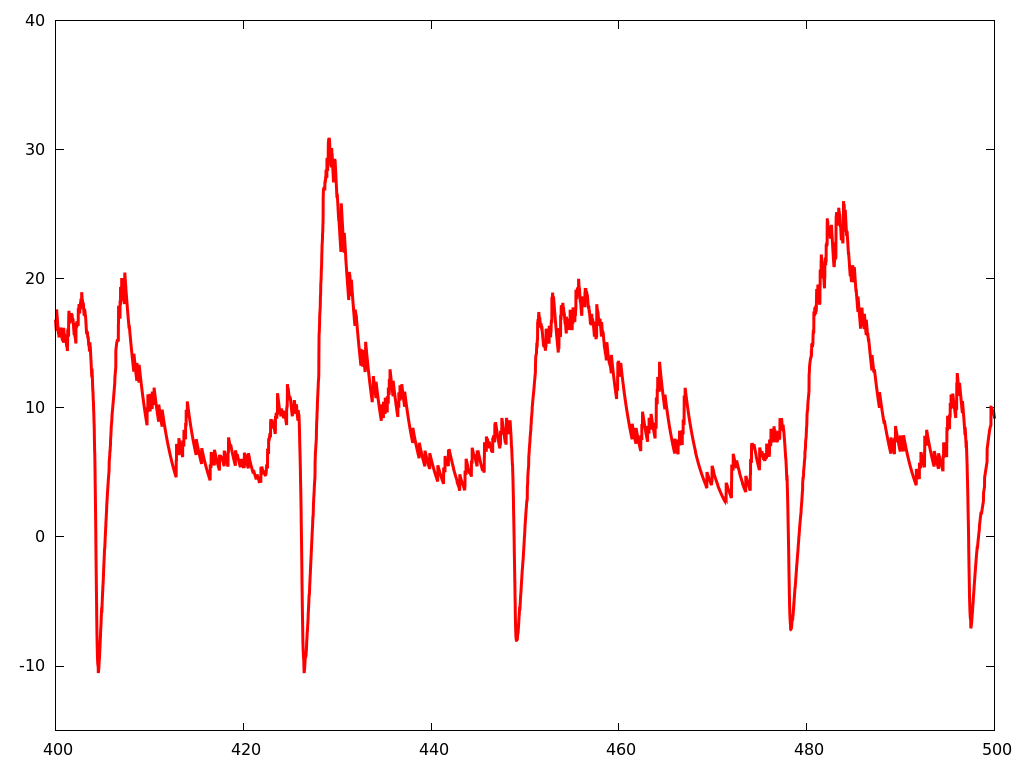}
\includegraphics[height=2.9cm,width=2.9cm]{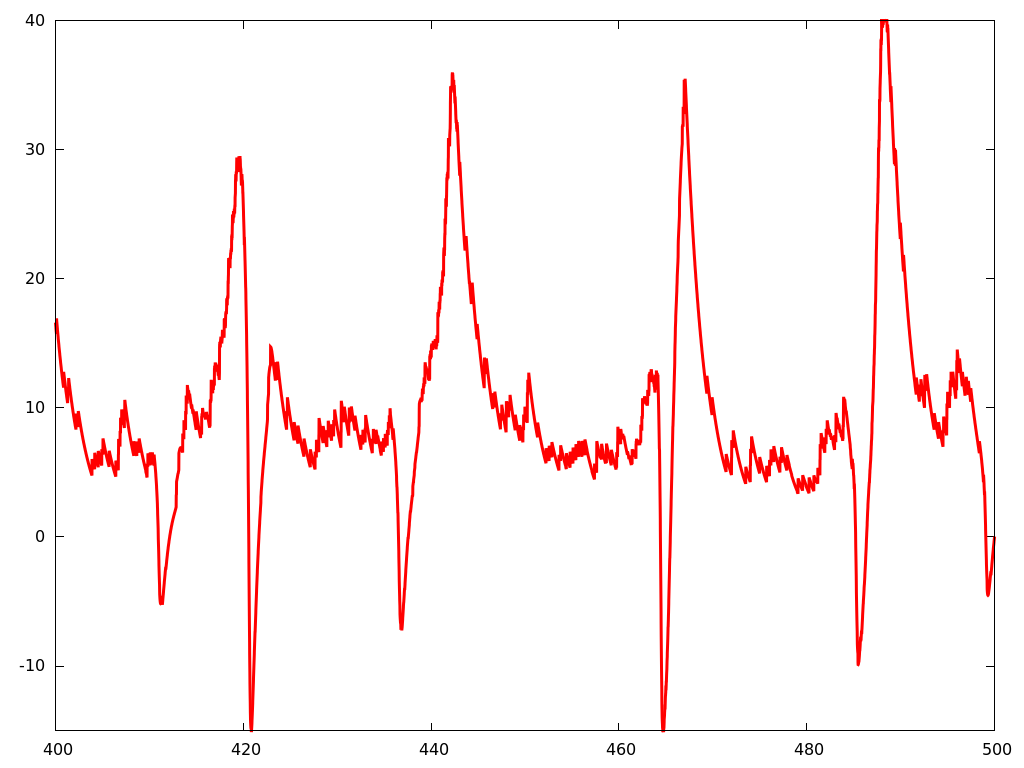}\\
\includegraphics[height=2.9cm,width=2.9cm]{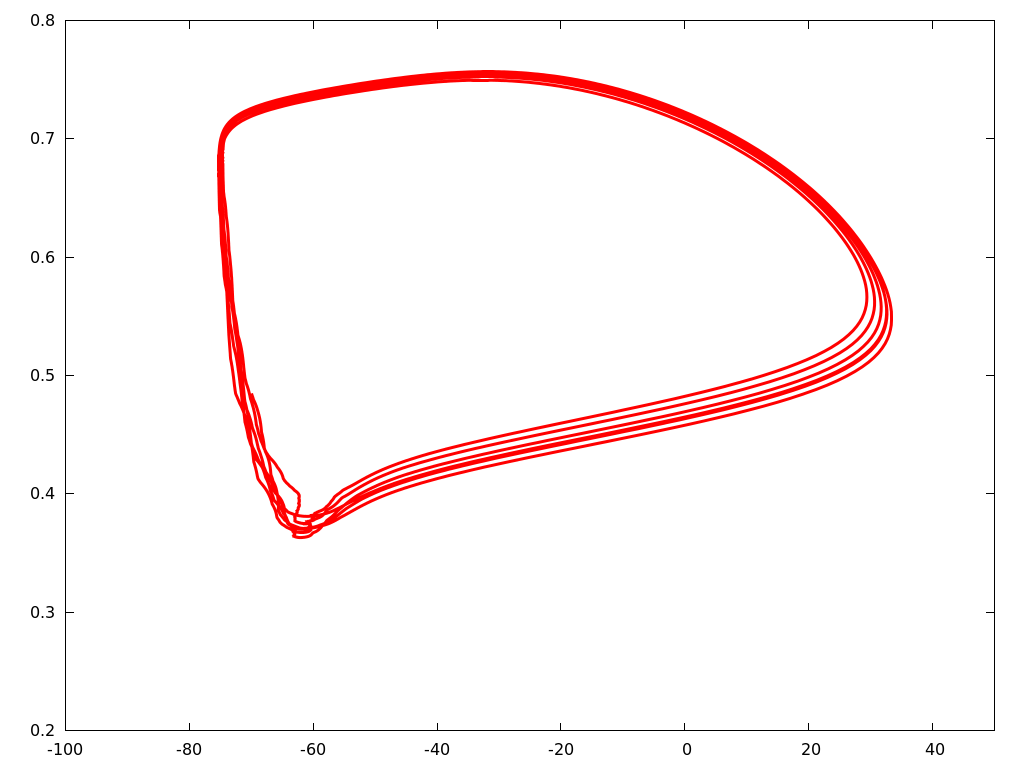}
\includegraphics[height=2.9cm,width=2.9cm]{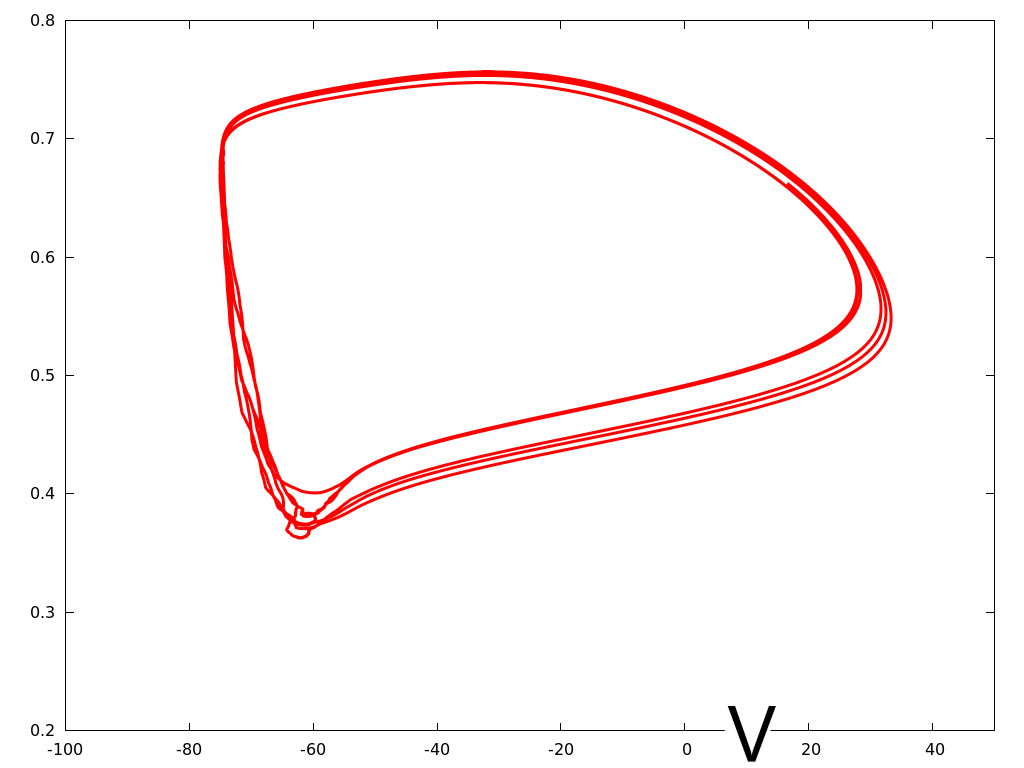}
\includegraphics[height=2.9cm,width=2.9cm]{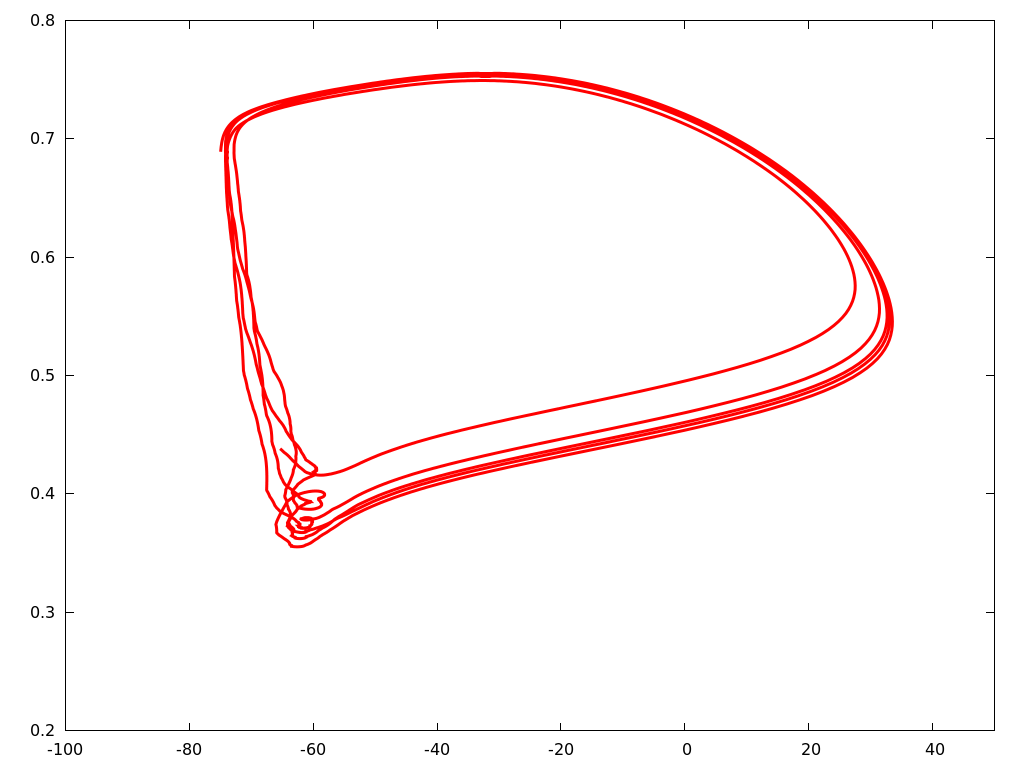}
\includegraphics[height=2.9cm,width=2.9cm]{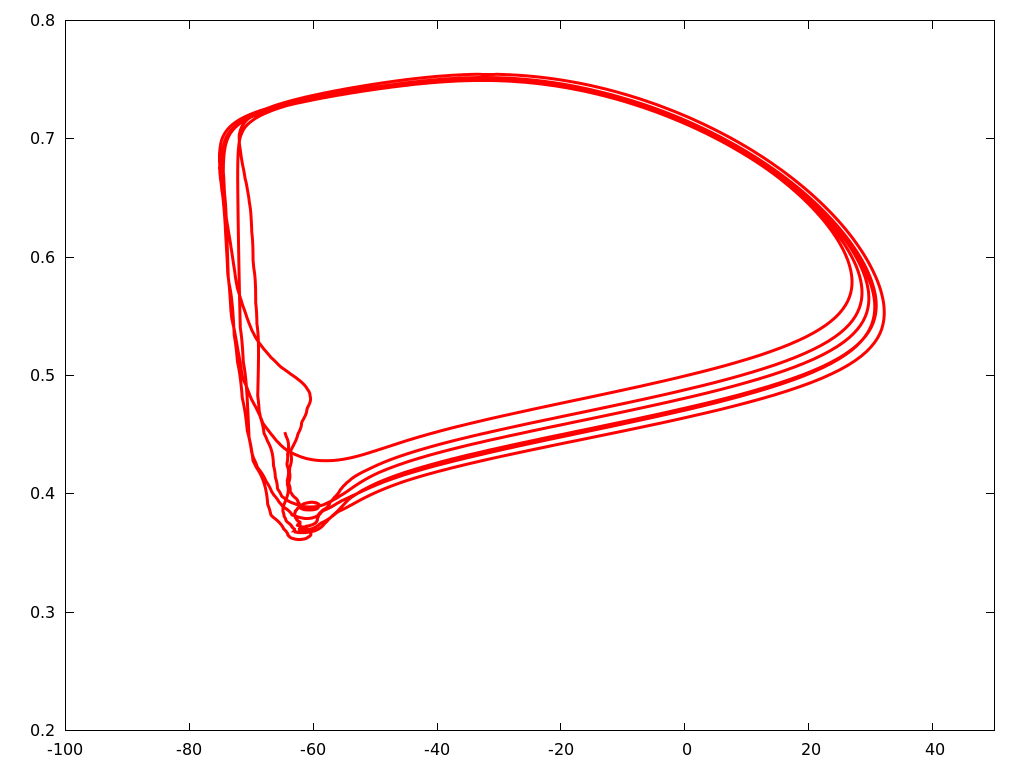}\\
\caption{Simulation of system \eqref{eq:HH-Network}. This figure illustrates a path from random homogeneity to synchronization as the parameter $S^{EE}$ is increased. In this picture, the parameters $SII=SEI=SIE=0.01$ are fixed and each column from left to right corresponds to a specific value of $S^{EE}$. Respectively: $SEE=0.01, 0.017, 0.02$ and $0.03$.  The first row represents the number of $E$ and $I$-spikes occurring during an identified event.  The rows 2 to 6 are analog to those of figure 14, but for a I-neuron.  }
\end{center}
\end{figure}

\begin{figure}
\begin{center}
 \label{fig:PartialSynchro}
\includegraphics[height=6cm,width=7cm]{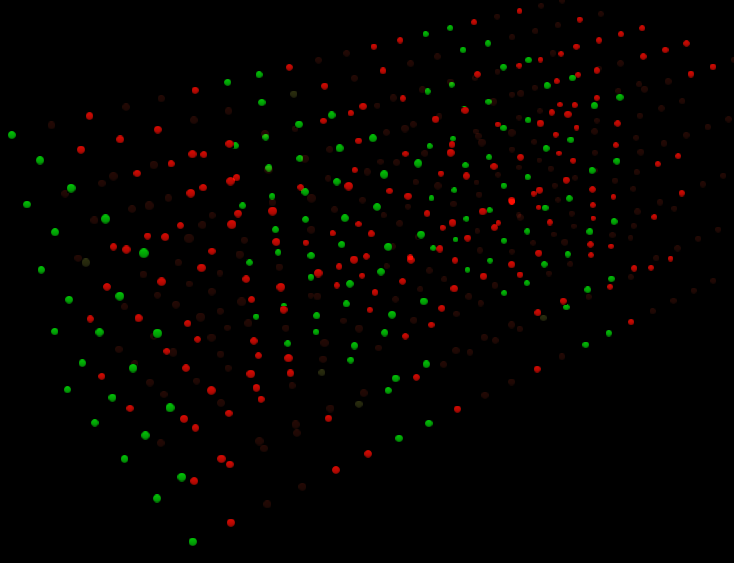}
\caption{Simulation of equation \eqref{eq:HH-Network} for $SII=SEI=SIE=0.01$ and $SEE =0.017$. This figure illustrates partial synchronization. Only some part of the neurons are spiking during the time interval $[425,450]$. The neurons which spike during this interval appear in highlighted color. Red for $E-$neurons and green for $I-$neurons.}
\end{center}
\end{figure}

\begin{table}
\label{ta:SIE}
\begin{tabular}{|c|c|c|c|}
\hline
$S^{IE}$&Ess&Iss\\
\hline
0.005&11.12&44.72\\
\hline
0.01&11.4933&48.48\\
\hline
0.02&11.7867&52.56\\
\hline
0.03&11.7333&60.88\\
\hline
\end{tabular}
\caption{Variation of $S^{IE}$}
\end{table}

\begin{table}
\label{ta:SEI}
\begin{tabular}{|c|c|c|c|}
\hline
$S^{EI}$&Ess&Iss\\
\hline
0.001&13.84&48.64\\
\hline
0.01&11.4933&48.48\\
\hline
0.02&10.2933&47.28\\
\hline
0.03&9.6&43.6\\
\hline
\end{tabular}
\caption{Variation of $S^{EI}$}
\end{table}
  
\begin{table}
\label{ta:SII}
\begin{tabular}{|c|c|c|c|}
\hline
$S^{II}$&Ess&Iss\\
\hline
0.005&11.7067&47.68\\
\hline
0.01&11.4933&48.48\\
\hline
0.02&11.9467&45.84\\
\hline
0.03&11.5733&44.16\\
\hline
\end{tabular}
\caption{Variation of $S^{II}$}
\end{table}

%
%


%
%



\end{document}